%% file: PR.tex
\documentclass[letter, 11pt]{article}
\usepackage[T1]{fontenc}
\usepackage{tgpagella}
\usepackage{graphicx,amsmath,amsfonts,amscd,amssymb,bm,url,color,latexsym}
\usepackage{fullpage}
\usepackage[small,bf]{caption}
\usepackage{subcaption}
\usepackage{bbm}
\usepackage{microtype}
\usepackage{algorithm}
\usepackage{algorithmicx}
\usepackage{algpseudocode} 
\usepackage{verbatim}
\usepackage{framed}
\usepackage{comment}
\usepackage{tikz}
\usepackage{fge}
\usepackage{mathtools}

\allowdisplaybreaks  

\usepackage{hyperref}
\hypersetup{
    colorlinks=true,%
    citecolor=blue,%
    filecolor=blue,%
    linkcolor=blue,%
    urlcolor=blue
}
\usepackage[toc, page]{appendix}

\newtheorem{theorem}{Theorem}[section]
\newtheorem{lemma}[theorem]{Lemma}

\newtheorem{proposition}[theorem]{Proposition}

\newtheorem{remark}[theorem]{Remark}

\renewcommand{\mathbf}{\boldsymbol}

\newcommand{\mb}{\mathbf}
\newcommand{\mc}{\mathcal}

\newcommand{\bb}{\mathbb}
 
\newcommand{\set}[1]{\left\{ #1 \right\}}

\newcommand{\reals}{\bb R}

\newcommand{\eps}{\varepsilon}
\newcommand{\R}{\reals}
\newcommand{\Cp}{\bb C}

\newcommand{\N}{\bb N}

\newcommand{\indicator}[1]{\mathbbm 1_{#1}}

\newcommand{ \Brac }[1]{\left\lbrace #1 \right\rbrace}
\newcommand{ \brac }[1]{\left[ #1 \right]}
\newcommand{ \paren }[1]{ \left( #1 \right) }

\DeclareMathOperator{\dist}{dist}

\DeclareMathOperator{\trace}{tr}

\DeclareMathOperator{\diag}{diag}
\DeclareMathOperator{\sign}{sign}

\DeclareMathOperator{\mini}{minimize}

\DeclareMathOperator{\st}{subject\; to}

\newcommand{\event}{\mc E}

\newcommand{\e}{\mathrm{e}}
\newcommand{\im}{\mathrm{i}}

\newcommand{\wh}{\widehat}
\newcommand{\wt}{\widetilde}
\newcommand{\ol}{\overline}

\newcommand{\norm}[2]{\left\| #1 \right\|_{#2}}
\newcommand{\abs}[1]{\left| #1 \right|}
\newcommand{\innerprod}[2]{\left\langle #1,  #2 \right\rangle}
\newcommand{\prob}[1]{\bb P\left[ #1 \right]}
\newcommand{\expect}[1]{\bb E\left[ #1 \right]}

\newcommand{\js}[1]{#1}

\numberwithin{equation}{section}

\def \endprf{\hfill {\vrule height6pt width6pt depth0pt}\medskip}
\newenvironment{proof}{\noindent {\bf Proof} }{\endprf\par}

\title{A Geometric Analysis of Phase Retrieval}
\author{Ju Sun, Qing Qu, and John Wright \\
\vspace{1mm}
\texttt{\{js4038, qq2105, jw2966\}@columbia.edu} \\
Department of Electrical Engineering, Columbia University, New York, USA
}
\date{February 22, 2016  \quad Revised: \today}

\begin{document}
\maketitle

\vspace{-0.3in}
\begin{abstract}
Can we recover a complex signal from its Fourier magnitudes? More generally, given a set of $m$ measurements, $y_k = \abs{\mb a_k^* \mb x}$ for $k = 1, \dots, m$, is it possible to recover $\mb x \in \Cp^n$ (i.e., length-$n$ complex vector)? This {\em generalized phase retrieval} (GPR) problem is a fundamental task in various disciplines, and has been the subject of much recent investigation. Natural nonconvex heuristics often work remarkably well for GPR in practice, but lack clear theoretical explanations. In this paper, we take a step towards bridging this gap. We prove that when the measurement vectors $\mb a_k$'s are \emph{generic} (i.i.d.\ complex Gaussian) and \emph{numerous} enough ($m \ge C  n \log^3 n$), with high probability, a natural least-squares formulation for GPR has the following benign geometric structure: (1) there are no spurious local minimizers, and all global minimizers are equal to the target signal $\mb x$, up to a global phase; and (2) the objective function has a negative directional curvature around each saddle point. This structure allows a number of iterative optimization methods to efficiently find a global minimizer, without special initialization. To corroborate the claim, we describe and analyze a second-order trust-region algorithm. 
\end{abstract}





\input{sec/intro}

\input{sec/geometry}

\input{sec/finite}

\input{sec/algorithm}

\input{sec/convergence}

\input{sec/exp}\label{sec:experiment}

\input{sec/discussion}\label{sec:discussion}

\input{sec/proof_finite}

\input{sec/proof_algorithm}

\noindent {\bf Acknowledgement.} This work was partially supported by funding from the Gordon and Betty Moore Foundation, the Alfred P.\ Sloan Foundation, and the grants ONR N00014-13-1-0492, NSF CCF 1527809, and NSF IIS 1546411. We thank Nicolas Boumal for helpful discussion related to the Manopt package. We thank Mahdi Soltanolkotabi for pointing us to his early result on the local convexity around the target set for GPR in $\R^n$. We also thank Yonina Eldar, Kishore Jaganathan, Xiaodong Li for helpful feedback on a prior version of this paper. We also thank the anonymous reviewers for their careful reading the paper, and for constructive comments which have helped us to substantially improve the presentation. 

\appendices
\input{sec/app_tools}\label{app:tools}

{\small
\bibliographystyle{amsalpha}
\bibliography{../pr,../ncvx}
}

\end{document}

%% file: sec/intro.tex

\section{Introduction}

\subsection{Generalized Phase Retrieval and a Nonconvex Formulation}

This paper concerns the problem of recovering an $n$-dimensional complex vector $\mb x$ from the magnitudes  $y_k = | \mb a_k^* \mb x|$ of its projections onto a collection of known complex vectors $\mb a_1, \dots, \mb a_m \in \bb C^n$. Obviously, one can only hope to recover $\mb x$ up to a global phase, as $\mb x \e^{\im \phi}$ for all $\phi \in [0, 2\pi)$ gives exactly the same set of measurements. The {\em generalized phase retrieval} problem asks whether it is possible to recover $\mb x$, up to this fundamental ambiguity:
\begin{quote}
\textbf{Generalized Phase Retrieval Problem}: Is it possible to \emph{efficiently} recover an unknown $\mb x$ from $y_k = \abs{\mb a_k^* \mb x}$ ($k = 1, \dots, m$), up to a global phase factor $\e^{\im \phi}$?
\end{quote}
This problem has attracted substantial recent interest, due to its connections to fields such as crystallography, optical imaging and astronomy. In these areas, one often has access only to the Fourier magnitudes of a complex signal $\mb x$, i.e., $\abs{\mc F(\mb x)}$ ~\cite{millane1990phase,robert1993phase, walther1963question, fienup1987phase}. The phase information is hard or infeasible to record due to physical constraints. The problem of recovering the signal $\mb x$ from its Fourier magnitudes $\abs{\mc F(\mb x)}$ is naturally termed (Fourier) phase retrieval (PR). It is easy to see PR as a special version of GPR, with the $\mb a_k$'s the Fourier basis vectors. GPR also sees applications in electron microscopy \cite{jianwei2002high}, diffraction and array imaging \cite{bunk2007diffractive,anwei2011array}, acoustics \cite{balan2006signal,balan2010signal}, quantum mechanics \cite{Corbett2006pauli,reichenbach1965philosophic} and quantum information \cite{heinosaari2013quantum}. We refer the reader to survey papers~\cite{shechtman2015phase,jaganathan2015phase} for accounts of recent developments in the theory, algorithms, and applications of GPR. 

For GPR, heuristic methods based on nonconvex optimization often work surprisingly well in practice (e.g., ~\cite{fienup1982phase,gerchberg1972practical}, and many more cited in~\cite{shechtman2015phase,jaganathan2015phase}). However, investigation into provable recovery methods, particularly based on nonconvex optimization, has started only relatively recently~\cite{netrapalli2013phase,candes2013phase-matrix,candes2013phaselift,candes2014solving,candes2015diffraction,waldspurger2015phase, voroninski2014strong, alexeev2014phase, candes2015phase, chen2015solving, white2015local,zhang2016provable,zhang2016reshaped,wang2016solving,kolte2016phase,gao2016gauss,bendory2016non,waldspurger2016phase}. The surprising effectiveness of nonconvex heuristics on GPR remains largely mysterious. In this paper, we take a step towards bridging this gap. 

We focus on a natural least-squares formulation\footnote{Another least-squares formulation, $\mini_{\mb z}\; \frac{1}{2m} \sum_{k=1}^m (y_k - \abs{\mb a_k^* \mb z})^2$, was first studied in the seminal works~\cite{fienup1982phase,gerchberg1972practical}. An obvious advantage of the $f(\mb z)$ studied here is that it is differentiable \js{in the sense of Wirtinger calculus introduced later}. } -- discussed systematically in~\cite{shechtman2015phase,jaganathan2015phase} and first studied theoretically in \cite{candes2015phase, white2015local},
\begin{align}\label{eqn:finite-f}
	\mini_{\mb z\in \bb C^n} f(\mb z) \doteq \frac{1}{2m} \sum_{k=1}^m \paren{ y_k^2 - \abs{ \mb a_k^* \mb z}^2 }^2.
\end{align}
We assume the $\mb a_k$'s are independent identically distributed (i.i.d.) complex Gaussian:
\begin{align} \label{eqn:complex-gaussian}
	\mb a_k = \frac{1}{\sqrt{2}}\paren{ X_k + \im Y_k}, \; \text{with} \;  X_k,Y_k\sim \mc N(\mb 0,\mb I_n) \; \text{independent}.
\end{align}
$f(\mb z)$ is a fourth-order polynomial in $\mb z$,\footnote{\js{Strictly speaking, $f(\mb z)$ is not a complex polynomial in $\mb z$ over the complex field; complex polynomials are necessarily complex differentiable. However, $f(\mb z)$ is a fourth order real polynomial in real and complex parts of $\mb z$. }} and is nonconvex. A-priori, there is little reason to believe that simple iterative methods can solve this problem without special initialization. Typical local convergence (i.e., convergence to a local minimizer) guarantees in optimization require an initialization near the target minimizer~\cite{bertsekas1999nonlinear}. Moreover, existing results on provable recovery using~\eqref{eqn:finite-f} and related formulations rely on careful initialization in the vicinity of the ground truth~\cite{netrapalli2013phase, candes2015phase, chen2015solving, white2015local,zhang2016provable,zhang2016reshaped,wang2016solving,kolte2016phase,gao2016gauss,bendory2016non,waldspurger2016phase}.

\subsection{A Curious Experiment}
\begin{figure}[t]
\centerline{\includegraphics[width=0.8\linewidth]{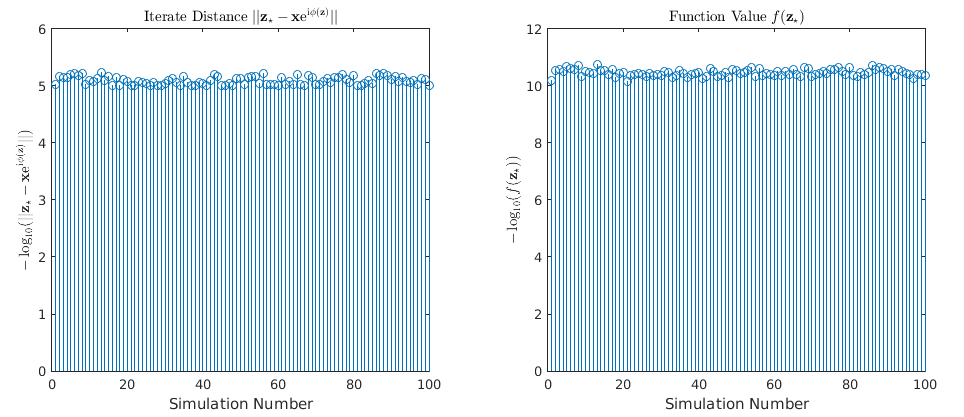}}
\caption{Gradient descent with random initialization seems to always return a global solution for~\eqref{eqn:finite-f}! Here $n = 100$, $m = 5 n\log n$, step size $\mu = 0.05$, and stopping criterion is $\norm{\nabla_{\mb z} f(\mb z) }{} \leq 10^{-5}$. We fix the set of random measurements and the ground-truth signal $\mb x$. The experiments are repeated for $100$ times with independent random initializations. $\mb z_\star$ denotes the final iterate at convergence. (Left) Final distance to the target; (Right) Final function value ($0$ if globally optimized). Both vertical axes are on $-\log_{10}(\cdot)$ scale. } 
\label{fig:grad-conv}
\end{figure}
We apply gradient descent to $f(\mb z)$, starting from a {\em random initialization} $\mb z^{(0)}$: 
\begin{align*}
	\mb z^{(r+1)} = \mb z^{(r)} - \mu \nabla_{\mb z} f(\mb z^{(r)}),
\end{align*}
where the step size $\mu$ is fixed for simplicity\footnote{Mathematically, $f(\mb z)$ is not complex differentiable; here the gradient is defined based on the Wirtinger calculus~\cite{ken2009complex}; see also~\cite{candes2015phase}. This notion of gradient is a natural choice when optimizing real-valued functions of complex variables.}. The result is quite striking (Figure~\ref{fig:grad-conv}): for a fixed problem instance (fixed set of random measurements and fixed target $\mb x$), gradient descent seems to always return a \emph{global minimizer} (i.e., the target $\mb x$ up to a global phase shift), across many independent random initializations! This contrasts with the typical ``mental picture'' of nonconvex objectives as possessing many spurious local minimizers.

\subsection{A Geometric Analysis} \label{sec:intro_geo_analysis}

The numerical surprise described above is not completely isolated.  Simple heuristic methods have been observed to work surprisingly well 
for practical PR~\cite{fienup1982phase,gerchberg1972practical, shechtman2015phase,jaganathan2015phase}. In this paper, we take a step towards explaining this phenomenon. We show that \emph{although the function \eqref{eqn:finite-f} is nonconvex, when $m$ is reasonably large, it actually has benign global geometry which allows it to be globally optimized by efficient iterative methods, regardless of the initialization}.

This geometric structure is evident for \js{real GPR (i.e., real signals with real random measurements)} in $\R^2$. Figure~\ref{fig:geo-2d} plots the function landscape of $f(\mb z)$ for this case with large $m$ (i.e., $\bb E_{\mb a}[f(\mb z)]$ approximately). 
\begin{figure}[t]
\centerline{\includegraphics[width=0.8\linewidth]{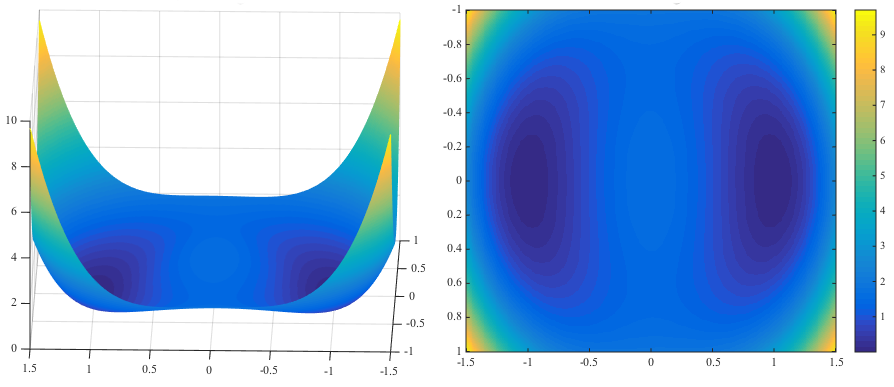}}
\caption{Function landscape of~\eqref{eqn:finite-f} for $\mb x =[1; 0]$ and $m \to \infty$. The only local and also global minimizers are $\pm \mb x$. There are two saddle points near $\pm [0; 1/\sqrt{2}]$, around each there is a negative curvature direction along $\pm \mb x$. (Left) The function graph; (Right) The same function visualized as a color image. The measurement vectors $\mb a_k$'s are taken as i.i.d. standard real Gaussian in this version. } 
\label{fig:geo-2d}
\end{figure}
Notice that (i) the only local minimizers are exactly $\pm \mb x$ -- they are also global minimizers;\footnote{Note that the global sign cannot be recovered. } (ii) there are saddle points (and a local maximizer), but around them there is a negative curvature in the $\pm \mb x$ direction. Intuitively, any algorithm that can successfully escape from this kind of saddle point (and local maximizer) can in fact find a global minimizer, i.e., recover the target signal $\mb x$.  

We prove that an analogous geometric structure exists, with high probability (w.h.p.)\footnote{The probability is with respect to drawing of $\mb a_k$'s.}, for GPR in $\Cp^n$, when $m$ is reasonably large (Theorem~\ref{thm:finite-landscape}). In particular, we show that when $m \ge C n \log^3 n$, w.h.p., (i) the only local and also global minimizers to~\eqref{eqn:finite-f} are the target $\mb x \e^{\im \phi}$ for $\phi \in [0, 2\pi)$; (ii) at any point in $\bb C^n$, either the gradient is large, or the curvature is negative in a certain direction, or it is near a minimizer. Moreover, in the vicinity of the minimizers, on the orthogonal complement of a single flat direction (which occurs because $f(\mb z e^{\im \phi}) = f(\mb z)$ for every $\mb z$, $\phi$), the objective function is strongly convex (\js{a weaker version of this local restricted strong convexity was first established in~\cite{candes2015phase}; see also~\cite{white2015local}}). 

Because of this global geometry, a wide range of efficient iterative methods can obtain a global minimizer to $f(\mb z)$, regardless of initialization. Examples include the noisy gradient and stochastic gradient methods~\cite{ge2015escaping} (see also~\cite{lee2016gradient,panageas2016gradient}), curvilinear search~\cite{goldfarb1980curvilinear} and trust-region methods \cite{conn2000trust,nesterov2006cubic,sun2015nonconvex}. The key property that the methods must possess is the ability to escape saddle points at which the Hessian has a strictly negative eigenvalue\footnote{Such saddle points are called ridable saddles~\cite{sun2015nonconvex} or strict saddles~\cite{ge2015escaping}; see~\cite{anandkumar2016efficient} for computational methods for escaping from higher-order saddles also. }. 

We corroborate this claim by developing a second-order trust-region algorithm for this problem, and prove that (Theorem~\ref{thm:TRM-conv}) (i) from any initialization, it efficiently obtains a close approximation (i.e., up to numerical precision) of the target $\mb x$ (up to a global phase) and (ii) it exhibits quadratic convergence in the vicinity of the global minimizers. 

In sum, our geometrical analysis produces the following result. \vspace{.1in}

\noindent{\bf{Informal Statement of Our Main Results; See Theorem~\ref{thm:finite-landscape} and Theorem~\ref{thm:TRM-conv} and Remark~\ref{rmk:complexity}.}}{\em{
When $m \ge Cn \log^3 n$, with probability at least $1 - cm^{-1}$, the function $f(\mb z)$ has no spurious local minimizers. The only global minimizers are the target $\mb x$ and its equivalent copies, and at all saddle points the function has directional negative curvature. Moreover, with at least the same probability, the trust-region method with properly set step size parameter find a global minimizer of $f(\mb z)$ in polynomial time, from an arbitrary initialization in the zero-centered complex ball with radius $R_0 \doteq 3(\tfrac{1}{m}\sum_{k=1}^m y_k^2)^{1/2}$. Here $C$ and $c$ are absolute positive constants. }} 
\vspace{.05in}

The choice of $R_0$ above allows us to state a result with a concise bound on the number of iterations required to converge. However, under our probability model, w.h.p., the trust- region method succeeds from any initialization. There are two caveats to this claim. First, one must choose the parameters of the method appropriately. Second, the number of iterations depends on how far away from the truth the method starts. 

Our results asserts that when the $\mb a_k$'s are \emph{numerous} and \emph{generic} enough, GPR can be solved \js{in polynomial time by optimizing the nonconvex formulation~\eqref{eqn:finite-f}.}
Similar conclusions have been obtained in~\cite{netrapalli2013phase, candes2015phase, chen2015solving, white2015local,zhang2016provable,zhang2016reshaped,wang2016solving,kolte2016phase,gao2016gauss,bendory2016non,waldspurger2016phase}, \js{also based on nonconvex optimization}. One salient feature of our result is that the optimization method is ``initialization free'' - any initialization in the prescribed ball works. \js{This follows directly from the benign global geometry of $f(\mb z)$. In contrast, all prior nonconvex methods require careful initializations that are already near the unknown target $\mb x \e^{\im \phi}$, based on characterization of only local geometry. We believe our global geometrical analysis sheds light on mechanism of the above numerical surprise.}  

\js{The second-order trust-region method, albeit polynomial-time, may not be the most practical algorithm for solving GPR. Deriving the most practical algorithms is not the focus of this paper. We mentioned above that any iterative method with saddle-escaping capability can be deployed to solve the nonconvex formulation; our geometrical analysis constitutes a solid basis for developing and analyzing much more practical algorithms for GPR. }

\subsection{Prior Arts and Connections}
The survey papers~\cite{shechtman2015phase,jaganathan2015phase} provide comprehensive accounts of recent progress on GPR. In this section, we focus on provable efficient (particularly, nonconvex) methods for GPR, and draw connections to other work on provable nonconvex heuristics for practical problems. 

\paragraph{Provable methods for GPR.} 
Although heuristic methods for GPR have been used effectively in practice~\cite{gerchberg1972practical,fienup1982phase,shechtman2015phase,jaganathan2015phase}, only recently have researchers begun to develop methods with provable performance guarantees. The first results of this nature were obtained using semidefinite programming (SDP) relaxations~\cite{candes2013phase-matrix,candes2013phaselift,candes2014solving,candes2015diffraction,waldspurger2015phase,voroninski2014strong}. While this represented a substantial advance in theory, the computational complexity of semidefinite programming limits the practicality of this approach.\footnote{Another line of research~\cite{balan2006signal, balan2009painless, alexeev2014phase} seeks to co-design the measurements and recovery algorithms based on frame- or graph-theoretic tools. While revising this work, new convex relaxations based on second-order cone programming have been proposed~\cite{goldstein2016phasemax,bahmani2016phase,hand2016elementary,hand2016compressed}. } 

Recently, several provable {\em nonconvex} methods have been proposed for GPR. \cite{netrapalli2013phase} augmented the seminal error-reduction method~\cite{gerchberg1972practical} with spectral initialization and resampling to obtain the first provable nonconvex method for GPR. \cite{candes2015phase} studied the nonconvex formulation~\eqref{eqn:finite-f} under the same hypotheses as this paper, and showed that a combination of spectral initialization and local gradient descent recovers the true signal with near-optimal sample complexity. \cite{chen2015solving} worked with a different nonconvex formulation, and refined the spectral initialization and the local gradient descent with a step-adaptive truncation. With the modifications, they reduced the sample requirement to the optimal order.\footnote{In addition, \cite{chen2015solving} shows that the measurements can be non-adaptive, in the sense that a single, randomly chosen collection of vectors $\mb a_i$ can simultaneously recover every $\mb x \in \Cp^n$. Results in \cite{netrapalli2013phase, candes2015phase} and this paper pertain only to adaptive measurements that recover any fixed signal $\mb x$ with high probability.} \js{More recent work in this line~\cite{zhang2016provable,zhang2016reshaped,wang2016solving,kolte2016phase,gao2016gauss,bendory2016non,waldspurger2016phase} concerns error stability, alternative formulations, algorithms, and measurement models.} Compared to the SDP-based methods, these methods are more scalable and closer to methods used in practice. \js{All these analyses are based on local geometry in nature, and hence depend on the spectral initializer being sufficiently close to the target set}. In contrast, we explicitly characterize the global function landscape of~\eqref{eqn:finite-f}. Its benign \js{global} geometric structure allows several algorithmic choices (see Section~\ref{sec:intro_geo_analysis}) that need \emph{no special initialization} \js{and scale much better than the convex approaches}. 

\js{Near the target set (i.e., $\mc R_3$ in Theorem~\ref{thm:finite-landscape}), \cite{candes2015phase,chen2015solving} established a local curvature property that is strictly weaker than our restricted strong convexity result. The former is sufficient for obtaining convergence results for first-order methods, while the latter is necessary for establishing convergence results for second-order method (see our detailed comments in Section~\ref{sec:geometry-key}). Besides these, \cite{soltanolkotabi2014algorithms} and \cite{white2015local} also explicitly established local strong convexity near the target set for real GPR in $\R^n$; the Hessian-form characterization presented in~\cite{white2015local} is real-version counterpart to ours here. }

\paragraph{(Global) Geometric analysis of other nonconvex problems.} 
The approach taken here is similar in spirit to our recent geometric analysis of a nonconvex formulation for complete dictionary learning~\cite{sun2015complete}. For that problem, we also identified a similar geometric structure that allows efficient global optimization without special initialization. There, by analyzing the geometry of a nonconvex formulation, we derived a provable efficient algorithm for recovering square invertible dictionaries when the coefficient matrix has a constant fraction of nonzero entries. Previous results required the dictionary matrix to have far fewer nonzero entries. \cite{sun2015nonconvex} provides a high-level overview of the common geometric structure that arises in dictionary learning, GPR and several other problems. This approach has also been applied to other problems~\cite{ge2015escaping,bandeira2016low,boumal2016non,soudry2016no,kawaguchi2016deep,bhojanapalli2016global,ge2016matrix,park2016non}. Despite these similarities, GPR raises several novel technical challenges: the objective is heavy-tailed, and minimizing the number of measurements is important\footnote{\js{The same challenge is also faced by~\cite{candes2015phase,chen2015solving}}.}. 

Our work sits amid the recent surge of work on provable nonconvex heuristics for practical problems. Besides GPR studied here, this line of work includes low-rank matrix recovery~\cite{keshavan2010matrix, jain2013low,hardt2014understanding,hardt2014fast,netrapalli2014non,jain2014fast,sun2014guaranteed, wei2015guarantees,sa2015global,zheng2015convergent,tu2015low,chen2015fast}, tensor recovery~\cite{jain2014provable,anandkumar2014analyzing,anandkumar2014guaranteed,anandkumar2015tensor,ge2015escaping}, structured element pursuit~\cite{qu2014finding,hopkins2015speeding}, dictionary learning~\cite{agarwal2013learning,arora2013new,agarwal2013exact,arora2014more,arora2015simple,sun2015complete}, mixed regression~\cite{yi2013alternating,sedghi2014provable}, blind deconvolution~\cite{lee2013near, lee2015rip, lee2015blind}, super resolution~\cite{eftekhari2015greed}, phase synchronization~\cite{boumal2016nonconvex}, numerical linear algebra~\cite{jain2015computing}, and so forth. Most of the methods adopt the strategy of initialization plus local refinement we alluded to above. In contrast, our global geometric analysis allows flexible algorithm design (i.e., separation of geometry and algorithms) and gives some clues as to the behavior of nonconvex heuristics used in practice, which often succeed without clever initialization. 

\paragraph{Recovering low-rank positive semidefinite matrices. }

The phase retrieval problem has a natural generalization to recovering low-rank positive semidefinite matrices. 
Consider the problem of recovering an unknown rank-$r$ matrix $\mb M \succeq \mb 0$ in $\R^{n \times n}$ from linear measurement of the form $z_k = \trace(\mb A_k \mb M)$ with symmetric $\mb A_k$ for $k = 1, \dots, m$. One can solve the problem by considering the ``factorized'' version: recovering $\mb X \in \R^{n \times r}$ (up to right invertible transform) from measurements $z_k = \trace(\mb X^* \mb A_k \mb X)$. This is a natural generalization of GPR, as one can write the GPR measurements as $y_k^2 = \abs{\mb a_k^* \mb x}^2 = \mb x^* (\mb a_k \mb a_k^*) \mb x$. This generalization and related problems have recently been studied in~\cite{sa2015global,zheng2015convergent,tu2015low,chen2015fast,bhojanapalli2016global}.

\subsection{Notations, Organization, and Reproducible Research}
\paragraph{Basic notations and facts.}
Throughout the paper, we define complex inner product as: $\innerprod{\mb a}{\mb b} \doteq \mb a^* \mb b$ for any $\mb a, \mb b \in \Cp^n$. We use $\bb{CS}^{n-1}$ for the complex unit sphere in $\Cp^n$. $\bb{CS}^{n-1}(\lambda)$ with $\lambda >0$ denotes the centered complex sphere with radius $\lambda$ in $\Cp^n$. Similarly, we use $\bb{CB}^n(\lambda)$ to denote the centered complex ball of radius $\lambda$. We use $\mc{CN}(k)$ for a standard complex Gaussian vector of length $k$ defined in~\eqref{eqn:complex-gaussian}. We reserve $C$ and $c$, and their indexed versions to denote absolute constants. Their value vary with the context. 

Let $\Re\paren{\mb z} \in \R^n$  and $\Im(\mb z) \in \R^n$ denote the real and imaginary part of a complex vector $\mb z \in \Cp^n$. \js{We will often use the canonical identification of $\Cp^n$ and $\R^{2n}$, which assign $\mb z \in \Cp^n$ to $[\Re\paren{\mb z}; \Im\paren{\mb z}] \in \R^{2n}$. This is so natural that we will not explicitly state the identification when no confusion is caused. We say two complex vectors are orthogonal in the geometric (real) sense if they are orthogonal after the canonical identification\footnote{Two complex vectors $\mb w, \mb v$ are orthogonal in complex sense if $\mb w^* \mb v = 0$.}. } It is easy to see that two complex vectors $\mb a$ and $\mb b$ are orthogonal in the geometric (real) sense if and only if $\Re(\mb w^* \mb z) = 0$. 

For any $\mb z$, obviously $f(\mb z) = f(\mb z \e^{\im \phi})$ for all $\phi$, and the set $\set{\mb z \e^{\im \phi}: \phi \in [0, 2\pi)}$ forms a one-dimensional (in the real sense) circle in $\Cp^n$. Throughout the paper, we reserve $\mb x$ for the unknown target signal, and define the target set as $\mc X \doteq \Brac{ \mb x \e^{\im \phi }:  \phi \in [0, 2\pi ) } $. Moreover, we define
\begin{align}\label{eqn:basic-def}
   \phi(\mb z) \doteq \mathop{\arg\min}_{\phi \in [0,2\pi) } \norm{\mb z - \mb x \e^{\im \phi} }{},\quad \mb h(\mb z) \doteq \mb z - \mb x \e^{\im  \phi\paren{\mb z}  }, \quad \dist\paren{\mb z, \mc X} \doteq \norm{ \mb h(\mb z) }{}.
\end{align}
for any $\mb z \in \Cp^n$. It is not difficult to see that 
$
\mb z^*\mb x \e^{\im \phi(\mb z)}  = \abs{\mb x^* \mb z}
$. Moreover, $\mb z_T \doteq \im \mb z/\norm{\mb z}{}$ and $-\mb z_T$ are the unit vectors tangent to the circle $\set{\mb z \e^{\im \phi}: \phi \in [0, 2\pi)}$ at point $\mb z$.

\paragraph{Wirtinger calculus.}
Consider a real-valued function $g(\mb z): \bb C^n \mapsto \bb R$. Unless $g$ is constant, it is not complex differentiable. However, if one identifies $\Cp^n$ with $\R^{2n}$ and treats $g$ as a function in the real domain, $g$ may still be differentiable in the real sense. Doing calculus for $g$ directly in the real domain tends to produce cumbersome expressions. A more elegant way is adopting the Wirtinger calculus, which can be thought of a neat way of organizing the real partial derivatives. Here we only provide a minimal exposition of Wirtinger calculus; similar exposition is also given in~\cite{candes2015phase}. A systematic development with emphasis on applications in optimization is provided in the article~\cite{ken2009complex}. 

Let $\mb z = \mb x + \im \mb y$ where $\mb x = \Re(\mb z)$ and $\mb y = \Im(\mb z)$. For a complex-valued function $g(\mb z) = u(\mb x, \mb y) + \im v(\mb x, \mb y)$, the Wirtinger derivative is well defined so long as the real-valued functions $u$ and $v$ are differentiable with respect to (w.r.t.) $\mb x$ and $\mb y$. Under these conditions, the Wirtinger derivatives can be defined {\em formally} as 
\begin{align*}
	\frac{\partial g}{\partial \mb z} 
	& \doteq \left. \frac{\partial g(\mb z, \ol{\mb z}) }{\partial \mb z}\right |_{\ol{\mb z} \text{ constant} } = 
	\left.
	\brac{
	\frac{\partial g(\mb z, \ol{\mb z}) }{\partial z_1}, \dots, \frac{\partial g(\mb z, \ol{\mb z}) }{\partial z_n}
	}
	\right|_{\ol{\mb z} \text{ constant} }\\
	\frac{\partial g}{\partial \ol{\mb z}} 
	& \doteq \left. \frac{\partial g(\mb z, \ol{\mb z}) }{\partial \ol{\mb z}}\right |_{\mb z \text{ constant} } = 
	\left.
	\brac{
	\frac{\partial g(\mb z, \ol{\mb z}) }{\partial \ol{z_1}}, \dots, \frac{\partial g(\mb z, \ol{\mb z}) }{\partial \ol{z_n}}
	}
	\right|_{\mb z \text{ constant} }. 
\end{align*}
The notation above should only be taken at a formal level. Basically it says when evaluating $\partial g/\partial \mb z$, one just treats $\ol{\mb z}$ as if it was a constant, and vise versa. To evaluate the individual partial derivatives, such as $\frac{\partial g(\mb z, \ol{\mb z}) }{\partial z_i}$, all the usual rules of calculus apply.\footnote{The precise definition is as follows: write $\mb z = \mb u + \im \mb v$. Then $\frac{\partial g}{\partial \mb z} \doteq \tfrac{1}{2} \left( \frac{\partial g}{\partial \mb u} - \im \frac{\partial g}{\partial \mb v} \right)$. Similarly, $\frac{\partial g}{\partial \bar{\mb z}} \doteq \tfrac{1}{2} \left( \frac{\partial g}{\partial \mb u} + \im \frac{\partial g}{\partial \mb v} \right)$.}

Note that above the partial derivatives $\frac{\partial g}{\partial \mb z}$ and $\frac{\partial g}{\partial \ol{\mb z}}$ are row vectors. The Wirtinger gradient and Hessian are defined as
\begin{align}
\nabla g(\mb z) = \brac{\frac{\partial g}{\partial \mb z}, \frac{\partial g}{\partial \ol{\mb z}}}^*
\quad 
\nabla^2 g(\mb z) = 
\begin{bmatrix}
\frac{\partial}{\partial \mb z} \paren{\frac{\partial g}{\partial \mb z}}^* & \frac{\partial}{\partial \ol{\mb z}} \paren{\frac{\partial g}{\partial \mb z}}^* \\
\frac{\partial}{\partial \mb z} \paren{\frac{\partial g}{\partial \ol{\mb z}}}^* & \frac{\partial}{\partial \ol{\mb z}} \paren{\frac{\partial g}{\partial \ol{\mb z}}}^*
\end{bmatrix}, \label{eqn:wirtinger}
\end{align}
where we sometimes write $\nabla_{\mb z} g \doteq \paren{\frac{\partial g}{\partial \mb z}}^*$ and naturally $\nabla_{\ol{\mb z}} g \doteq \paren{\frac{\partial g}{\partial \ol{\mb z}}}^*$. With gradient and Hessian, the second-order Taylor expansion of $g(\mb z)$ at a point $\mb z_0$ is defined as 
\begin{align*}
\wh{g}(\mb \delta; \mb z_0) = g(\mb z_0) + 
\paren{\nabla g(\mb z_0)}^*
\begin{bmatrix}
\mb \delta \\
\ol{\mb \delta} 
\end{bmatrix}
+ \frac{1}{2}
\begin{bmatrix}
\mb \delta \\
\ol{\mb \delta} 
\end{bmatrix}^* 
\nabla^2 g(\mb z_0)
\begin{bmatrix}
\mb \delta \\
\ol{\mb \delta} 
\end{bmatrix}. 
\end{align*}
For numerical optimization, we are most interested in real-valued $g$. A real-valued $g$ is stationary at a point $\mb z$ if and only if 
\begin{align*}
\nabla_{\mb z} g(\mb z) = \mb 0. 
\end{align*}
This is equivalent to the condition $\nabla_{\ol{\mb z}} g = \mb 0$, as $\nabla_{\mb z} g = \ol{\nabla_{\ol{\mb z}} g}$ when $g$ is real-valued. The curvature of $g$ at a stationary point $\mb z$ is dictated by the Wirtinger Hessian $\nabla^2 g(\mb z)$. An important technical point is that the Hessian quadratic form involves left and right multiplication with a $2n$-dimensional vector consisting of a conjugate pair $(\mb \delta, \bar{\mb \delta})$. 
 
For our particular function $f(\mb z): \bb C^n \mapsto \bb R$ defined in \eqref{eqn:finite-f}, direct calculation gives 
\begin{align}
\nabla f(\mb z) & = \frac{1}{m}\sum_{k=1}^m
\begin{bmatrix}
 \paren{\abs{\mb a_k^* \mb z}^2 - y_k^2} \paren{\mb a_k \mb a_k^*} \mb z \\ \paren{\abs{\mb a_k^* \mb z}^2 - y_k^2} \paren{\mb a_k \mb a_k^*}^\top \ol{\mb z}
\end{bmatrix}, \label{eqn:grad-finite}  \\
\nabla^2 f(\mb z) & = \frac{1}{m}\sum_{k=1}^m
\begin{bmatrix}
\paren{2\abs{\mb a_k^* \mb z}^2 - y_k^2} \mb a_k \mb a_k^* & \paren{\mb a_k^* \mb z}^2 \mb a_k \mb a_k^\top \\
\paren{\mb z^* \mb a_k}^2 \ol{\mb a_k} \mb a_k^* & \paren{2\abs{\mb a_k^* \mb z}^2 - y_k^2} \ol{\mb a_k} \mb a_k^\top 
\end{bmatrix}. \label{eqn:hessian-finite}
\end{align}
Following the above notation, we write $\nabla_{\mb z} f(\mb z)$ and $\nabla_{\ol{\mb z}} f(\mb z)$ for denoting the first and second half of $\nabla f(\mb z)$, respectively.

\paragraph{Organization.}
The remainder of this paper is organized as follows. In Section~\ref{sec:func-geometry}, we provide a quantitative characterization of the global geometry for GPR and highlight main technical challenges in establishing the results. Based on this characterization, in Section~\ref{sec:algorithm} we present a modified trust-region method for solving GPR from an arbitrary initialization, which leads to our main computational guarantee. In Section~\ref{sec:experiment} we study the empirical performance of our method for GPR. Section \ref{sec:discussion} concludes the main body with a discussion of open problems. Section~\ref{app:finite} and Section~\ref{app:algorithm} collect the detailed proofs to technical results for the geometric analysis and algorithmic analysis, respectively. 

\paragraph{Reproducible research.}
The code to reproduce all the figures and the experimental results can be found online: 
\begin{quote}
\centering
\url{https://github.com/sunju/pr_plain} . 
\end{quote}

%% file: sec/geometry.tex

\section{The Geometry of the Objective Function}\label{sec:func-geometry}
The low-dimensional example described in the introduction (Figure~\ref{fig:geo-2d}) provides some clues about the high-dimensional geometry of the objective function $f(\mb z)$. Its properties can be seen most clearly through the population objective function $\bb E_{\mb a}[ f(\mb z)]$, which can be thought of as a ``large sample'' version in which $m \to \infty$. We characterize this large-sample geometry in Section \ref{sec:geometry-asymptotic}. In Section \ref{sec:geometry-finite}, we show that the most important characteristics of this large-sample geometry are present even when the number of observations $m$ is close to the number of degrees of freedom $n$ in the target $\mb x$. Section \ref{sec:geometry-key} describes several technical problems that arise in the finite sample analysis, and states a number of key intermediate results, which are proved in Section~\ref{app:finite}. 

\subsection{A Glimpse of the Asymptotic Function Landscape}\label{sec:geometry-asymptotic}

To characterize the geometry of $\bb E_{\mb a}[f(\mb z)]$ (written as $\expect{f}$ henceforth), we simply calculate the expectation of the first and second derivatives of $f$ at each point $\mb z \in \bb C^n$. We characterize the location of the critical points, and use second derivative information to characterize their signatures. An important conclusion is that every local minimum of $\expect{f}$ is of the form $\mb x e^{\im \phi}$, and that all other critical points have a direction of strict negative curvature:

\begin{theorem}
	When $\mb x \not = \mb 0$, the only critical points of $\expect{f}$ are $\mb 0$, $\mc X$ and $\mc S \doteq\Brac{\mb z: \mb x^*\mb z = 0,\; \norm{\mb z}{} = \norm{\mb x}{}/\sqrt{2} }$, which are the local maximizer, the set of local/global minimizers, and the set of saddle points, respectively.	Moreover, the saddle points and local maximizer have negative curvature in the $\mb x \e^{\im \phi(\mb z)}$ direction. 
\end{theorem}

\begin{proof}
We show the statement by partitioning the space $\bb C^n$ into several regions and analyzing each region individually using the expected gradient and Hessian. These are calculated in Lemma \ref{lem:expect-func}, and reproduced below:	
\begin{align}
		\expect{f} &= \norm{\mb x}{}^4 + \norm{\mb z}{}^4 - \norm{\mb x}{}^2 \norm{\mb z}{}^2 - \abs{\mb x^* \mb z}^2,  \\
		\nabla \expect{f} &= \begin{bmatrix}
			\nabla_{\mb z} \expect{f} \\ \nabla_{\ol{\mb z}} \expect{f}
		\end{bmatrix} = \begin{bmatrix}
			\paren{2\norm{\mb z}{}^2 \mb I - \norm{\mb x}{}^2 \mb I - \mb x \mb x^*} \mb z \\
			\paren{2\norm{\mb z}{}^2\mb I  - \norm{\mb x}{}^2 \mb I - \mb x \mb x^*} \ol{\mb z} 
		\end{bmatrix}, \\
		\nabla^2 \expect{f} &= 
\begin{bmatrix}
2 \mb z \mb z^* - \mb x \mb x^* + \paren{2\norm{\mb z}{}^2 - \norm{\mb x}{}^2} \mb I & 2\mb z \mb z^\top  \\
2 \ol{\mb z} \mb z^* & 2 \ol{\mb z} \mb z^\top - \ol{\mb x} \mb x^\top + \paren{2\norm{\mb z}{}^2 - \norm{\mb x}{}^2} \mb I
\end{bmatrix}.
	\end{align}
 Based on this, we observe:
\begin{itemize}
\item $\mb z = \mb 0$ is a critical point, and the Hessian 
\begin{align*}
\nabla^2 \bb E\brac{ f(\mb 0)} = 
\diag\paren{
-\mb x \mb x^* - \norm{\mb x}{}^2 \mb I,  -\ol{\mb x} \mb x^\top - \norm{\mb x}{}^2 \mb I} \prec \mb 0. 
\end{align*}
Hence, $\mb z = \mb 0$ is a local maximizer. 
\item In the region $\set{\mb z: 0 < \norm{\mb z}{}^2 < \tfrac{1}{2} \norm{\mb x}{}^2 }$, we have
\begin{align*}
	\begin{bmatrix}
		\mb z \\ \ol{\mb z}
	\end{bmatrix}^* \nabla \bb E\brac{f} = 2 \paren{ 2\norm{\mb z}{}^2 - \norm{\mb x}{}^2  } \norm{\mb z}{}^2 - 2\abs{\mb x^*\mb z}^2<0. 
\end{align*}
So there is no critical point in this region. 
\item When $\norm{\mb z}{}^2 = \tfrac{1}{2} \norm{\mb x}{}^2$, the gradient is $ \nabla_{\mb z} \bb E\brac{f} =  -\mb x \mb x^* \mb z$. The gradient vanishes whenever $\mb z \in \mathrm{null}\paren{\mb x \mb x^*}$, which is true if and only if $\mb x^* \mb z = 0$. Thus, we can see that any $\mb z \in \mc S$ is a critical point. Moreover, for any $\mb z\in \mc S$, 
\begin{align*}
\begin{bmatrix}
\mb x \e^{\im \phi(\mb z)} \\
\ol{\mb x \e^{\im \phi(\mb z)} }  
\end{bmatrix}^* 
\nabla^2 \expect{f}
\begin{bmatrix}
\mb x \e^{\im \phi(\mb z)} \\ 
\ol{\mb x \e^{\im \phi(\mb z)} } 
\end{bmatrix}
= -2\norm{\mb x}{}^4. 
\end{align*}
Similarly, one can show that in $\mb z$ direction there is positive curvature. Hence, every $\mb z \in \mc S$ is a saddle point.
\item In the region $\set{\mb z: \tfrac{1}{2} \norm{\mb x}{}^2 < \norm{\mb z}{}^2 < \norm{\mb x}{}^2}$, any potential critical point must satisfy 
\begin{align*}
\left ( 2\norm{\mb z}{}^2 - \norm{\mb x}{}^2 \right) \mb z = \mb x \mb x^* \mb z.
\end{align*}
In other words, $2\norm{\mb z}{}^2 - \norm{\mb x}{}^2$ is the positive eigenvalue of the rank-one PSD Hermitian matrix $\mb x \mb x^*$. Hence $2\norm{\mb z}{}^2 - \norm{\mb x}{}^2 = \norm{\mb x}{}^2$. This would imply that $\norm{\mb z}{} = \norm{\mb x}{}$, which does not occur in this region. 
\item When $\norm{\mb z}{}^2 = \norm{\mb x}{}^2$, critical points must satisfy 
\begin{align*}
 \paren{\norm{\mb x}{}^2 \mb I - \mb x\mb x^* }\mb z = \mb 0, 
\end{align*}
and so $\mb z \not \in \mathrm{null}\paren{\mb x\mb x^* }$. Given that $\norm{\mb z}{} =\norm{\mb x}{}$, we must have $\mb z = \mb x\e^{\im \theta}$ for some $\theta\in [0,2\pi)$. Since $f$ is a nonnegative function, and $f(\mb z) = 0$ for any $\mb z \in \mc X$, $\mc X$ is indeed also the global optimal set.
\item  For $\norm{\mb z}{} > \norm{\mb x}{}$, since the gradient $\begin{bmatrix}
		\mb z \\ \ol{\mb z}
	\end{bmatrix}^* \nabla \bb E\brac{ f(\mb z) } > 0$, there is no critical point present.
\end{itemize}
Summarizing the above observations completes the proof.	
\end{proof}

This result suggests that the same qualitative properties that we observed for $f(\mb z)$ with $\mb z \in \reals^2$ also hold for higher-dimensional, complex $\mb z$. The high-dimensional analysis is facilitated by the unitary invariance of the complex normal distribution -- the properties of $\expect{f}$ at a given point $\mb z$ depend only the norm of $\mb z$ and its inner product with the target vector $\mb x$, i.e., $\mb x^* \mb z$. In the next section, we will show that the important qualitative aspects of this structure are preserved even when $m$ is as small as $Cn \log^3 n$.

%% file: sec/finite.tex

\subsection{The Finite-Sample Landscape} \label{sec:geometry-finite}

The following theorem characterizes the geometry of the objective function $f(\mb z)$, when the number of samples $m$ is roughly on the order of $n$ -- degrees of freedom of $\mb x$. The main conclusion is that the space $\bb C^n$ can be divided into three regions, in which the objective either exhibits negative curvature, strong gradient, or restricted strong convexity. 

The result is not surprising in view of the above characterization of the ``large-sample'' landscape. The intuition is as follows: since the objective function is a sum of independent random variables, when $m$ is sufficiently large, the function values, gradients and Hessians should be uniformly close to their expectations. Some care is required in making this intuition precise, however. Because the objective function contains fourth powers of Gaussian random variables, it is heavy tailed. Ensuring that $f$ and its derivatives are uniformly close to their expectations requires $m \ge Cn^2$. This would be quite wasteful, since $\mb x$ has only $n$ degrees of freedom. 

Fortunately, when $m \ge C n \; \mathrm{polylog}( n )$, w.h.p.\ $f$ {\em still} has benign global geometry, even though its gradient is not uniformly close to its expectation. Perhaps surprisingly, the heavy tailed behavior of $f$ only helps to {\em prevent} spurious local minimizers -- away from the global minimizers and saddle points, the gradient can be sporadically large, but it cannot be sporadically small. This behavior will follow by expressing the decrease of the function along a certain carefully chosen descent direction as a sum of random variables which are heavy tailed, but are also {\em nonnegative}. Because they are nonnegative, their deviation below their expectation is bounded, and their lower tail is well-behaved. More discussion on this can be found in the next section and the proofs in Section~\ref{app:finite}. 


Our main geometric result is as follows:
\js{
\begin{theorem}[Main Geometric Results]\label{thm:finite-landscape}
	There exist positive absolute constants $C, c$, such that when $m \geq Cn \log^3 n$, it holds with probability at least $1 - cm^{-1}$ that $f(\mb z)$ has no spurious local minimizers and the only local/global minimizers are exactly the target set $\mc X$. More precisely, with the same probability, 
	\begin{align*}
		 \frac{1}{\norm{\mb x}{}^2}
\begin{bmatrix}
\mb x \e^{\im \phi(\mb z)} \\
\ol{\mb x} \e^{-\im \phi(\mb z)}
\end{bmatrix}^*   \nabla^2 f(\mb z) 
\begin{bmatrix}
\mb x \e^{\im \phi(\mb z)}\\
\ol{\mb x} \e^{-\im \phi(\mb z)}
\end{bmatrix} 
\;& \le \;  
- \frac{1}{100} \norm{\mb x}{}^2, &\forall\; &  \mb z \in \mc R_1, &\quad \text{(Negative Curvature)} \\
\norm{\nabla_{\mb z} f(\mb z)}{} \;&\ge\; \frac{1}{1000} \norm{\mb x}{}^2 \norm{\mb z}{},  &\forall\; & \mb z \in \mc R_2, &\quad \text{(Large Gradient)}\\
\begin{bmatrix}
\mb g(\mb z) \\
\ol{\mb g(\mb z)}
\end{bmatrix}^* 
\nabla^2 f(\mb z) 
\begin{bmatrix}
\mb g(\mb z) \\
\ol{\mb g(\mb z)}
\end{bmatrix} 
\;&\ge\; \frac{1}{4} \norm{\mb x}{}^2, &\forall\; &  \mb z \in \mc R_3, &\quad \text{(Restricted Strong Convexity)}
\end{align*}
where, assuming $\mb h(\mb z)$ as defined in \eqref{eqn:basic-def},  
\begin{align*}
	\mb g(\mb z) \doteq 
\begin{cases}
\mb h(\mb z)/\norm{\mb h(\mb z)}{}  & \text{if}\; \mathrm{dist}(\mb z, \mc X) \neq 0, \\
\mb h \in \mc S \doteq \set{\mb h: \Im(\mb h^* \mb z) = 0, \norm{\mb h}{} = 1}   & \text{if}\; \mb z \in \mc X. 
\end{cases}
\end{align*}
Here the regions $\mc R_1,\;\mc R_2^{\mb z},\;\mc R_2^{\mb h}$ are defined as 
\begin{align}
	\mc R_1 \;&\doteq\; \set{\mb z: \begin{bmatrix}
\mb x \e^{\im \phi(\mb z)} \\
\ol{\mb x} \e^{-\im \phi(\mb z)}
\end{bmatrix}^*  \expect{ \nabla^2 f(\mb z)} 
\begin{bmatrix}
\mb x \e^{\im \phi(\mb z)}\\
\ol{\mb x} \e^{-\im \phi(\mb z)}
\end{bmatrix} \le -\frac{1}{100} \norm{\mb x}{}^2 \norm{\mb z}{}^2 - \frac{1}{50} \norm{\mb x}{}^4},\label{eqn:region-1} \\
	\mc R_3 \;&\doteq\; \set{\mb z: \mathrm{dist}(\mb z, \mc X) \le \frac{1}{\sqrt{7}}\norm{\mb x}{}}, \label{eqn:region-3} \\
	\mc R_2 \;&\doteq\; \paren{\mc R_1 \cup \mc R_3}^c. 
\end{align}
\end{theorem}
}
\begin{proof}
The quantitative statements are proved sequentially in Proposition \ref{prop:nega-curv}, Proposition \ref{prop:grad-region-z}, Proposition \ref{prop:grad-region-zx}, Proposition \ref{prop:str_cvx} and Proposition \ref{prop:region-cover} in the next section. We next show $\mc X$ are the only local/global minimizers. Obviously local minimizers will not occur in \js{$\mc R_1 \cup \mc R_2$}, as at each such point either the gradient is nonzero, or there is a negative curvature direction. So local/global minimizers can occur only in $\mc R_3$. From~\eqref{eqn:grad-finite}, it is easy to check that $\nabla_{\mb z} f(\mb x \e^{\im \phi}) = \mb 0$ and $f(\mb x \e^{\im \phi}) = 0$ for any $\phi \in [0, 2 \pi)$. Since $f(\mb z) \ge 0$ for all $\mb z \in \Cp^n$, all elements of $\mc X$ are local/global minimizers. To see there is no other critical point in $\mc R_3$, note that any point $\mb z \in \mc R_3\setminus \mc X$ can be written as 
\begin{align*}
\mb z = \mb x \e^{\im \phi(\mb z)} + t \mb g, \quad \mb g \doteq \mb h(\mb z)/\norm{\mb h(\mb z)}{}, \; t \doteq \mathrm{dist}(\mb z, \mc X). 
\end{align*}
By the restricted strong convexity we have established for $\mc R_3$, and the integral form of Taylor's theorem in Lemma~\ref{lem:Taylor-integral-form}, 
\begin{align*}
f(\mb z) 
= f(\mb x \e^{\im \phi(\mb z)}) + t
\begin{bmatrix}
\mb g \\
\ol{\mb g}
\end{bmatrix}^*
\nabla f(\mb x \e^{\im \phi(\mb z)}) 
+ 
t^2 \int_0^1 (1-s) 
\begin{bmatrix}
\mb g \\
\ol{\mb g}
\end{bmatrix}^*
\nabla^2 f(\mb x \e^{\im \phi(\mb z)} + st\mb g) 
\begin{bmatrix}
\mb g \\
\ol{\mb g}
\end{bmatrix} 
\; ds 
\ge \frac{1}{8} \norm{\mb x}{}^2 t^2. 
\end{align*}
similarly, we obtain 
\begin{align*}
f(\mb x \e^{\im \phi(\mb z)}) = 0 
& \ge f(\mb z) -  t
\begin{bmatrix}
\mb g \\
\ol{\mb g}
\end{bmatrix}^*
\nabla f(\mb z) 
+ t^2 \int_0^1 (1-s) 
\begin{bmatrix}
\mb g \\
\ol{\mb g}
\end{bmatrix}^*
\nabla^2 f(\mb z - st\mb g) 
\begin{bmatrix}
\mb g \\
\ol{\mb g}
\end{bmatrix} 
\; ds \\
& \ge f(\mb z) - 
\begin{bmatrix}
\mb g \\
\ol{\mb g}
\end{bmatrix}^*
\nabla f(\mb z) + \frac{1}{8} \norm{\mb x}{}^2 t^2. 
\end{align*}
Summing up the above two inequalities, we obtain 
\begin{align*}
t
\begin{bmatrix}
\mb g \\
\ol{\mb g}
\end{bmatrix}^*
\nabla f(\mb z) \ge \frac{1}{4} \norm{\mb x}{}^2 t^2
\Longrightarrow
\norm{\nabla f(\mb z)}{} \ge \frac{1}{4\sqrt{2}} \norm{\mb x}{}^2 t, 
\end{align*}
as desired.  
\end{proof}

\begin{figure}[t]
\centerline{\includegraphics[width=0.8\linewidth]{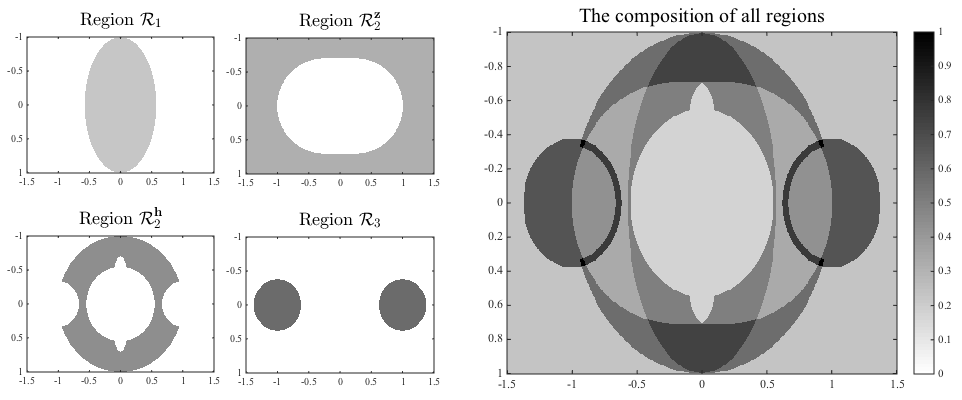}}
\caption{Schematic illustration of partitioning regions for Theorem~\ref{thm:finite-landscape}. This plot corresponds to Figure~\ref{fig:geo-2d}, i.e., the target signal is $\mb x = [1; 0]$ and measurements are real Gaussians, such that the function is defined in $\R^2$. \js{Here $\mc R_2^{\mb z} \cup \mc R_2^{\mb h}$ is $\mc R_2$; we will need the further sub-division of $\mc R_2$ in the proof}. } 
\label{fig:region}
\end{figure}

Figure~\ref{fig:region} visualizes the different regions described in Theorem \ref{thm:finite-landscape}, and gives an idea of how they cover the space. For $f(\mb z)$, a point $\mb z \in \Cp^n$ is either near a critical point such that the gradient $\nabla_{\mb z} f(\mb z)$ is small (in magnitude), or far from a critical point such that the gradient is large. Any point in \js{$\mc R_2$} is far from a critical point. The rest of the space consists of points near critical points, and is covered by $\mc R_1 \cup \mc R_3$. For any $\mb z$ in $\mc R_1$, the quantity 
\begin{align*}
\frac{1}{\norm{\mb x}{}^2}
\begin{bmatrix}
\mb x \e^{\im \phi(\mb z)} \\
\ol{\mb x} \e^{-\im \phi(\mb z)}
\end{bmatrix}^*   \nabla^2 f(\mb z) 
\begin{bmatrix}
\mb x \e^{\im \phi(\mb z)}\\
\ol{\mb x} \e^{-\im \phi(\mb z)}
\end{bmatrix}
\end{align*}
measures the local curvature of $f(\mb z)$ in the $\mb x \e^{\im \phi(\mb z)}$ direction. Strict negativity of this quantity implies that the neighboring critical point is either a local maximizer, or a saddle point. Moreover, $\mb x \e^{\im \phi(\mb z)}$ is a local descent direction, even if $\nabla_{\mb z} f(\mb z) = \mb 0$. For any $\mb z \in \mc R_3$, $\mb g(\mb z)$ is the unit vector that points to $\mb x \e^{\im \phi(\mb z)}$, and is also geometrically orthogonal to the $\im \mb x \e^{\im \phi(\mb z)}$ which is tangent the circle $\mc X$ at $\mb x \e^{\im \phi(\mb z)}$. The strict positivity of the quantity
\begin{align*}
\begin{bmatrix}
\mb g(\mb z) \\
\ol{\mb g(\mb z)}
\end{bmatrix}^* 
\nabla^2 f(\mb z) 
\begin{bmatrix}
\mb g(\mb z) \\
\ol{\mb g(\mb z)}
\end{bmatrix} 
\end{align*}
implies that locally $f(\mb z)$ is strongly convex in $\mb g(\mb z)$ direction, although it is flat on the complex circle $\set{\mb z\e^{\im \phi}: \phi \in [0, 2\pi) }$. In particular, the result applied to $\mb z \in \mc X$ implies that on $\mc X$, $f(\mb z)$ is strongly convex in any direction orthogonal to $\mc X$ (i.e., any ``radial'' direction w.r.t. $\mc X$). This observation, together with the fact that the Hessian is Lipschitz, implies that there is a neighborhood $N(\mc X)$ of $\mc X$, such that for all $\mb z \in N(\mc X)$, $\mb v^* \nabla^2 f(\mb z) \mb v > 0$ for {\em every} $\mb v$ that is orthogonal to the trivial direction $\im \mb z$, not just the particular direction $\mb g(\mb z)$. This stronger property (captured by Lemma~\ref{lem:hessian-lower-upper}) can be used to study the asymptotic convergence rate of algorithms; in particular, we will use it to obtain quadratic convergence for a certain variant of the trust-region method. 

In the asymptotic version, we characterized only the critical points. In this finite-sample version, we characterize the whole space and particularly provide quantitative control for regions near critical points (i.e., $\mc R_1 \cup \mc R_3$). These concrete quantities are important for algorithm design and analysis (see Section~\ref{sec:algorithm}). 

In sum, our objective $f(\mb z)$ has the benign geometry that each $\mb z \in \Cp^n$ has either large gradient or negative directional curvature, or lies in the vicinity of local minimizers around which the function is locally restrictedly strongly convex. Functions with this property lie in the ridable-saddle function class~\cite{ge2015escaping,sun2015nonconvex}. Functions in this class admit simple iterative methods (including the noisy gradient method, curvilinear search, and trust-region methods), which avoid being trapped near saddle points, and obtain a local minimizer \js{asymptotically}. Theorem \ref{thm:finite-landscape} shows that for our problem, every local minimizer is global, and so for our problem, these algorithms obtain a global minimizer \js{asymptotically}. \js{Moreover, with appropriate quantitative assumptions on the geometric structure as we obtained (i.e., either gradient is \emph{sufficiently} large, or the direction curvature is \emph{sufficiently} negative, or local directional convexity is \emph{sufficiently} strong), these candidate methods actually find a global minimizer in polynomial time}. 

\subsection{Key Steps in the Geometric Analysis} \label{sec:geometry-key} Our proof strategy is fairly simple: we work out uniform bounds on the quantities for each of the three regions, and finally show the regions together cover the space. Since~\eqref{eqn:finite-f} and associated derivatives take the form of summation of $m$ independent random variables, the proof involves concentration and covering arguments \cite{vershynin2012introduction}. The main challenge in our argument will be the heavy-tailedness nature of $f$ and its gradient. 

\begin{proposition} [Negative Curvature] \label{prop:nega-curv}
When $m \ge Cn\log n$, it holds with probability at least $1 - c m^{-1}$ that 
\begin{align*}
\frac{1}{\norm{\mb x}{}^2}
\begin{bmatrix}
\mb x \e^{\im \phi(\mb z)} \\
\ol{\mb x} \e^{-\im \phi(\mb z)}
\end{bmatrix}^*   \nabla^2 f(\mb z) 
\begin{bmatrix}
\mb x \e^{\im \phi(\mb z)}\\
\ol{\mb x} \e^{-\im \phi(\mb z)}
\end{bmatrix} 
\le 
- \frac{1}{100} \norm{\mb x}{}^2
\end{align*}
for all $\mb z \in \mc R_1$ defined in \eqref{eqn:region-1}. Here $C, c$ are positive absolute constants. 
\end{proposition}
\begin{proof}
See Section~\ref{pf:prop:nega-curv} on Page~\pageref{pf:prop:nega-curv}. 
\end{proof}

Next, we show that \js{near $\mc X$, the objective $f$ is strongly convex in any geometrically normal direction to the target set $\mc X$ (which is a one-dimensional circle)}. Combined with the smoothness property, this allows us to achieve a quadratic asymptotic rate of convergence with the modified trust-region algorithm we propose later. 
\begin{proposition}[Restricted Strong Convexity near $\mc X$] \label{prop:str_cvx}
When $m \ge Cn\log n$ for a sufficiently large constant $C$, it holds with probability at least $1 - cm^{-1}$ that 
\begin{align*}
\begin{bmatrix}
\mb g(\mb z) \\
\ol{\mb g(\mb z)}
\end{bmatrix}^* 
\nabla^2 f(\mb z) 
\begin{bmatrix}
\mb g(\mb z) \\
\ol{\mb g(\mb z)}
\end{bmatrix} 
\ge \frac{1}{4} \norm{\mb x}{}^2 
\end{align*}
for all $\mb z \in \mc R_3$ defined in \eqref{eqn:region-3} and for all 
\begin{align*}
\mb g(\mb z) \doteq 
\begin{cases}
\paren{\mb z - \mb x \e^{\im \phi(\mb z)}}/\norm{\mb z - \mb x \e^{\im \phi(\mb z)}}{}  & \text{if}\; \mathrm{dist}(\mb z, \mc X) \neq 0, \\
\mb h \in \mc S \doteq \set{\mb h: \Im(\mb h^* \mb z) = 0, \norm{\mb h}{} = 1}   & \text{if}\; \mb z \in \mc X. 
\end{cases}
\end{align*}
Here $c$ is a positive absolute constant.  
\end{proposition}
\begin{proof}
See Section~\ref{pf:prop:str_cvx} on Page~\pageref{pf:prop:str_cvx}. 
\end{proof}
\js{This restricted strong convexity result is qualitatively stronger than the local curvature property (i.e., Condition 7.11) established in~\cite{candes2015phase}. Specifically, our result is equivalent to the following: for any line segment $\mc L$ that is normal to the circle $\mc X$ and contained in $\mc S_3$, it holds that 
\begin{align*}
\innerprod{\nabla_{\mb z}f(\mb z_1) - \nabla_{\mb z}(\mb z_2)}{\mb z_1 - \mb z_2} \ge C \norm{\mb x}{}^2\norm{\mb z_1 - \mb z_2}{}^2, \quad \forall\; \mb z_1, \mb z_2 \in \mc L.  
\end{align*} 
In contrast, the local curvature property~\cite{candes2015phase} only states that in such line segment,  
\begin{align*}
\innerprod{\nabla_{\mb z} f(\mb w) - \nabla_{\mb z} f(\mb x \e^{\im \phi})}{\mb w - \mb x \e^{\im \phi}} \ge C' \norm{\mb x}{}^2 \norm{\mb z - \mb x \e^{\im \phi}}{}^2, \quad \forall\; \mb w \in \mc L. 
\end{align*}
While the local curvature property is sufficient to establish local convergence result of first-order method, the stronger restricted strong convex that provides uniform second-order curvature controls in the above $\mc L$'s are necessary to showing quadratic convergence result of second-order method, as we will do in the next section. 
}

\js{Next we show that the gradients in $\mc R_2$ are bounded away from zero. This is the most tricky part in the proof. Directly working the gradient entails arguing concentration of heavy-tailed random vectors. With only $O(n \mathrm{polylog}(n))$ samples, such concentration is not guaranteed. We get around the problem by arguing directional derivatives in well-chosen directions are concentrated and bounded away from zero, indicating non-vanishing gradients. A natural choice of the direction is the expected gradient, $\expect{\nabla_{\mb z}f}$, which is a linear combination of $\mb z$ and $\mb x$. It turns out directly working with the resulting directional derivatives still faces obstruction due to the heavy-tailed nature of the random variables. Thus, we carefully} divide $\mc R_2$ into two overlapped regions, $\mc R_2^{\mb z}$ and $\mc R_2^{\mb h}$, roughly matching the case $\Re\paren{\mb z^* \expect{\nabla_{\mb z} f(\mb z)}} > 0$ and the case $\Re\paren{\paren{\mb z - \mb x \e^{\im \phi(\mb z)}}^* \expect{\nabla_{\mb z} f(\mb z)}} > 0$, respectively. \js{The two sub-regions are defined as:
\begin{align}
	\mc R_2^{\mb z} \;&\doteq\; \set{\mb z: \Re\paren{\innerprod{\mb z}{\nabla_{\mb z} \expect{f}} } \ge \frac{1}{100} \norm{\mb z}{}^4 + \frac{1}{500} \norm{\mb x}{}^2 \norm{\mb z}{}^2 }, \label{eqn:region-2-z} \\
	\mc R_2^{\mb h} \;&\doteq\; \left\{\mb z: \Re\paren{\innerprod{\mb h(\mb z) }{\nabla_{\mb z}\expect{f}}} \geq \frac{1}{250}\norm{\mb x}{}^2 \norm{\mb z}{} \norm{\mb h(\mb z)}{}, \right. \nonumber\\
	 & \qquad \qquad \qquad \left. \frac{11}{20} \norm{\mb x}{} \le \norm{\mb z}{} \le \norm{\mb x}{}, \mathrm{dist}(\mb z, \mc X) \ge \frac{\norm{\mb x}{}}{3} \right\}. \label{eqn:region-2-h} 
\end{align}
Figure~\ref{fig:region} provides a schematic visualization of the division in $\R^2$. } 

\js{In Proposition~\ref{prop:region-cover} below, we will show that $\mc R_2^{\mb z} \cup \mc R_2^{\mb h}$ indeed cover $\mc R_2$. We first show that the gradients in either region are bounded away from zero.}

\begin{proposition}\label{prop:grad-region-z}
When $m \ge Cn\log n$, it holds with probability at least $1 - cm^{-1}$ that 
\begin{align*}
\frac{\mb z^* \nabla_{\mb z} f(\mb z)}{\norm{\mb z}{}} \geq \frac{1}{1000} \norm{\mb x}{}^2 \norm{\mb z}{} 
\end{align*}
for all $\mb z \in \mc R_2^{\mb z}$ defined in \eqref{eqn:region-2-z}. Here $C,c$ are positive absolute constants. 
\end{proposition}
\begin{proof}
See Section~\ref{pf:prop:grad-region-z} on Page~\pageref{pf:prop:grad-region-z}. 
\end{proof}

\js{It follows immediately that $\norm{\nabla_{\mb z} f(\mb z)}{} \ge \norm{\mb x}{}^2/1000$ for all $\mb z \in \mc R_2^{\mb z}$. }

\begin{proposition}\label{prop:grad-region-zx}
When $m \ge C n \log^3n$, it holds with probability at least $1 - c m^{-1}$ that 
\begin{align*}
\Re\paren{\mb h(\mb z)^* \nabla_{\mb z} f(\mb z)} \geq  \frac{1}{1000} \norm{\mb x}{}^2 \norm{\mb z}{}\norm{\mb h(\mb z)}{}
\end{align*}
for all $\mb z \in \mc R_2^{\mb h}$ defined in \eqref{eqn:region-2-h}. Here $C, c$ are positive absolute constants. 
\end{proposition}
\begin{proof}
See Section~\ref{pf:prop:grad-region-zx} on Page~\pageref{pf:prop:grad-region-zx}. 
\end{proof}

\js{This clearly implies that $\norm{\nabla_{\mb z} f(\mb z)}{} \ge \norm{\mb x}{}^2/1000$ for all $\mb z \in \mc R_2^{\mb h}$. }

\js{The quantity we want to control in the above proposition, $\Re\paren{\mb h(\mb z)^* \nabla_{\mb z} f(\mb z)}$, is the same quantity to be controlled in the local curvature condition (e.g., Condition 7.11) in~\cite{candes2015phase}. There only points near $\mc X$ (i.e., roughly our $\mc R_3$ below) are considered, and the target is proving the local curvature is in a certain sense positive. Here the points of interest are not close to $\mc X$, and the target is only showing that at these points, the directional derivative in $\mb h(\mb z)$ direction is bounded away from zero. }

Finally, we show that the two sub-regions, $\mc R_2^{\mb z}$ and $\mc R_2^{\mb h}$, together cover $\mc R_2$.  Formally,  
\js{
\begin{proposition}\label{prop:region-cover}
We have $\mc R_2 \subset \mc R_2^{\mb z} \cup \mc R_2^{\mb h}$.
\end{proposition}
}
\begin{proof}
See Section~\ref{pf:prop:region-cover} on Page~\pageref{pf:prop:region-cover}. 
\end{proof}

The main challenge is that the function~\eqref{eqn:finite-f} is a fourth-order polynomial, and most quantities arising in the above propositions involve heavy-tailed random variables. For example, we need to control
\begin{align} \label{eq:quan-1}
\frac{1}{m} \sum_{k=1}^m \abs{\mb a_k^* \mb z}^4 \quad \text{for all} \; \mb z \in \mc R_2^{\mb z}
\end{align}  
in proving Propositions~\ref{prop:nega-curv} and~\ref{prop:grad-region-z},
\begin{align} \label{eq:quan-3}
\frac{1}{m} \sum_{k=1}^m \abs{\mb a_k^* \mb w}^2 \abs{\mb a_k^* \mb z}^2 \quad \text{for all}\; \mb w, \mb z
\end{align}
in proving Proposition~\ref{prop:str_cvx}, and   
\begin{align} \label{eq:quan-2}
\frac{1}{m}\sum_{k=1}^m \abs{\mb a_k^* \mb z}^2 \Re\paren{(\mb z - \mb x \e^{\im \phi})^* \mb a_k \mb a_k^* \mb z} \quad \text{for all} \; \mb z \in \mc R_2^{\mb h}
\end{align}
in proving Proposition~\ref{prop:grad-region-zx}. With only $Cn \log^3 n$ samples, these quantities do not concentrate uniformly about their expectations. Fortunately, this heavy-tailed behavior does not prevent the objective function from being globally well-structured for optimization. Our bounds on the gradient and Hessian depend only on the \emph{lower tails} of the above quantities. For~\eqref{eq:quan-1} and~\eqref{eq:quan-3} that are sum of independent nonnegative random variables, the lower tails concentrate uniformly as these lower-bounded variables are sub-Gaussian viewed from the lower tails (see Lemma~\ref{lem:subgauss_nonneg} and Lemma~\ref{lem:min_fourth}); \js{such one-sided concentration was also exploited in prior work~\cite{candes2015phase, chen2015solving} to control similar quantities. The actual concentration inequalities in use are slightly different.} For~\eqref{eq:quan-2}, \js{it is not possible to have two-sided control of the operator 
$1/m \cdot \sum_{k=1}^m \abs{\mb a_k^* \mb z}^2 \mb a_k \mb a_k^*$, or the vector $1/m \cdot \sum_{k=1}^m \abs{\mb a_k^* \mb z}^2 \mb a_k \mb a_k^*\mb z$ uniformly for all $\mb z \in \Cp^n$ with only $O(n \mathrm{polylog}(n))$ samples. The sample constraint also precludes carrying out a typical ``concentration plus union bound'' argument on the original quantity. To get around the difficulty, we carefully construct a proxy quantity that is summation of bounded random variables and tightly bounds~\eqref{eq:quan-2} from below. This proxy quantity is well behaved in lower tail and amenable to a typical concentration argument. 
}

%% file: sec/algorithm.tex

\section{Optimization by Trust-Region Method (TRM)}\label{sec:algorithm}

Based on the geometric characterization in Section~\ref{sec:geometry-finite}, we describe a second-order trust-region algorithm that produces a close approximation (i.e., up to numerical precision) to a global minimizer of~\eqref{eqn:finite-f} in polynomial number of steps. One interesting aspect of $f$ in the complex space is that each point has a ``circle'' of equivalent points that have the same function value. Thus, we constrain each step to move ``orthogonal'' to the trivial direction. This simple modification helps the algorithm to converge faster in practice, and proves important to the quadratic asymptotic convergence rate in theory. 

\subsection{A Modified Trust-Region Algorithm}\label{sec:trust-region-algorithm}

The basic idea of the trust-region method is simple: we generate a sequence of iterates $\mb z^{(0)}, \mb z^{(1)}, \dots$, by repeatedly constructing quadratic approximations $\wh{f}(\mb \delta; \mb z^{(r)}) \approx f(\mb z^{(r)} + \mb \delta )$, minimizing $\wh{f}$ to obtain a step $\mb \delta$, and setting $\mb z^{(r+1)} = \mb z^{(r)} + \mb \delta$. 
More precisely, we approximate $f(\mb z)$ around $\mb z^{(r)}$ using the second-order Taylor expansion, 
\begin{align*}
	\wh{f}(\mb \delta; \mb z^{(r)}) =  f(\mb z^{(r)}) + \begin{bmatrix}
		\mb \delta \\ \ol{\mb \delta}\end{bmatrix}^* \nabla f(\mb z^{(r)}) + \frac{1}{2} \begin{bmatrix}
		\mb \delta \\ \ol{\mb \delta}	\end{bmatrix}^* \nabla^2 f(\mb z^{(r)})  \begin{bmatrix}
		\mb \delta \\ \ol{\mb \delta}	\end{bmatrix},
\end{align*}
and solve 
\begin{align}\label{eqn:trm-subproblem}
	\mini_{\mb \delta \in \Cp^n}\; \wh{f}(\mb \delta; \mb z^{(r)} ),\quad \st \quad \Im\paren{ \mb \delta^* \mb z^{(r)}  } = 0,\quad  \norm{\mb \delta }{} \leq \Delta,
\end{align}
to obtain the step $\mb \delta$. In \eqref{eqn:trm-subproblem}, $\Delta$ controls the trust-region size. The first linear constraint further forces the movement $\mb \delta$ to be geometrically orthogonal to the $\im \mb z$ direction, along which the possibility for reducing the function value is limited. Enforcing this linear constraint is a strategic modification to the classical trust-region subproblem. 

\paragraph{Reduction to the standard trust-region subproblem.} The modified trust-region subproblem is easily seen to be equivalent to the classical trust-region subproblem (with no constraint) over $2n-1$ real variables. Notice that $\set{ \mb w \in \bb C^n : \Im( \mb w^*\mb z^{(r)}) =0}$ forms a subspace of dimension $2n-1$ over $\R^{2n}$ (the canonical identification of $\Cp^n$ and $\R^{2n}$ applies whenever needed! ). Take any matrix $\mb U(\mb z^{(r)}) \in \Cp^{n \times (2n-1)}$ whose columns form an orthonormal basis for the subspace, i.e., $\Re(\mb U_i^* \mb U_j) = \delta_{ij}$ for any columns $\mb U_i$ and $\mb U_j$. The subproblem can then be reformulated as ($\mb U$ short for $\mb U(\mb z^{(r)})$)
\begin{align}\label{eqn:TRM-subproblem-0}
	\mini_{\mb \xi \in \bb R^{2n-1} }  \wh{f}(\mb U \mb \xi;\mb z^{(r)} ),\quad \st \quad \norm{\mb \xi}{} \leq \Delta. 
\end{align}
Let us define 
\begin{align}\label{eqn:def-g-H}
	\mb g(\mb z^{(r)}) \doteq \begin{bmatrix}
		\mb U \\ \ol{\mb U}	\end{bmatrix}^* \nabla f(\mb z^{(r)}), \quad \mb H(\mb z^{(r)}) \doteq \begin{bmatrix}
		\mb U \\ \ol{\mb U}	\end{bmatrix}^* \nabla^2 f(\mb z^{(r)}) \begin{bmatrix}
		\mb U \\ \ol{\mb U}	\end{bmatrix}. 
\end{align}
Then, the quadratic approximation of $f(\mb z)$ around $\mb z^{(r)}$ can be rewritten as
\begin{align}\label{eqn:quadratic-simplified}
	\wh{f}(\mb \xi; \mb z^{(r)}) = f(\mb z^{(r)}) + \mb \xi^\top \mb g(\mb z^{(r)}) + \frac{1}{2} \mb \xi^\top \mb H(\mb z^{(r)}) \mb \xi.
\end{align}
By structure of the Wirtinger gradient $\nabla f(\mb z^{(r)})$ and Wirtinger Hessian $\nabla^2 f(\mb z^{(r)})$, $\mb g(\mb z^{(r)})$ and $\mb H(\mb z^{(r)})$ contain only real entries. \js{Thus, the problem~\eqref{eqn:TRM-subproblem-0} is in fact an instance of the classical trust-region subproblem w.r.t. real variable $\mb \xi$. A minimizer to~\eqref{eqn:trm-subproblem} can be obtained from a minimizer of~\eqref{eqn:TRM-subproblem-0} $\mb \xi_\star$ as $\mb \delta_\star = \mb U \mb \xi_\star$. } 

So, any method which can solve the classical trust-region subproblem can be directly applied to the modified problem \eqref{eqn:trm-subproblem}. Although the resulting problem can be nonconvex (\js{as $\mb H(\mb z^{(r)})$ in~\eqref{eqn:quadratic-simplified} can be indefinite}), it can be solved in polynomial time, by root-finding~\cite{more1983computing,conn2000trust} or SDP relaxation~\cite{rendl1997semidefinite,fortin2004trust}. \js{Our convergence guarantees assume an exact solution of this problem; we outline below how to obtain an $\eps$-near solution (for arbitrary $\eps > 0$) via bisection search in polynomial time (i.e., polynomial in $\eps^{-1}$), due to~\cite{ye1992affine,vavasis1990proving}. At the end of Section~\ref{sec:conv_trm}, we will discuss robustness of our convergence guarantee to the numerical imperfection, and provide an estimate of the total time complexity}. In practice, though, even very inexact solutions of the trust-region subproblem suffice.\footnote{This can also be proved, in a relatively straightforward way, using the geometry of the objective $f$. In the interest of brevity, we do not pursue this here.} Inexact iterative solvers for the trust-region subproblem can be engineered to avoid the need to densely represent the Hessian; these methods have the attractive property that they attempt to optimize the amount of Hessian information that is used at each iteration, in order to balance rate of convergence and computation. 

\js{Now we describe briefly how to apply bisection search to find an approximate solution to the classical trust-region subproblem, i.e., 
\begin{align} \label{eq:trm_sp}
\mini_{\mb w \in \R^d} Q\paren{\mb w} \doteq \frac{1}{2} \mb w^\top \mb A \mb w + \mb b^\top \mb w, \; \st \; \norm{\mb w}{} \le r, 
\end{align}  
where $\mb A \in \R^{d \times d}$ is symmetric and $\mb b \in \R^d$, and $r > 0$. The results are distilled from~\cite{ye1992affine,vavasis1990proving,nocedal2006numerical}. Optimizers to~\eqref{eq:trm_sp} are characterized as follows: a feasible $\mb w_\star$ is a global minimizer to~\eqref{eq:trm_sp} if and only if there exists a $\lambda_\star \ge 0$ that satisfies: 
\begin{align*}
\paren{\mb A + \lambda_\star \mb I}\mb w_\star = -\mb b, \quad \lambda_\star \paren{\norm{\mb w_\star}{} - r} = 0, \quad \mb A + \lambda_\star \mb I \succeq \mb 0. 
\end{align*}

We first isolate the case when there is a unique interior global minimizer. In this case $\lambda_\star = 0$ and $\mb A \succeq \mb 0$. When $\mb A$ has a zero eigenvalue with an associated eigenvector $\mb v$ and $\mb w_\star$ is an interior minimizer, all feasible $\mb w_\star + t \mb v$ are also minimizers. So uniqueness here requires $\mb A \succ \mb 0$. Thus, one can try to sequentially (1) test if $\mb A$ is positive definite; (2) solve for $\mb w_\star$ from $\mb A\mb w_\star = -\mb b$; and (3) test if $\mb w_\star$ is feasible. If the procedure goes through, a minimizer has been found. It is obvious the arithmetic complexity is $O(d^3)$.

When $\mb A$ is not positive definite, there must be minimizers on the boundary. Then, we need to find a $\lambda_\star \ge 0$ such that 
\begin{align*}
\paren{\mb A + \lambda_\star \mb I}\mb w_\star = -\mb b, \quad \mb A + \lambda_\star \mb I \succeq \mb 0. 
\end{align*}
One remarkable fact is that although the minimizers may not be unique, the multiplier $\lambda_\star$ is unique. The unique $\lambda_\star$ can be efficiently approximated via a bisection search algorithm (Algorithm~\ref{alg:bisection_search}). 
\begin{algorithm}[!htbp]
\caption{Bisection Search for Finding $\lambda_\star$ and $\mb w_\star$~\cite{vavasis1990proving}}
\label{alg:bisection_search}
\begin{algorithmic}[1]
\renewcommand{\algorithmicrequire}{\textbf{Input:}}
\renewcommand{\algorithmicensure}{\textbf{Output:}}
\Require{Data: $\mb A \in \R^{d \times d}$, $\mb b \in \R^d$, $r \in \R$, target accuracy $\eps > 0$}
\Ensure{$\wh{\lambda}$}
\State Initialize $\lambda_L = 0$, $\lambda_U = \norm{\mb b}{}/r + d \norm{\mb A}{\infty}$
\While{$\lambda_U - \lambda_L \ge \eps$}
\State Set $\lambda_M = \paren{\lambda_L + \lambda_U}/2$
\State Compute small eigenvalue $\lambda_{\min}$ of $\mb A + \lambda_M \mb I$
\If{$\lambda_{\min} < 0$}
\State Set $\lambda_L = \lambda_M$
\Else 
\If{$\lambda_{\min} < 0$}
\State Set $\lambda_M = \lambda_M + \eps/10$
\EndIf
\State Solve for $\mb w$ from $\paren{\mb A + \lambda_M \mb I} \mb w = - \mb b$
\If{$\norm{\mb w}{} \ge r$} 
\State Set $\lambda_L = \lambda_M$
\Else 
\State Set $\lambda_U = \lambda_M$
\EndIf
\EndIf
\EndWhile
\State Set $\wh{\lambda} = \lambda_U$
\end{algorithmic}
\end{algorithm}
The above algorithm finds a $\wh{\lambda}$ with $|\wh{\lambda} - \lambda_\star| \le \eps$ with arithmetic complexity $O(d^3 \log\paren{1/\eps})$. Moreover, this can be translated into a convergence result in function value. Define $\wh{\mb w} = -(\mb A + \wh{\lambda} \mb I)^{-1} \mb b$ if $\mb A + \wh{\lambda} \mb I \succ \mb 0$, and $\wh{\mb w} = -(\mb A + \wh{\lambda} \mb I)^{\dagger} \mb b + t \mb v$ if $\mb A + \wh{\lambda} \mb I$ has zero eigenvalue with an associated eigenvector $\mb v$ and $t$ makes $\norm{\wh{\mb w}}{} = r$. Then, 
\begin{align}
Q(\wh{\mb w}) - Q\paren{\mb w_\star} \le \eps. 
\end{align}
}

%% file: sec/convergence.tex

\subsection{Convergence Analysis} \label{sec:conv_overview}


Our convergence proof proceeds as follows. Let $\mb \delta^\star$ denote the optimizer of the trust-region subproblem at a point $\mb z$. If $\norm{\nabla f(\mb z)}{}$ is bounded away from zero, or $\lambda_{\mathrm{min}}(\nabla^2 f(\mb z))$ is bounded below zero, we can guarantee that $\widehat{f}(\mb \delta^\star, \mb z ) -f(\mb z) < - \eps$, for some $\eps$ which depends on our bounds on these quantities. Because $f(\mb z + \mb \delta^\star) \approx \widehat{f}(\mb \delta^\star, \mb z ) < f (\mb z ) - \eps$, we can guarantee (roughly) an $\eps$ decrease in the objective function at each iteration. Because this $\eps$ is uniformly bounded away from zero over the gradient and negative curvature regions, the algorithm can take at most finitely many steps in these regions. Once it enters the strong convexity region around the global minimizers, the algorithm behaves much like a typical Newton-style algorithm; in particular, it exhibits asymptotic quadratic convergence. Below, we prove quantitative versions of these statements. We begin by stating several basic facts that are useful for the convergence proof. 
\paragraph{Norm of the target vector and initialization.} In our problem formulation, $\norm{\mb x}{}$ is not known ahead of time. However, it can be well estimated. When $\mb a \sim \mc{CN}(n)$, $\bb E\abs{\mb a^* \mb x}^2 = \norm{\mb x}{}^2$. By Bernstein's inequality, $\tfrac{1}{m}\sum_{k=1}^m \abs{\mb a_k^* \mb x}^2 \ge \tfrac{1}{9} \norm{\mb x}{}^2$ with probability at least $1 - \exp(-cm)$. Thus, with the same probability, the quantity 
$
R_0 \doteq 3(\tfrac{1}{m}\sum_{k=1}^m \abs{\mb a_k^* \mb x}^2)^{1/2}
$
is an upper bound for $\norm{\mb x}{}$. For the sake of analysis, we will assume the initialization $\mb z^{(0)}$ is an \emph{arbitrary point} over $\bb{CB}^n(R_0)$. Now consider a fixed $R_1 > R_0$. By Lemma~\ref{lem:op_norm_1}, Lemma~\ref{lem:min_fourth}, and the fact that $\max_{k \in [m]} \norm{\mb a_k}{}^4 \le 10 n^2 \log^2 m$ with probability at least $1 - c_am^{-n}$, we have that the following estimate 
\begin{align*}
\MoveEqLeft {\inf_{\mb z, \mb z':\; \norm{\mb z}{} \le R_0, \; \norm{\mb z'}{} \ge R_1} f(\mb z') - f(\mb z)} \\
=\; & \inf_{\mb z, \mb z':\; \norm{\mb z}{} \le R_0, \; \norm{\mb z'}{} \ge R_1} \frac{1}{m}\sum_{k=1}^m \brac{\abs{\mb a_k^* \mb z'}^4 - \abs{\mb a_k^* \mb z}^4 - 2 \abs{\mb a_k^* \mb z'}^2 \abs{\mb a_k^* \mb x}{}^2 + 2 \abs{\mb a_k^* \mb z}^2 \abs{\mb a_k^* \mb x}{}^2} \\
\ge\; & \inf_{\mb z, \mb z':\; \norm{\mb z}{} \le R_0, \; \norm{\mb z'}{} \ge R_1} \frac{199}{200} \norm{\mb z'}{}^4 - 10 n^2 \log^2 m \norm{\mb z}{}^4 - \frac{201}{200} \paren{\norm{\mb z'}{}^2 \norm{\mb x}{}^2 + \abs{\mb x^* \mb z'}^2}\\
\ge\; & \inf_{\mb z': \norm{\mb z'}{} \ge R_1} \frac{199}{200} \norm{\mb z'}{}^4 - 10 n^2 \log^2 m R_0^4 - \frac{201}{100} \norm{\mb z'}{}^2 R_0^2 
\end{align*}
holds with probability at least $1 - c_bm^{-1}$, provided $m \ge Cn\log n$ for a sufficiently large $C$. It can be checked that when 
\begin{align}
R_1 = 3 \sqrt{n \log m} R_0,
\end{align}
we have 
\begin{align*}
\inf_{\mb z': \norm{\mb z'}{} \ge R_1} \frac{199}{200} \norm{\mb z'}{}^4 - 10 n^2 \log^2 m R_0^4 - \frac{201}{100} \norm{\mb z'}{}^2 R_0^2  \ge 40 n^2 \log^2 m R_0^4. 
\end{align*}
Thus, we conclude that when $m \ge Cn\log n$, w.h.p., the sublevel set $\set{\mb z: f(\mb z) \le f(\mb z^{(0)})}$ is contained in the set 
\begin{align} \label{eq:def_gamma}
\Gamma \doteq \bb{CB}^n(R_1). 
\end{align}

\paragraph{Lipschitz Properties} We write $\mb A \doteq [\mb a_1,\cdots,\mb a_m]$ so that $ \norm{ \mb A}{\ell^1 \rightarrow \ell^2} = \max_{k \in [m]} \norm{\mb a_k }{} $. We next provide estimates of Lipschitz constants of $f$ and its derivatives, restricted to a slightly larger region than $\Gamma$: 
\begin{lemma}[Local Lipschitz Properties]\label{lem:lipschitz}
The Lipschitz constants for $f(\mb z)$, $\nabla f(\mb z)$, and $\nabla^2 f(\mb z)$ over the set $\Gamma' \doteq \bb{CB}^n(2R_1)$, denoted as $L_f$, $L_g$, and $L_h$ respectively, can be taken as 
	\begin{gather*}
		L_f  \doteq 7 \times10^6 \cdot (n\log m)^{\frac{3}{2}} \norm{\mb A}{\ell^1 \to \ell^2}^2 \norm{\mb x}{}^3 ,\quad L_g \doteq 19000\sqrt{2} n \log m\norm{\mb A}{\ell^1 \to \ell^2}^2 \norm{\mb x}{}^2, \\ L_h \doteq 480 \cdot (n\log m)^{\frac{1}{2}}\norm{\mb A}{\ell^1 \to \ell^2}^2 \norm{\mb x}{}
	\end{gather*}
	with probability at least $1 - c_a \exp(-c_b m)$, provided $m \ge Cn$ for a sufficiently large absolute constant $C$. Here $c_a$ through $c_e$ are positive absolute constants. 
\end{lemma}

\begin{proof}
	See Section~\ref{pf:lem:lipschitz} on Page~\pageref{pf:lem:lipschitz}.
\end{proof}

\paragraph{Property of Hessians near the Target Set $\mc X$. } Define a region 
\begin{align} \label{eq:r3p_def}
	\mc R_3' \doteq \Brac{\mb z:  \norm{\mb h(\mb z)}{} \leq \frac{1}{10 L_h} \norm{\mb x}{}^2}. 
\end{align}
We will provide spectral upper and lower bounds for the (restricted) Hessian matrices $\mb H(\mb z)$, where $\mb H(\mb z)$ is as defined in~\eqref{eqn:def-g-H}.  These bounds follow by bounding $\mb H(\mb z)$ on $\mc X$, and then using the Lipschitz property of the Hessian to extend the bounds to a slightly larger region around $\mc X$.

\begin{lemma}[Lower and Upper Bounds of Restricted Hessian in $\mc R_3'$]\label{lem:hessian-lower-upper}
When $m \geq Cn\log n$, it holds with probability at least $1 - cm^{-1}$ that 
\begin{align*}
	m_H \mb I \preceq \mb H(\mb z) \preceq M_H \mb I 
\end{align*}
for all $\mb z \in \mc R_3'$ with $m_H = 22/25\norm{\mb x}{}^2$ and $M_H = 9/2 \norm{\mb x}{}^2$. Here $C,c$ are positive absolute constants. 
\end{lemma}

\begin{proof}
	See Section~\ref{pf:lem:hessian-lower-upper} on Page~\pageref{pf:lem:hessian-lower-upper}.
\end{proof}

\subsection{Proof of TRM Convergence} \label{sec:conv_trm}
We are now ready to prove the convergence of the TRM algorithm. Throughout, we will assume $m \ge C n \log^3 n$ for a sufficiently large constant $C$, so that all the events of interest hold w.h.p..  

Our initialization is an arbitrary point $\mb z^{(0)} \in \bb{CB}^n(R_0) \subseteq \Gamma$. We will analyze effect of a trust-region step from any iterate $\mb z^{(r)} \in \Gamma$. Based on these arguments, we will show that whenever $\mb z^{(r)} \in \Gamma$, $ \mb z^{(r+1)} \in \Gamma$, and so the entire iterate sequence remains in $\Gamma$. The analysis will use the fact that $f$ and its derivatives are Lipschitz over the trust-region $\mb z + \bb C \bb B^n(\Delta)$. This follows from Proposition \ref{lem:lipschitz}, provided 
\begin{align}
\Delta \le R_1.
\end{align}


The next auxiliary lemma makes precise the intuition that whenever there exists a descent direction,  the step size parameter $\Delta$ is sufficiently small, a trust-region step will decrease the objective.

\begin{lemma}\label{lem:func-decent}
For any $\mb z \in \Gamma$, suppose there exists a vector $\mb \delta$ with $\norm{\mb \delta}{} \leq \Delta$ such that
\begin{align*} 
	\Im(\mb \delta^* \mb z) = 0 \quad \text{and} \quad f(\mb z+ \mb \delta) \leq f(\mb z) - d,
\end{align*}
for a certain $d > 0$. Then the trust-region subproblem \eqref{eqn:trm-subproblem} returns a point $\mb \delta_\star$ with $\norm{\mb \delta_\star }{} \leq \Delta$ and 
\begin{align*}
	f(\mb z+\mb \delta_\star) \leq f(\mb z) - d + \frac{2}{3}L_h \Delta^3.
\end{align*}
\end{lemma}
\begin{proof}
See Section~\ref{pf:lem:func-decent} on Page~\pageref{pf:lem:func-decent}. 
\end{proof}

The next proposition says when $\Delta$ is chosen properly, a trust-region step from a point with negative local curvature decreases the function value by a concrete amount. 
\begin{proposition}[Function Value Decrease in Negative Curvature Region $\mc R_1$]\label{prop:negative-curvature}
Suppose the current iterate $\mb z^{(r)} \in \mc R_1 \cap \Gamma$, and our trust-region size satisfies
	\begin{align}
		\Delta \leq \frac{1}{400L_h} \norm{\mb x}{}^2. 
	\end{align} 
	Then an optimizer $\mb \delta_\star$ to~\eqref{eqn:trm-subproblem} leads to $\mb z^{(r+1)} = \mb z^{(r)} + \mb \delta_\star$ that obeys 
	\begin{align} 
		f(\mb z^{(r+1)}) - f(\mb z^{(r)}) \leq - d_1 \doteq - \frac{1}{400} \Delta^2 \norm{\mb x}{}^2.
	\end{align}
\end{proposition}

\begin{proof}
	See Section~\ref{pf:prop:negative-curvature} on Page~\pageref{pf:prop:negative-curvature}.
\end{proof}

The next proposition shows that when $\Delta$ is chosen properly, a trust-region step from a point with strong gradient decreases the objective by a concrete amount. 

\begin{proposition}[Function Value Decrease in Large Gradient Region $\mc R_2$] \label{prop:large-gradient}
	Suppose our current iterate $\mb z^{(r)} \in \mc R_2 \cap \mc R_1^c \cap \Gamma$, and our trust-region size satisfies
	\begin{align}
		\Delta \le \min \set{\frac{\norm{\mb x}{}^3}{8000L_g}, \sqrt{\frac{3\norm{\mb x}{}^3}{16000L_h}}}. 
	\end{align}
Then an optimizer $\mb \delta_\star$ to~\eqref{eqn:trm-subproblem} leads to $\mb z^{(r+1)} = \mb z^{(r)} + \mb \delta_\star$ that obeys 
	\begin{align}
		f(\mb z^{(r+1)}) - f(\mb z^{(r)}) \le - d_2 \doteq -\frac{1}{4000} \Delta \norm{\mb x}{}^3. 
	\end{align}
\end{proposition}

\begin{proof}
	See Section~\ref{pf:prop:large-gradient} on Page~\pageref{pf:prop:large-gradient}.
\end{proof}

Now, we argue about $\mc R_3$, in which the behavior of the algorithm is more complicated. For the region $\mc R_3\setminus \mc R_3'$, the restricted strong convexity in radial directions around $\mc X$ as established in Proposition~\ref{prop:str_cvx} implies that the gradient at any point in $\mc R_3\setminus \mc R_3'$ is nonzero. Thus, one can  treat this as another strong gradient region, and carry out essentially the same argument as in Proposition~\ref{prop:large-gradient}. 
\begin{proposition}[Function Value Decrease in $\mc R_3\setminus \mc R_3'$] \label{prop:func-value-decrease-R-3}
Suppose our current iterate $\mb z^{(r)}\in \mc R_3 \setminus \mc R_3'$, and our trust-region size satisfies
	\begin{align}
		\Delta \le \min \set{\frac{\norm{\mb x}{}^4}{160L_h L_g}, \sqrt{\frac{3}{320}} \frac{\norm{\mb x}{}^2}{L_h}}.  
   \end{align}
Then an optimizer $\mb \delta_\star$ to~\eqref{eqn:trm-subproblem} leads to $\mb z^{(r+1)} = \mb z^{(r)} + \mb \delta_\star$ that obeys 
	\begin{align}
		f(\mb z^{(r+1)}) - f(\mb z^{(r)}) \le - d_3 \doteq -\frac{1}{80L_h} \Delta \norm{\mb x}{}^4. 
	\end{align}
\end{proposition}

\begin{proof}
	See Section~\ref{pf:prop:func-value-decrease-R-3} on Page~\pageref{pf:prop:func-value-decrease-R-3}.
\end{proof}

Our next several propositions show that when the iterate sequence finally moves into $\mc R_3'$, \js{in general it makes a finite number of consecutive constrained steps in which the trust-region constraints are always active, followed by ultimate consecutive unconstrained steps in which the the trust-region constraints are always inactive until convergence. Depending on the initialization and optimization parameters, either constrained or unconstrained steps can be void}. The next proposition shows that when $\Delta$ is chosen properly, a constrained step in $\mc R_3'$ decreases the objective by a concrete amount. 
\begin{proposition}\label{prop:fix-decrease-R-3'}
Suppose our current iterate $\mb z^{(r)} \in \mc R_3'$, and the trust-region subproblem takes a constrained step, i.e., the optimizer to~\eqref{eqn:trm-subproblem} satisfies $\norm{\mb \delta_\star}{} = \Delta$. We have the $\mb \delta_\star$ leads to 
	\begin{align}
		f(\mb z^{(r+1)}) - f(\mb z^{(r)}) \leq -d_4 \doteq -\frac{m_H^2 \Delta^2}{4M_H}. 
	\end{align}
provided that 
\begin{align}
\Delta \le m_H^2/(4M_H L_h). 
\end{align}
Here $m_H$ and $M_H$ are as defined in Lemma~\ref{lem:hessian-lower-upper}. 
\end{proposition}

\begin{proof}
	See Section~\ref{pf:prop:fix-decrease-R-3'} on Page~\pageref{pf:prop:fix-decrease-R-3'}.
\end{proof}

The next proposition shows that when $\Delta$ is properly tuned, an unconstrained step in $\mc R_3'$ dramatically reduces the norm of the gradient. 

\begin{proposition}[Quadratic Convergence of the Norm of the Gradient] \label{prop:grad-quad-conv}
Suppose our current iterate $\mb z^{(r)} \in \mc R_3'$, and the trust-region subproblem takes an unconstrained step, i.e., the unique optimizer to~\eqref{eqn:trm-subproblem} satisfies $\norm{\mb \delta_\star}{} < \Delta$. We have the $\mb \delta_\star$ leads to $\mb z^{(r+1)} = \mb z^{(r)} + \mb \delta_\star$ that obeys  
\begin{align}
	\|\nabla f(\mb z^{(r+1)})\| \le \frac{1}{m_H^2} \paren{L_h + \frac{32}{\norm{\mb x}{}} M_H} \| \nabla f(\mb z^{(r)})\|^2, 
\end{align}
provided 
\begin{align}
\Delta \le \norm{\mb x}{}/10. 
\end{align}
Here $M_H$ and $m_H$ are as defined in Lemma~\ref{lem:hessian-lower-upper}.  
\end{proposition}

\begin{proof}
	See Section~\ref{pf:prop:grad-quad-conv} on Page~\pageref{pf:prop:grad-quad-conv}.
\end{proof}

The next proposition shows that when $\Delta$ is properly tuned, as soon as an unconstrained $\mc R_3'$ step is taken, all future iterations take unconstrained $\mc R_3'$ steps. Moreover, the sequence converges quadratically to the target set $\mc X$. 
\begin{proposition}[Quadratic Convergence of the Iterates in $\mc R_3'$] \label{prop:unconstraint-steps}
Suppose the trust-region algorithm starts to take an unconstrained step in $\mc R_3'$ at $\mb z^{(r)}$ for a certain $r \in \N$. Then all future steps will be unconstrained steps in $\mc R_3'$, and 
\begin{align}
\norm{\mb h(\mb z^{(r + r')})}{} \le \frac{4\sqrt{2} m_H^2}{\norm{\mb x}{}^2} \left(L_h + \frac{32}{\norm{\mb x}{}} M_H \right)^{-1} 2^{-2^{r'}}
\end{align}
for all integers $r' \ge 1$, provided that 
\begin{align}
	\Delta \le \min\set{\frac{\norm{\mb x}{}}{10}, \frac{m_H\norm{\mb x}{}^2}{M_H \sqrt{40\sqrt{2} L_h(L_h + 32M_H/\norm{\mb x}{})}}, \frac{m_H^3}{\sqrt{2}M_H^2 (L_h + 32M_H/\norm{\mb x}{})}}. 
\end{align}	
\end{proposition}
\begin{proof}
	See Section~\ref{pf:prop:unconstraint-steps} on Page~\pageref{pf:prop:unconstraint-steps}.
\end{proof}

Now we are ready to piece together the above technical propositions to prove our main algorithmic theorem.
\begin{theorem}[TRM Convergence] \label{thm:TRM-conv}
Suppose $m \ge Cn\log^3 n$ for a sufficiently large constant $C$. Then with probability at least $1 - c_am^{-1}$, the trust-region algorithm with an \emph{arbitrary initialization} $\mb z^{(0)} \in \bb{CB}^n(R_0)$, where $R_0 = 3(\frac{1}{m}\sum_{k=1}^m y_k^2)^{1/2}$, will return a solution that is $\eps$-close to the target set $\mc X$ in 
\begin{align}
\frac{c_b}{\Delta^2 \norm{\mb x}{}^2} f(\mb z^{(0)}) + \log \log \paren{\frac{c_c \norm{\mb x}{}}{\eps}}
\end{align}
steps, provided that 
\begin{align}
\Delta \le c_d (n^{7/2} \log^{7/2} m)^{-1} \norm{\mb x}{}. 
\end{align}
Here $c_a$ through $c_d$ are positive absolute constants. 
\end{theorem}
\begin{proof}
When $m \ge C_1 n \log^3 n$ for a sufficiently large constant $C_1$, the assumption of Theorem \ref{thm:finite-landscape} is satisfied. Moreover, with probability at least $1-c_2m^{-1}$, the following estimates hold: 
\begin{gather*}
L_f = C_3 n^{5/2} \log^{5/2} m \norm{\mb x}{}^3, \quad L_g = C_3 n^2 \log^2 m \norm{\mb x}{}^2, \quad L_h = C_3 n^{3/2} \log^{3/2} m \norm{\mb x}{}, \\
m_H = 22/25 \norm{\mb x}{}^2, \quad M_H  = 9/2 \norm{\mb x}{}^2  
\end{gather*}
for a certain positive absolute constant $C_3$. From the technical lemmas and propositions in Section~\ref{sec:conv_trm}, it can be verified that when 
\begin{align*}
\Delta \le c_4 (n^{7/2} \log^{7/2} m)^{-1}\norm{\mb x}{}, 
\end{align*}
for a positive absolute constant $c_4$, all requirements on $\Delta$ are satisfied. 

Write $\mc R_A \doteq \Gamma \setminus \mc R_3'$, where $\Gamma \doteq \bb{CB}^n(R_1)$ with $R_1 = 3 \sqrt{n \log m} R_0$. Then a step in $\Gamma$ is either a $\mc R_A$/constrained $\mc R_3'$ step that reduces the objective value by a concrete amount, or an unconstrained $\mc R_3'$ step with all subsequent steps being unconstrained $\mc R_3'$. From discussion in Section~\ref{sec:conv_overview}, for an arbitrary initialization $\mb z^{(0)} \in \Gamma$, our choice of $R_1$ ensures that w.h.p.\ the sublevel set $\Pi \doteq \set{\mb z: f(\mb z) \le f(\mb z^{(0)})}$ is contained in $\Gamma$. Since $\mc R_A$ and constrained $\mc R_3'$ steps reduce the objective function, their following iterates stay in $\Pi$, and hence also stay in $\Gamma$. Moreover, $\mc R_3' \subset \Gamma$ and unconstrained $\mc R_3'$ steps stay within $\mc R_3'$. Thus, the iterate sequence as a whole stays in $\Gamma$. 

In fact, the previous argument implies a generic iterate sequence consists of two phases: the first phase that takes consecutive $\mc R_A$ or constrained $\mc R_3'$ steps, and thereafter the second phase that takes consecutive unconstrained $\mc R_3'$ steps till convergence. Either of the two can be absent depending on the initialization and parameter setting for the TRM algorithm.  
 
Since $f \ge 0$, by Proposition~\ref{prop:negative-curvature}, ~\ref{prop:large-gradient}, ~\ref{prop:func-value-decrease-R-3}, and~\ref{prop:fix-decrease-R-3'}, from $\mb z^{(0)}$ it takes at most 
\begin{align*}
f(\mb z^{(0)}) / \min(d_1, d_2, d_3, d_4)
\end{align*} 
steps for the iterate sequence to enter $\mc R_3'$.\footnote{\js{It is possible to refine the argument a bit by proving that the sequence does not exit $\mc R_3'$ once entering it, in which case the bound can be tightened as $f(\mb z^{(0)}) /\min(d_1, d_2, d_3)$. We prefer to state this crude bound to save the additional technicality. }}
Let $r_0$ denote the index of the first iteration for which $\mb z^{(r_0)} \in \mc R_3'$. Once the sequence enters $\mc R_3'$, there are three possibilities: 
\begin{itemize}
	\item The sequence always takes constrained steps in $\mc R_3'$ and since the function $f(\mb z)$ is lower bounded ($\ge 0$), it reaches the target set $\mc X$ in finitely many steps.
	\item The sequence takes constrained steps until reaching certain point $\mb z' \in \mc R_3'$ such that $f(\mb z') \leq f(\mb x) + d_4 = d_4$, where $d_4$ is defined in Proposition \ref{prop:fix-decrease-R-3'}. Since a constrained step in $\mc R_3'$ must decrease the function value by at least $d_4$, all future steps must be unconstrained. Proposition \ref{prop:unconstraint-steps} suggests that the sequence will converge quadratically to the target set $\mc X$.
	\item The sequence starts to take unconstrained steps at a certain point $\mb z'' \in \mc R_3'$ such that $f(\mb z'')\geq f(\mb x) + d_4$. Again Proposition \ref{prop:unconstraint-steps} implies that the sequence will converge quadratically to the target set $\mc X$.
\end{itemize}
In sum, by Proposition~\ref{prop:negative-curvature}, Proposition~\ref{prop:large-gradient}, Proposition~\ref{prop:func-value-decrease-R-3}, Proposition~\ref{prop:fix-decrease-R-3'}, and Proposition~\ref{prop:unconstraint-steps}, the number of iterations to obtain an $\eps$-close solution to the target set $\mc X$ can be grossly bounded by
\begin{align*}
\#\text{Iter}\;&\leq\; \frac{f(\mb z^{(0)}) }{\min\Brac{d_1,d_2,d_3,d_4 } }+ \log\log \paren{ \frac{4\sqrt{2}m_H^2 }{  (L_h + 32M_H/\norm{\mb x}{}) \norm{\mb x}{}^2\eps } }. 
\end{align*}
Using our previous estimates of $m_H$, $M_H$, and $L_H$, and taking $\min\{d_1, d_2, d_3, d_4\} = c_5 \Delta^2 \norm{\mb x}{}^2$, we arrive at the claimed result. 
\end{proof}
\js{
\begin{remark}[On Stability of Our Analysis and Total Complexity] \label{rmk:complexity}
Our above results are conditioned on exact arithmetic for solving the trust-region subproblem. We first discuss the effect of inexact arithmetics. Fix an $\eps$ to be determined later, suppose our trust-region subproblem solver always returns a $\wh{\mb w}$ such that $Q(\wh{\mb w}) - Q(\mb w_\star) \le \eps$ -- based on our discussion in Section~\ref{sec:trust-region-algorithm}, such $\wh{\mb w}$ can always be found in $O(n^3 \log(1/\eps))$ time. Then, in Lemma~\ref{lem:func-decent}, the trust-region subproblem returns a feasible point $\norm{\mb \delta_\star}{}$ such that 
\begin{align*}
f\paren{\mb z + \mb \delta_\star} \le f\paren{\mb z} -d + \frac{2}{3} L_h \Delta^3 + \eps. 
\end{align*}
Accordingly, the $d_i$ for $i \in [4]$ in Proposition~\ref{prop:negative-curvature} through Proposition~\ref{prop:fix-decrease-R-3'} are changed to $d_i - \eps$, with other conditions unaltered. Thus, combining the above with arguments in Theorem~\ref{thm:TRM-conv}, we can take 
\begin{align*}
\eps = \frac{1}{2} \min\paren{d_1, d_2, d_3, d_4} = c \Delta^2 \norm{\mb x}{}^2,  
\end{align*} 
such that from an initialization $\mb z^{(0)}$, the iterate sequence takes at most 
\begin{align*}
\frac{2f\paren{\mb z^{(0)}}}{\min\paren{d_1, d_2, d_3, d_4}} = \frac{2cf\paren{\mb z^{(0)}}}{\Delta^2 \norm{\mb x}{}^2}
\end{align*}
steps to stay in region $\mc R_3'$ and possibly to start consecutive unconstrained steps. For the unconstrained steps in $\mc R_3'$, since exact arithmetic is possible as we discussed in Section~\ref{sec:trust-region-algorithm}, the results in Proposition~\ref{prop:grad-quad-conv} and Proposition~\ref{prop:unconstraint-steps} are intact. The step estimate for this part in Theorem~\ref{thm:TRM-conv} remains valid. Overall, by our above choice, inexact arithmetics at most double the number of iterations to attain an $\eps$-near solution to $\mc X$. 

If we set $\Delta = c n^{-7/2} \log^{-7/2} m \norm{\mb x}{}$, then each trust-region iteration costs $O(n^3 \log\paren{n \norm{\mb x}{} \log m})$. Moreover, it takes 
\begin{align*}
O\paren{\frac{1}{\norm{\mb x}{}^4} n^7 \log^7 m + \log\log\paren{\frac{\norm{\mb x}{}}{\eps}}}
\end{align*}
iterations to arrive at an $\eps$-near solution to $\mc X$.

\end{remark}
}

%% file: sec/exp.tex

\section{Numerical Simulations}
Our convergence analysis for the TRM is based on two idealizations: (i) the trust-region subproblem is solved exactly; and (ii) the step-size is fixed to be sufficiently small. These simplifications ease the analysis, but also render the TRM algorithm impractical. In practice, the trust-region subproblem is never exactly solved, and the trust-region step size is adjusted to the local geometry, by backtracking. It is relatively straightforward to modify our analysis to account for inexact subproblem solvers; for sake of brevity, we do not pursue this here\footnote{\js{The proof ideas are contained in Chap 6 of~\cite{conn2000trust}; see also~\cite{absil2009}. Intuitively, such result is possible because reasonably good approximate solutions to the TRM subproblem make qualitatively similar progress as the exact solution. Recent work~\cite{cartis2012complexity,boumal2016global} has established worst-case polynomial iteration complexity (under reasonable assumptions on the geometric parameters of the functions, of course) of TRM to converge to point verifying the second-order optimality conditions. Their results allow inexact trust-region subproblem solvers, as well as adaptive step sizes. Based on our geometric result, we could have directly called their results, producing slightly worse iteration complexity bounds. It is not hard to adapt their proof taking advantage of the stronger geometric property we established and produce tighter results.} }. 


In this section, we investigate experimentally the number of measurements $m$ required to ensure that $f(\mb z)$ is well-structured, in the sense of our theorems. This entails solving large instances of $f(\mb z)$. To this end, we deploy the Manopt toolbox~\cite{boumal2014manopt}\footnote{Available online: \url{http://www.manopt.org}. }. Manopt is a user-friendly Matlab toolbox that implements several sophisticated solvers for tackling optimization problems on Riemannian manifolds. The most developed solver is based on the TRM. This solver uses the truncated conjugate gradient (tCG; see, e.g., Section 7.5.4 of~\cite{conn2000trust}) method to (approximately) solve the trust-region subproblem (vs.\ the exact solver in our analysis). It also dynamically adjusts the step size. However, the original implementation (Manopt 2.0) is not adequate for our purposes. Their tCG solver uses the gradient as the initial search direction, which does not ensure that the TRM solver can escape from saddle points~\cite{absil2007trust,absil2009}. We modify the tCG solver, such that when the current gradient is small and there is a negative curvature direction (i.e., the current point is near a saddle point or a local maximizer for $f(\mb z)$), the tCG solver explicitly uses the negative curvature direction\footnote{...adjusted in sign to ensure positive correlation with the gradient -- if it does not vanish.} as the initial search direction. This modification\footnote{\js{Similar modification is also adopted in the TRM algorithmic framework in the recent work~\cite{boumal2016global} (Algorithm 3).} } ensures the TRM solver always escapes saddle points/local maximizers with directional negative curvature. Hence, the modified TRM algorithm based on Manopt is expected to have the same qualitative behavior as the idealized version we analyzed. 

We fix $n = 1,000$ and vary the ratio $m/n$ from $4$ to $10$. 
\begin{figure}[!htbp] 
		\centering
		\begin{subfigure}{0.5\textwidth}
		\includegraphics[width=\linewidth]{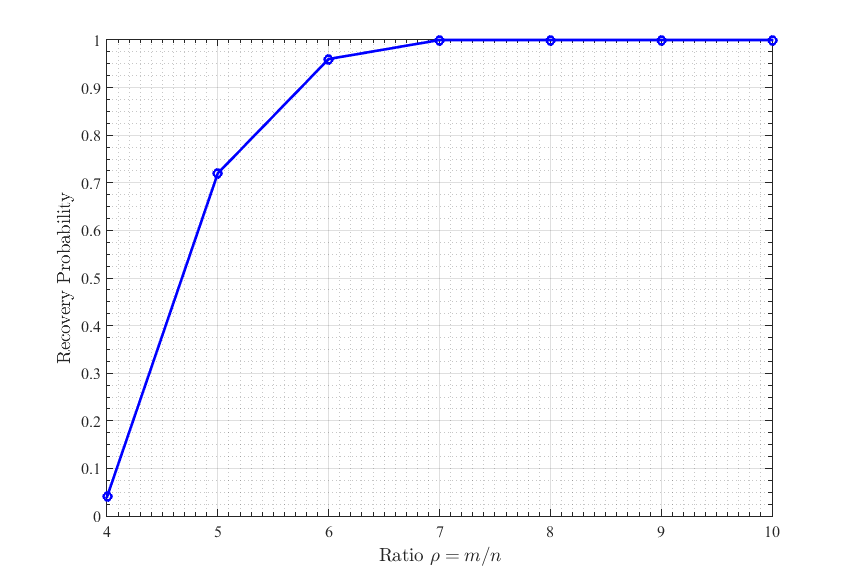} 
		\end{subfigure}
		\begin{subfigure}{0.4\textwidth}
		\includegraphics[width=\linewidth]{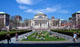}
		\end{subfigure}
    \caption{(Left) Recovery performance for GPR when optimizing~\eqref{eqn:finite-f} with the TRM. With $n = 1000$ and $m$ varying,  we consider a fixed problem instance for each $m$, and run the TRM algorithm $25$ times from independently random initializations. The empirical recovery probability is a test of whether the benign geometric structure holds. (Right) A small ``artistic'' Columbia University campus image we use for comparing TRM and gradient descent. }
    \label{fig:m_n_scaling}
\end{figure}
For each $m$, we generate a fixed instance: a fixed signal $\mb x$, and a fixed set of complex Gaussian vectors. We run the TRM algorithm $25$ times for each problem instance, with independent random initializations. Successfully recovery is declared if at termination the optimization variable $\mb z_{\infty}$ satisfies 
\begin{align*}
\eps_{\mathrm{Rel}} \doteq \|\mb z_{\infty} - \mb x \e^{\im \phi(\mb z_{\infty})}\| /\norm{\mb x}{} \le 10^{-3}. 
\end{align*}
The recovery probability is empirically estimated from the $25$ repetitions for each $m$. Intuitively, when the recovery probability is below one, there are spurious local minimizers. In this case, the number of samples $m$ is not large enough to ensure the finite-sample function landscape $f(\mb z)$ to be qualitatively the same as the asymptotic version $\bb E_{\mb a}[f(\mb z)]$. Figure~\ref{fig:m_n_scaling} shows the recovery performance. It seems that $m = 7n$ samples may be sufficient to ensure the geometric property holds.\footnote{This prescription should be taken with a grain of salt, as here we have only tested a single fixed $n$.} On the other hand, $m = 6n$ is not sufficient, whereas in theory it is known $4n$ samples are enough to guarantee measurement injectivity for complex signals~\cite{balan2006signal}.\footnote{Numerics in~\cite{chen2015solving} suggest that under the same measurement model, $m = 5n$ is sufficient for efficient recovery. Our requirement on control of the whole function landscape and hence ``initialization-free" algorithm may need the additional complexity.} 

We now briefly compare TRM and gradient descent in terms of running time. We take a small ($n = 80 \times 47$) image of Columbia University campus (Figure~\ref{fig:m_n_scaling} (Right)), and make $m = 5n \log n$ complex Gaussian measurements. The TRM solver is the same as above, and the gradient descent solver is one with backtracking line search. We repeat the experiment $10$ times, with independently generated random measurements and initializations each time. On average, the TRM solver returns a solution with $\eps_{\mathrm{Rel}} \le 10^{-4}$ in about 2600 seconds, while the gradient descent solver produces a solution with $\eps_{\mathrm{Rel}} \sim 10^{-2}$ in about 6400 seconds. The point here is not to exhaustively benchmark the two -- they both involve many implementation details and tuning parameters and they have very different memory requirements. It is just to suggest that second-order methods can be implemented in a practical manner for large-scale GPR problems.\footnote{The main limitation in this experiment was not the TRM solver, but the need to store the vectors $\mb a_1, \dots \mb a_m$. For other measurement models, such as the coded diffraction model~\cite{candes2015diffraction}, ``matrix-free'' calculation is possible, and storage is no longer a bottleneck.}

%% file: sec/discussion.tex

\section{Discussion}
In this work, we provide a complete geometric characterization of the nonconvex formulation~\eqref{eqn:finite-f} for the GPR problem. The benign geometric structure allows us to design a second-order trust-region algorithm that efficiently finds a global minimizer of~\eqref{eqn:finite-f}, without special initializations. We close this paper by discussing possible extensions and relevant open problems. 

\paragraph{Sample complexity and measurement schemes.} Our result (Theorem~\ref{thm:finite-landscape} and Theorem~\ref{thm:TRM-conv}) indicates that $m \ge C_1 n\log^3(n)$ samples are sufficient to guarantee the favorable geometric property and efficient recovery, while our simulations suggested that $C_2n\log(n)$ or even $C_3n$ is enough. For efficient recovery only, $m \ge C_4n$ are known to be sufficient~\cite{chen2015solving} (and for all signals; see also~\cite{candes2015phase,wang2016solving,zhang2016reshaped}). It is interesting to see if the gaps can be closed. Our current analysis pertains to Gaussian measurements only which are not practical, it is important to extend the geometric analysis to more practical measurement schemes, such as t-designs~\cite{david2013partial} and masked Fourier transform measurements~\cite{candes2015diffraction}. A preliminary study of the low-dimensional function landscape for the latter scheme (\js{for reduced real version}) produces very positive result; see Figure~\ref{fig:geo_diffraction}.  
\begin{figure}[!htbp]
    \centering
    \includegraphics[width=0.8\linewidth]{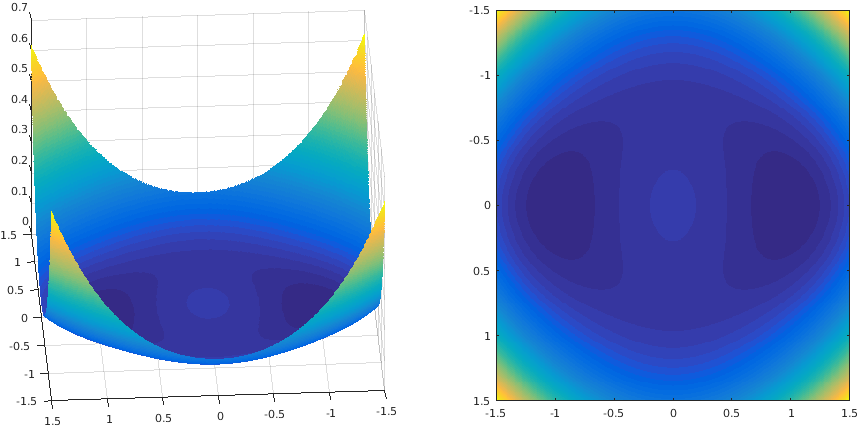}
    \caption{Function landscape of~\eqref{eqn:finite-f} for $\mb x =[1; 0]$ and $m \to \infty$ for \js{the real-value-masked discrete cosine transform measurements (i.e., real-valued version of the coded diffraction model~\cite{candes2015diffraction}). The mask takes i.i.d. values from $\set{1, 0, -1}$; each entry takes $1$ or $-1$ with probability $1/4$ respectively, and takes $0$ with probability $1/2$.} The landscape is qualitatively similar to that for the Gaussian model (Figure~\ref{fig:geo-2d}).}
    \label{fig:geo_diffraction}
\end{figure}

\paragraph{Sparse phase retrieval.} A special case of GPR is when the underlying signal $\mb x$ is known to be sparse, which can be considered as a quadratic compressed sensing problem \cite{ohlsson2013basis,ohlsson2013quadratic,ohlsson2012CPRL,li2013sparse,jaganathan2013sparse,shechtman2014GESPAR}.  Since $\mb x$ is sparse, the lifted matrix $\mb X = \mb x\mb x^*$ is sparse and has rank one. Thus, existing convex relaxation methods~\cite{ohlsson2013basis,ohlsson2013quadratic,li2013sparse,jaganathan2013sparse} formulated it as a simultaneously low-rank and sparse recovery problem. For the latter problem, however, known convex relaxations are suboptimal~\cite{oymak2012simultaneously,MuHWG14}. Let $k$ be the number of nonzeros in the target signal. \cite{li2013sparse,jaganathan2013sparse} showed that natural convex relaxations require $C_5k^2\log n$ samples for correct recovery, instead of the optimal order $O(k\log(n/k)$. A similar gap is also observed with certain nonconvex methods~\cite{cai2015optimal}. It is tempting to ask whether novel nonconvex formulations and analogous geometric analysis as taken here could shed light on this problem. 

\paragraph{Other structured nonconvex problems.} We have mentioned recent surge of works on provable nonconvex heuristics
\cite{jain2013low,hardt2014understanding,hardt2014fast,netrapalli2014non,jain2014fast,sun2014guaranteed, jain2014provable, wei2015guarantees,sa2015global,zheng2015convergent,tu2015low,chen2015fast, anandkumar2014analyzing,anandkumar2014guaranteed,anandkumar2015tensor,ge2015escaping,qu2014finding,hopkins2015speeding,  agarwal2013learning,arora2013new,agarwal2013exact,arora2014more,arora2015simple,sun2015complete, yi2013alternating,sedghi2014provable, lee2013near, lee2015rip, lee2015blind, eftekhari2015greed, boumal2016nonconvex, jain2015computing}. While the initialization plus local refinement analyses generally produce interesting theoretical results, they do not explain certain empirical successes that do not rely on special initializations. The geometric structure and analysis we work with in our recent work~\cite{sun2015complete,sun2015nonconvex} (see also~\cite{ge2015escaping, anandkumar2016efficient}, and~\cite{kawaguchi2016deep,soudry2016no,bhojanapalli2016global,ge2016matrix,park2016non,bandeira2016low,boumal2016non}) seem promising in this regard. It is interesting to consider whether analogous geometric structure exists for other practical problems. 

%% file: sec/proof_finite.tex

\section{Proofs of Technical Results for Function Landscape}  \label{app:finite}

\subsection{Auxiliary Lemmas}\label{app:aux-landscape}

\js{Lemma~\ref{lem:expect-func} to~\ref{lem:op_norm_1} first appeared in~\cite{candes2015phase}; we include the full proofs here for completeness. }

\begin{lemma}\label{lem:expect-func}
	For the function $f(\mb z): \bb C^n \mapsto \bb R$ defined in \eqref{eqn:finite-f}, we have
	\begin{align}
		\expect{f(\mb z)} &= \norm{\mb x}{}^4 + \norm{\mb z}{}^4 - \norm{\mb x}{}^2 \norm{\mb z}{}^2 - \abs{\mb x^* \mb z}^2,  \label{eqn:expect-f} \\
		\nabla \expect{f(\mb z)} &= \begin{bmatrix}
			\nabla_{\mb z} \expect{f(\mb z)} \\ \nabla_{\ol{\mb z}} \expect{f(\mb z)}
		\end{bmatrix} = \begin{bmatrix}
			\paren{2\norm{\mb z}{}^2 \mb I - \norm{\mb x}{}^2 \mb I - \mb x \mb x^*} \mb z \\
			\paren{2\norm{\mb z}{}^2\mb I  - \norm{\mb x}{}^2 \mb I - \mb x \mb x^*} \ol{\mb z} 
		\end{bmatrix}, \label{eqn:expect-grad} \\
		\nabla^2 \expect{f(\mb z)} &= 
\begin{bmatrix}
2 \mb z \mb z^* - \mb x \mb x^* + \paren{2\norm{\mb z}{}^2 - \norm{\mb x}{}^2} \mb I & 2\mb z \mb z^\top  \\
2 \ol{\mb z} \mb z^* & 2 \ol{\mb z} \mb z^\top - \ol{\mb x} \mb x^\top + \paren{2\norm{\mb z}{}^2 - \norm{\mb x}{}^2} \mb I
\end{bmatrix}. \label{eqn:expect-hess}
	\end{align}
\end{lemma}

\begin{proof}
	By definition \eqref{eqn:finite-f}, notice that 
\begin{align*}
\expect{f(\mb z)} 
& = \frac{1}{2} \bb E_{\mb a \sim \mc {CN}} \brac{\paren{\abs{\innerprod{\mb a}{\mb x}}^2 - \abs{\innerprod{\mb a}{\mb z}}^2}^2} \\
& = \frac{1}{2} \bb E_{\mb a \sim \mc {CN}} \brac{\abs{\innerprod{\mb a}{\mb x}}^4} + \frac{1}{2} \bb E_{\mb a \sim \mc {CN}} \brac{\abs{\innerprod{\mb a}{\mb z}}^4} - \bb E_{\mb a \sim \mc {CN}} \brac{\abs{\innerprod{\mb a}{\mb x}}^2\abs{\innerprod{\mb a}{\mb z}}^2}. 
\end{align*}
We now evaluate the three terms separately. Note that the law $\mc {CN}$ is invariant to unitary transform. Thus, 
\begin{align*}
\bb E_{\mb a \sim \mc {CN}} \brac{\abs{\innerprod{\mb a}{\mb x}}^4} = \bb E_{\mb a \sim \mc {CN}} \brac{\abs{\innerprod{\mb a}{\mb e_1}}^4} \norm{\mb x}{}^4 = \bb E_{a \sim \mc N(0, 1/2) + \im \; \mc N(0, 1/2)} \brac{\abs{a}^4} \norm{\mb x}{}^4 = 2\norm{\mb x}{}^4. 
\end{align*}
Similarly, we also obtain $\bb E_{\mb a \sim \mc {CN}} \brac{\abs{\innerprod{\mb a}{\mb z}}^4} = 2\norm{\mb z}{}^4$. Now for the cross term, 
\begin{align*}
\MoveEqLeft \bb E_{\mb a \sim \mc {CN}} \brac{\abs{\innerprod{\mb a}{\mb x}}^2\abs{\innerprod{\mb a}{\mb z}}^2} \\
=\; & \bb E_{\mb a \sim \mc {CN}} \brac{\abs{\innerprod{\mb a}{\mb e_1}}^2\abs{\innerprod{\mb a}{s_1 \e^{\im \phi_1} \mb e_1 + s_2 \e^{\im \phi_2} \mb e_2}}^2} \norm{\mb x}{}^2 \norm{\mb z}{}^2 \quad [\text{where $s_1^2 + s_2^2 = 1$}] \\
=\; & \bb E_{\mb a \sim \mc {CN}} \brac{\abs{a_1}^2\abs{s_1 \ol{a_1} \e^{\im \phi_1} + s_2 \ol{a_2} \e^{\im \phi_2}}^2} \norm{\mb x}{}^2 \norm{\mb z}{}^2\\
=\; & \bb E_{\mb a \sim \mc {CN}} \brac{\abs{a_1}^2\paren{s_1^2 \abs{a_1}^2 + s_2^2 \abs{a_2}^2}} \norm{\mb x}{}^2 \norm{\mb z}{}^2 \\
=\; & \paren{1+ s_1^2} \norm{\mb x}{}^2 \norm{\mb z}{}^2 = \norm{\mb x}{}^2 \norm{\mb z}{}^2  + \abs{\mb x^* \mb z}^2. 
\end{align*}
Gathering the above results, we obtain \eqref{eqn:expect-f}. By taking Wirtinger derivative \eqref{eqn:wirtinger} with respect to \eqref{eqn:expect-f}, we obtain the Wirtinger gradient and Hessian in \eqref{eqn:expect-grad}, \eqref{eqn:expect-hess} as desired.	
\end{proof}

\begin{lemma} \label{lem:fourth_exp}
For $\mb a \sim \mc{CN}(n)$ and any fixed vector $\mb v \in \Cp^n$, it holds that 
\begin{align*}
\expect{\abs{\mb a^* \mb v}^2 \mb a \mb a^*} = \mb v \mb v^* + \norm{\mb v}{}^2 \mb I,  \quad \text{and} \quad \expect{\paren{\mb a^* \mb v}^2 \mb a \mb a^\top} = 2 \mb v \mb v^\top. 
\end{align*}
\end{lemma}
\begin{proof}
Observe that for $i \ne j$, 
\begin{align*}
\mb e_i^* \expect{\abs{\mb a^* \mb v}^2 \mb a \mb a^*} \mb e_j = \sum_{q, \ell} \expect{\ol{a(q)} a(\ell) v(q) \ol{v(\ell)} a(i) \ol{a(j)}} = \expect{\abs{a(i)}^2 \abs{a(j)}^2} v(i) \ol{v(j)} =  v(i) \ol{v(j)}. 
\end{align*}
Similarly, 
\begin{align*}
\mb e_i^* \expect{\abs{\mb a^* \mb v}^2 \mb a \mb a^*} \mb e_i 
& =   \sum_{q, \ell} \expect{\ol{a(q)} a(\ell) v(q) \ol{v(\ell)} \abs{a(i)}^2} \\
& = \expect{\abs{a(i)}^4 \abs{v(i)}^2} + \sum_{q \ne i}\expect{\abs{a(q)}^2 \abs{v(q)}^2 \abs{a(i)}^2} = \abs{v(i)}^2 + \norm{\mb v}{}^2. 
\end{align*}
Similar calculation yields the second expectation. 
\end{proof}

\begin{lemma}  \label{lem:op_norm_1}
Let $\mb a_1, \dots, \mb a_m$ be i.i.d. copies of $\mb a \sim \mc{CN}(n)$. For any $\delta \in (0, 1)$ and any $\mb v \in \Cp^n$, when $m \ge C(\delta) n\log n$, we have that with probability at least $1 - c_a \delta^{-2} m^{-1} - c_b\exp\paren{-c_c \delta^2 m/\log m}$ 
\begin{align*}
\norm{\frac{1}{m}\sum_{k=1}^m \abs{\mb a_k^* \mb v}^2 \mb a_k \mb a_k^* - \paren{\mb v \mb v^* + \norm{\mb v}{}^2 \mb I}}{} & \le \delta \norm{\mb v}{}^2, \\
\norm{\frac{1}{m}\sum_{k=1}^m \paren{\mb a_k^* \mb v}^2 \mb a_k \mb a_k^\top - 2 \mb v \mb v^\top}{} & \le \delta \norm{\mb v}{}^2. 
\end{align*}
Here $C(\delta)$ is a constant depending on $\delta$ and $c_a$, $c_b$ and $c_c$ are positive absolute constants. 
\end{lemma}
\begin{proof}
We work out the results on $\frac{1}{m}\sum_{k=1}^m \abs{\mb a_k^* \mb v}^2 \mb a_k \mb a_k^*$ first. By the unitary invariance of the Gaussian measure and rescaling, it is enough to consider $\mb v = \mb e_1$. We partition each vector $\mb a_k$ as $\mb a_k = [a_k(1); \wt{\mb a}_k]$ and upper bound the target quantity as: 
\begin{align*}
\MoveEqLeft \norm{\frac{1}{m}\sum_{k=1}^m \abs{a_k(1)}^2 
\begin{bmatrix}
\abs{a_k(1)}^2  & a_k(1) \wt{\mb a}_k^* \\
\ol{a_k(1)} \wt{\mb a}_k & \wt{\mb a}_k \wt{\mb a}_k^*
\end{bmatrix}
- \paren{\mb e_1 \mb e_1^* + \mb I}}{}\\
\le\; & \abs{\frac{1}{m} \sum_{k=1}^m \paren{\abs{a_k(1)}^4 - 2}}{} + \norm{\frac{1}{m}\sum_{k=1}^m \abs{a_k(1)}^2 
\begin{bmatrix}
0 & a_k(1) \wt{\mb a}_k^* \\
\ol{a_k(1)} \wt{\mb a}_k & \mb 0
\end{bmatrix}
}{} \\
& \quad + \norm{\frac{1}{m}\sum_{k=1}^m \abs{a_k(1)}^2\paren{\wt{\mb a}_i \wt{\mb a}_k^* - \mb I_{n-1}}}{} + \abs{\frac{1}{m} \sum_{k=1}^m \paren{\abs{a_k(1)}^2 - 1}}{}. 
\end{align*}
By Chebyshev's inequality, we have with probability at least $1-c_1\delta^{-2}m^{-1}$, 
\begin{align*}
\abs{\frac{1}{m} \sum_{k=1}^m \paren{\abs{a_k(1)}^4 - 2}}{} \le \frac{\delta}{4} \quad \text{and} \quad \abs{\frac{1}{m} \sum_{k=1}^m \paren{\abs{a_k(1)}^2 - 1}}{} \le \frac{\delta}{4}. 
\end{align*}
To bound the second term, we note that 
\begin{align*}
\norm{\frac{1}{m}\sum_{k=1}^m \abs{a_k(1)}^2 
\begin{bmatrix}
0 & a_k(1) \wt{\mb a}_k^* \\
\ol{a_k(1)} \wt{\mb a}_k & \mb 0
\end{bmatrix}
}{} 
& =
\norm{\frac{1}{m}\sum_{k=1}^m \abs{a_k(1)}^2 a_k(1) \wt{\mb a}_k^*}{} \\
& = \sup_{\mb w \in \Cp^{n-1}: \norm{\mb w}{} = 1} \frac{1}{m}\sum_{k=1}^m \abs{a_k(1)}^2 a_k(1) \wt{\mb a}_k^*\mb w. 
\end{align*}
For all $\mb w$ and all $k\in [m]$ , $\wt{\mb a}_k^* \mb w$ is distributed as $\mc{CN}(1)$ that is independent of the $\{a_k(1)\}$ sequence. So for one realization of $\{a_k(1)\}$, the Hoeffding-type inequality of Lemma~\ref{lem:hoeffding_type} implies 
\begin{align*}
\prob{ \frac{1}{m}\sum_{k=1}^m \abs{a_k(1)}^2 a_k(1) \wt{\mb a}_k^*\mb w > t} \le \e \exp\paren{-\frac{c_2m^2 t^2}{\sum_{k=1}^m \abs{a_k(1)}^6}}, 
\end{align*}
for any $\mb w$ with $\norm{\mb w}{} = 1$ and any $t > 0$. Taking $t = \delta/8$, together with a union bound on a $1/2$-net on the sphere, we obtain 
\begin{align*} 
\prob{\norm{\frac{1}{m}\sum_{k=1}^m \abs{a_k(1)}^2 a_k(1) \wt{\mb a}_k^*}{} > \delta/4} \le \e \exp\paren{-\frac{c_2m^2\delta^2}{64\sum_{k=1}^m \abs{a_k(1)}^6} +12 (n-1) }. 
\end{align*}
Now an application of Chebyshev's inequality gives that 
$
\sum_{k=1}^m \abs{a_k(1)}^6 \le 20m
$
with probability at least $1 - c_3m^{-1}$. Substituting this into the above, we conclude that whenever $m \ge C_4\delta^{-2}n$ for some sufficiently large $C_4$, 
\begin{align*}
\norm{\frac{1}{m}\sum_{k=1}^m \abs{a_k(1)}^2 a_k(1) \wt{\mb a}_k^*}{} \le \delta/4
\end{align*}
with probability at least $1- c_3m^{-1} - \exp\paren{-c_5\delta^2 m}$. 

To bound the third term, we note that 
\begin{align*}
\norm{\frac{1}{m}\sum_{k=1}^m \abs{a_k(1)}^2 \paren{\wt{\mb a}_k \wt{\mb a}_k^* - \mb I_{n-1}}}{} = \sup_{\mb w \in \Cp^{n-1}: \norm{\mb w}{} = 1} \frac{1}{m}\sum_{k=1}^m \abs{a_k(1)}^2 \paren{\abs{\wt{\mb a}_k^* \mb w}^2 - 1}. 
\end{align*}
For all fixed $\mb w$ and all $k \in [m]$, $\wt{\mb a}_k^* \mb w \sim \mc{CN}(1)$. Thus, $\abs{\wt{\mb a}_k^* \mb w}^2 - 1$ is centered sub-exponential. So for one realization of $\{a_k(1)\}$, Bernstein's inequality (Lemma~\ref{lem:bernstein_type}) implies 
\begin{align*}
\prob{\frac{1}{m}\sum_{k=1}^m \abs{a_k(1)}^2 \paren{\abs{\wt{\mb a}_k^* \mb w}^2 - 1} > t} \le 2\exp\paren{-c_6\min\paren{\frac{t^2}{c_7^2 \sum_{k=1}^m \abs{a_k(1)}^4}, \frac{t}{c_7 \max_{i \in [m]} \abs{a_k(1)}^2}}}
\end{align*}
for any fixed $\mb w$ with $\norm{\mb w}{} = 1$ and any $t > 0$. Taking $t = \delta/8$, together with a union bound on a $1/2$-net on the sphere, we obtain 
\begin{multline*}
\prob{\norm{\frac{1}{m}\sum_{k=1}^m \abs{a_k(1)}^2 \paren{\wt{\mb a}_k \wt{\mb a_k}^* - \mb I_{n-1}}}{} > \frac{\delta}{4}} \\
\le 2\exp\paren{-c_6\min\paren{\frac{m^2\delta^2/64}{c_7^2 \sum_{k=1}^m \abs{a_k(1)}^4}, \frac{m\delta/8}{c_7 \max_{i \in [m]} \abs{a_k(1)}^2}} + 12(n-1)}. 
\end{multline*}
Chebyshev's inequality and the union bound give that 
\begin{align*}
\sum_{k=1}^m \abs{a_k(1)}^4 \le 10m , \quad \text{and} \quad \max_{i \in [m]} \abs{a_k(1)}^2 \le 10 \log m
\end{align*}
hold with probability at least $1-c_8 m^{-1} - m^{-4}$. To conclude, when $m \ge C_9(\delta) \delta^{-2} n \log n$ for some sufficiently large constant $C_9(\delta)$, 
\begin{align*}
\norm{\frac{1}{m}\sum_{k=1}^m \abs{a_k(1)}^2 \paren{\wt{\mb a}_k \wt{\mb a_k}^* - \mb I_{n-1}}}{} \le \frac{\delta}{4}
\end{align*}
with probability at least $1 - c_8 m^{-1} - m^{-4} - 2\exp\paren{-c_{10}\delta^2 m/\log m}$. 

Collecting the above bounds and probabilities yields the claimed results. Similar arguments prove the claim on $\frac{1}{m}\sum_{k=1}^m \paren{\mb a_k^* \mb v} \mb a_k \mb a_k^\top$ also, completing the proof. 
\end{proof}

\begin{lemma} \label{lem:min_fourth}
Let $\mb a_1, \dots, \mb a_m$ be i.i.d. copies of $\mb a \sim \mc{CN}(n)$. For any $\delta \in (0, 1)$, when $m \ge C(\delta)n \log n$, it holds with probability at least $1-c'\exp\paren{-c(\delta)m} - c''m^{-n}$ that 
\begin{align*}
\frac{1}{m}\sum_{k=1}^m \abs{\mb a_k^* \mb z}^2 \abs{\mb a_k^* \mb w}^2 
& \ge \paren{1-\delta} \paren{\norm{\mb w}{}^2 \norm{\mb z}{}^2 + \abs{\mb w^* \mb z}^2} \quad \text{for all} \;\mb z, \mb w \in \Cp^n, \\
\frac{1}{m}\sum_{k=1}^m \brac{\Re (\mb a_k^* \mb z)( \mb w^* \mb a_k) }^2 
& \ge  (1-\delta) \paren{\frac{1}{2} \norm{\mb z}{}^2 \norm{\mb w}{}^2 + \frac{3}{2}[\Re \mb z^* \mb w]^2 - \frac{1}{2} [\Im \mb z^* \mb w]^2}\quad \text{for all} \;\mb z, \mb w \in \Cp^n. 
\end{align*}
Here $C(\delta)$ and $c(\delta)$ are constants depending on $\delta$ and $c'$ and $c''$ are positive absolute constants. 
\end{lemma}
\begin{proof}
By Lemma~\ref{lem:fourth_exp}, $\expect{\abs{\mb a^* \mb w}^2 \abs{\mb a^* \mb z}^2} = \norm{\mb w}{}^2 \norm{\mb z}{}^2 + \abs{\mb w^* \mb z}^2$. By homogeneity, it is enough to prove the result for all $\mb w, \mb z \in \bb{CS}^{n-1}$. For a pair of fixed $\mb w, \mb z \in \bb{CS}^{n-1}$, Lemma~\ref{lem:subgauss_nonneg} implies that for any $\delta \in (0, 1)$, 
\begin{align*}
\sum_{k=1}^m \abs{\mb a_k^* \mb w}^2 \abs{\mb a_k^* \mb z}^2 \ge \paren{1 - \frac{\delta}{2}} m  \paren{1 + \abs{\mb w^* \mb z}^2}
\end{align*}
with probability at least $1 - \exp(-c_1 \delta^2 m)$. For a certain $\eps \in (0, 1)$ to be fixed later and an $\eps$-net $N_\eps^1 \times N_\eps^2$ that covers $\bb{CS}^{n-1} \times \bb{CS}^{n-1}$,  we have that the event 
\begin{align*}
\event_0 \doteq \Brac{\sum_{k=1}^m \abs{\mb a_k^* \mb w}^2\abs{\mb a_k^* \mb z}^2 \ge \paren{1 - \frac{\delta}{2}} m \paren{1 + \abs{\mb w^* \mb z}^2}\quad \forall \; \mb w, \mb z \in N_\eps^1 \times N_\eps^2} 
\end{align*}
holds with probability at least $1 - \exp\paren{-c_1\delta^2m + 4n\log(3/\eps)}$ by a simple union bound. Now conditioned on $\event_0$, we have for every $\mb z \in \bb{CS}^{n-1}$ can be written as $\mb z = \mb z_0 + \mb e$ for certain $\mb z_0 \in N_\eps^1$ and $\mb e$ with $\norm{\mb e}{} \le \eps$; similarly $\mb w = \mb w_0 + \mb \zeta$ for $\mb w_0 \in N_\eps^2$ and $\norm{\mb \zeta}{} \le \eps$. For the function $g(\mb w, \mb z) \doteq \sum_{k=1}^m \abs{\mb a_k^* \mb z}^2 \abs{\mb a_k^* \mb w}^2$, w.h.p., 
\begin{align*}
\norm{\frac{\partial g}{\partial \mb w}}{} & = \norm{\sum_{k=1}^m \abs{\mb a_k^* \mb z}^2 \mb w^* \mb a_k \ol{\mb a}_k}{} \le \norm{\mb z}{}^2 \norm{\mb w}{} \norm{\sum_{k=1}^m \norm{\mb a_k}{}^2 \mb a_k \mb a_k^*}{} \le 10mn\sqrt{\log m} , \\
\norm{\frac{\partial g}{\partial \mb z}}{} & = \norm{\sum_{k=1}^m \abs{\mb a_k^* \mb w}^2 \mb z^* \mb a_k \ol{\mb a}_k}{} \le \norm{\mb w}{}^2 \norm{\mb z}{}  \norm{\sum_{k=1}^m \norm{\mb a_k}{}^2 \mb a_k \mb a_k^*}{} \le 10mn\sqrt{\log m} , 
\end{align*}
as $\max_{k \in [m]} \norm{\mb a_k}{}^2 \le 5 n \log m$ with probability at least $1 - c_2 m^{-n}$, and $\norm{\sum_{k=1}^m \mb a_k \mb a_k^*}{} \le 2m$ with probability at least $1 - \exp(-c_3 m)$. Thus, 
\begin{align*}
\sum_{k=1}^m \abs{\mb a_k^* \mb z}^2 \abs{\mb a_k^* \mb w}^2
& \ge \paren{1 - \frac{\delta}{3}} m  - 40\eps mn \log m + \paren{1 - \frac{\delta}{3}} m\paren{\abs{\mb w_0^* \mb z_0}^2 - 4\eps}. 
\end{align*}
Taking $\eps = c_4(\delta)/(n\log m)$ for a sufficiently small $c_4(\delta) > 0$, we obtain that with probability at least $1 - \exp\paren{-c_1\delta^2m + 4n\log(3n\log m/c_4(\delta))} - c_5 m^{-n}$, 
\begin{align*}
\sum_{k=1}^m \abs{\mb a_k^* \mb z}^2 \abs{\mb a_k^* \mb w}^2 \ge \paren{1-\frac{2}{3}\delta} m\paren{1 + \abs{\mb w_0^* \mb z_0}^2}.
\end{align*}
which, together with continuity of the function $(\mb w, \mb z) \mapsto \abs{\mb w^* \mb z}^2$, implies 
\begin{align*}
\sum_{k=1}^m \abs{\mb a_k^* \mb z}^2 \abs{\mb a_k^* \mb w}^2 \ge \paren{1-\delta} m\paren{1 + \abs{\mb w^* \mb z}^2}.
\end{align*}
It is enough to take $m \ge C_6 \delta^{-2} n \log n$ to ensure the desired event happens w.h.p.. 

To show the second inequality, first notice that $\bb E \brac{\Re (\mb a_k^* \mb z)( \mb w^* \mb a_k) }^2 = \frac{1}{2} \norm{\mb z}{}^2 \norm{\mb w}{}^2 + \frac{3}{2}[\Re \mb z^* \mb w]^2 - \frac{1}{2} [\Im \mb z^* \mb w]^2$. The argument then proceeds to apply the discretization trick as above. 
\end{proof}

\subsection{Proof of Proposition~\ref{prop:nega-curv}} \label{pf:prop:nega-curv}
\begin{proof}
Direct calculation shows that
\begin{align*}
& \begin{bmatrix}
\mb x \e^{\im \phi(\mb z)}\\
\ol{\mb x} \e^{-\im \phi(\mb z)}
\end{bmatrix}^*   \nabla^2 f(\mb z) 
\begin{bmatrix}
\mb x \e^{\im \phi(\mb z)}\\
\ol{\mb x} \e^{-\im \phi(\mb z)}
\end{bmatrix}  \\
=\; & \frac{1}{m}\sum_{k=1}^m \paren{4\abs{\mb a_k^* \mb z}^2 \abs{\mb a_k^* \mb x}^2 - 2 \abs{\mb a_k^* \mb x}^4 + 2 \Re\brac{\paren{\mb a_k^* \mb z}^2 \paren{\mb x^* \mb a_k}^2 \e^{-2\im \phi(\mb z)}}} \\
=\; & \frac{1}{m}\sum_{k=1}^m \paren{2\abs{\mb a_k^* \mb z}^2 \abs{\mb a_k^* \mb x}^2 - 2 \abs{\mb a_k^* \mb x}^4} \\
& \qquad + \frac{1}{m}\sum_{k=1}^m \paren{2\abs{\mb a_k^* \mb z}^2 \abs{\mb a_k^* \mb x}^2 + 2 \Re\brac{\paren{\mb a_k^* \mb z}^2 \paren{\mb x^* \mb a_k}^2 \e^{-2\im \phi(\mb z)}}}. 
\end{align*}
Lemma~\ref{lem:op_norm_1} implies that when $m \ge C_1 n \log n$, w.h.p., 
\begin{align*}
\frac{2}{m} \sum_{k=1}^m \abs{\mb a_k^* \mb x}^2 \abs{\mb a_k^* \mb z}^2 \le 
\expect{\frac{2}{m} \sum_{k=1}^m \abs{\mb a_k^* \mb x}^2 \abs{\mb a_k^* \mb z}^2}+\frac{1}{200}\norm{\mb x}{}^2 \norm{\mb z}{}^2. 
\end{align*}
On the other hand, by Lemma~\ref{lem:subgauss_nonneg}, we have that 
\begin{align*}
\frac{2}{m}\sum_{k=1}^m \abs{\mb a_k^* \mb x}^4 \ge \expect{\frac{2}{m}\sum_{k=1}^m \abs{\mb a_k^* \mb x}^4} - \frac{1}{100}\norm{\mb x}{}^4
\end{align*}
holds with probability at least $1 - \exp(-c_2 m)$. For the second summation, we have 
\begin{align*}
& \frac{1}{m}\sum_{k=1}^m \paren{2\abs{\mb a_k^* \mb z}^2 \abs{\mb a_k^* \mb x}^2 + 2 \Re\brac{\paren{\mb a_k^* \mb z}^2 \paren{\mb x^* \mb a_k}^2 \e^{-2\im \phi(\mb z)}}} \\
=\; & 
\begin{bmatrix}
\mb z  \\ \ol{\mb z} 
\end{bmatrix}^* 
\nabla^2 f(\mb x \e^{\im \phi(\mb z)})
\begin{bmatrix}
\mb z \\ \ol{\mb z}
\end{bmatrix} \\
\le\; & 
\begin{bmatrix}
\mb z  \\ \ol{\mb z} 
\end{bmatrix}^*  
\nabla^2 \expect{f(\mb x\e^{\im \phi(\mb z)})} 
\begin{bmatrix}
\mb z \\ \ol{\mb z}
\end{bmatrix}
+ \frac{1}{200} \norm{\mb x}{}^2\norm{\mb z}{}^2, 
\end{align*}
w.h.p., provided $m \ge C_3 n \log n$, according to Lemma~\ref{lem:op_norm_1}. 

Collecting the above estimates, we have that when $m \ge C_4 n \log n$ for a sufficiently large constant $C_4$, w.h.p., 
\begin{align*}
\begin{bmatrix}
\mb x \e^{\im \phi(\mb z)}\\
\ol{\mb x} \e^{-\im \phi(\mb z)}
\end{bmatrix}^*   \nabla^2 f(\mb z) 
\begin{bmatrix}
\mb x \e^{\im \phi(\mb z)}\\
\ol{\mb x} \e^{-\im \phi(\mb z)}
\end{bmatrix}  
&\le 
\expect{\begin{bmatrix}
\mb x \e^{\im \phi(\mb z)}\\
\ol{\mb x} \e^{-\im \phi(\mb z)}
\end{bmatrix}^*   \nabla^2 f(\mb z) 
\begin{bmatrix}
\mb x \e^{\im \phi(\mb z)}\\
\ol{\mb x} \e^{-\im \phi(\mb z)}
\end{bmatrix} } + \frac{1}{100}\norm{\mb x}{}^2 \norm{\mb z}{}^2 + \frac{1}{100} \norm{\mb x}{}^4 \\
& \le -\frac{1}{100} \norm{\mb x}{}^4
\end{align*}
for all $\mb z \in \mc R_1$. Dividing both sides of the above by $\norm{\mb x}{}^2$ gives the claimed results. 
\end{proof}

\subsection{Proof of Proposition~\ref{prop:str_cvx}}  \label{pf:prop:str_cvx}
\begin{proof}
It is enough to prove that for all unit vectors $\mb g$ that are geometrically orthogonal to $\im \mb x$, i.e., $\mb g \in \mc T \doteq \set{\mb z: \Im\paren{\mb z^* \mb x} = 0, \norm{\mb z}{} =1}$ and all $t \in [0, \norm{\mb x}{}/\sqrt{7}]$, the following holds: 
\begin{align*}
\begin{bmatrix}
\mb g \\
\ol{\mb g}
\end{bmatrix}^* 
\nabla^2 f(\mb x + t\mb g) 
\begin{bmatrix}
\mb g \\
\ol{\mb g}
\end{bmatrix} 
\ge \frac{1}{4} \norm{\mb x}{}^2. 
\end{align*}
Direct calculation shows 
\begin{align*}
& \begin{bmatrix}
\mb g \\
\ol{\mb g}
\end{bmatrix}^* 
\nabla^2 f(\mb x + t\mb g) 
\begin{bmatrix}
\mb g \\
\ol{\mb g}
\end{bmatrix}  \\
=\; &  
\frac{1}{m}\sum_{k=1}^m 4\abs{\mb a_k^* (\mb x + t\mb g)}^2 \abs{\mb a_k^* \mb g}^2 - 2\abs{\mb a_k^* \mb x}^2 \abs{\mb a_k^* \mb g}^2 + 2\Re\brac{(t\mb a_k^* \mb g + \mb a_k^* \mb x)^2 (\mb g^* \mb a_k)^2} \\
\ge\; & \frac{1}{m}\sum_{k=1}^m 4\abs{\mb a_k^* (\mb x + t\mb g)}^2 \abs{\mb a_k^* \mb g}^2 - 2\abs{\mb a_k^* \mb x}^2 \abs{\mb a_k^* \mb g}^2 + 4\brac{\Re(t\mb a_k^* \mb g + \mb a_k^* \mb x)(\mb g^* \mb a_k)}^2 - 2 \abs{(t\mb a_k^* \mb g + \mb a_k^* \mb x)(\mb g^* \mb a_k)}^2 \\
\ge\; & \frac{1}{m}\sum_{k=1}^m 2\abs{\mb a_k^* (\mb x + t\mb g)}^2 \abs{\mb a_k^* \mb g}^2 - 2\abs{\mb a_k^* \mb x}^2 \abs{\mb a_k^* \mb g}^2 + 4\brac{\Re(t\mb a_k^* \mb g + \mb a_k^* \mb x)(\mb g^* \mb a_k)}^2. 
\end{align*}
Lemma~\ref{lem:min_fourth} implies when $m \ge C_1 n\log n$ for sufficiently large constant $C_1$, w.h.p., 
\begin{align}
\frac{1}{m}\sum_{k=1}^m 2\abs{\mb a_k^* (\mb x + t\mb g)}^2 \abs{\mb a_k^* \mb g}^2 \ge \frac{199}{100} \abs{(\mb x + t\mb g)^* \mb g}^2 + \frac{199}{100} \norm{\mb x+ t\mb g}{}^2 \norm{\mb g}{}^2
\end{align}
for all $\mb g \in \Cp^n$ and all $t \in [0, \norm{\mb x}{}/\sqrt{7}]$. Lemma~\ref{lem:op_norm_1} implies that when $m \ge C_2 n \log n$ for sufficiently large constant $C_2$, w.h.p., 
\begin{align}
\frac{1}{m}\sum_{k=1}^m 2\abs{\mb a_k^* \mb x}^2 \abs{\mb a_k^* \mb g}^2 \le \frac{201}{100} \abs{\mb x^* \mb g}^2 + \frac{201}{100} \norm{\mb x}{}^2 \norm{\mb g}{}^2
\end{align}
for all $\mb g \in \Cp^n$. Moreover, Lemma~\ref{lem:min_fourth} implies when $m \ge C_3 n \log n$ for sufficiently large constant $C_3$, w.h.p., 
\begin{align*}
\frac{4}{m}\sum_{k=1}^m \brac{\Re(t\mb a_k^* \mb g + \mb a_k^* \mb x)(\mb g^* \mb a_k)}^2 \ge 2\norm{\mb x + t\mb g}{}^2 \norm{\mb g}{}^2 + 6 \abs{(\mb x + \mb g)^* \mb g}^2 - \frac{1}{400} \norm{\mb x}{}^2 \norm{\mb g}{}^2
\end{align*}
for all $\mb g \in \mc T$, where we have used that $\Im(\mb g^* \mb x) = 0 \Longrightarrow \Im (\mb x+ \mb g)^* \mb g = 0$ to simplify the results. 

Collecting the above estimates, we obtain that when $m \ge C_4 n \log n$, w.h.p., 
\begin{align*}
& \begin{bmatrix}
\mb g \\
\ol{\mb g}
\end{bmatrix}^* 
\nabla^2 f(\mb x + t\mb g) 
\begin{bmatrix}
\mb g \\
\ol{\mb g}
\end{bmatrix}  \\
\ge\; & 
\paren{\frac{399}{100} \norm{\mb x + t\mb g}{}^2 - \frac{161}{80} \norm{\mb x}{}^2} + \paren{\frac{799}{100} \abs{(\mb x + t\mb g)^* \mb g}^2 - \frac{201}{100} \abs{\mb x^* \mb g}^2 }  \\
=\; & \frac{791}{400} \norm{\mb x}{}^2 + \frac{598}{100} \abs{\mb x^* \mb g}^2 +  \frac{1198}{100}t^2 + \frac{2396}{100} t \Re(\mb x^* \mb g). 
\end{align*}
To provide a lower bound for the above, we let $\Re(\mb x^* \mb g) = \mb x^* \mb g = \lambda \norm{\mb x}{}$ with $\lambda \in [-1, 1]$ and $t = \eta \norm{\mb x}{}$ with $\eta \in [0, 1/\sqrt{7}]$. Then 
\begin{align*}
\frac{598}{100} \abs{\mb x^* \mb g}^2 +  \frac{1198}{100}t^2 + \frac{2396}{100} t \Re(\mb x^* \mb g) = \norm{\mb x}{}^2 \paren{\frac{598}{100} \lambda^2 + \frac{1198}{100} \eta^2 + \frac{2396}{100} \lambda  \eta} \doteq \norm{\mb x}{}^2 \phi(\lambda, \eta). 
\end{align*}
For any fixed $\eta$, it is easy to see that minimizer occurs when $\lambda = -\frac{599}{299} \eta$. Plugging this into $\phi(\lambda, \eta)$, one obtains 
$
\phi(\lambda, \eta) \ge -\frac{241}{20} \eta^2 \ge -\frac{241}{140}
$. 
Thus, 
\begin{align*}
\begin{bmatrix}
\mb g \\
\ol{\mb g}
\end{bmatrix}^* 
\nabla^2 f(\mb x + t\mb g) 
\begin{bmatrix}
\mb g \\
\ol{\mb g}
\end{bmatrix} 
\ge \paren{\frac{791}{400} - \frac{241}{140} } \norm{\mb x}{}^2 \ge \frac{1}{4} \norm{\mb x}{}^2, 
\end{align*}
as claimed. 
\end{proof}

\subsection{Proof of Proposition~\ref{prop:grad-region-z}} \label{pf:prop:grad-region-z}
\begin{proof}
Note that 
\begin{align*}
\mb z^* \nabla_{\mb z} f(\mb z)
& = \frac{1}{m}\sum_{k=1}^m \abs{\mb a_k^* \mb z}^4- \frac{1}{m}\sum_{k=1}^m\abs{\mb a_k^* \mb x}^2 \abs{\mb a_k^* \mb z}^2. 
\end{align*}
By Lemma~\ref{lem:min_fourth}, when $m \ge C_1 n\log n$ for some sufficiently large $C_1$, w.h.p., 
\begin{align*}
\frac{1}{m}\sum_{k=1}^m \abs{\mb a_k^* \mb z}^4 \ge \expect{\frac{1}{m}\sum_{k=1}^m \abs{\mb a_k^* \mb z}^4} - \frac{1}{100} \norm{\mb z}{}^4
\end{align*}
for all $\mb z \in \Cp^n$. On the other hand, Lemma~\ref{lem:op_norm_1} implies that when $m \ge C_2 n \log n$ for some sufficiently large $C_2$, w.h.p., 
\begin{align*}
\frac{1}{m} \sum_{k=1}^m \abs{\mb a_k^* \mb x}^2 \abs{\mb a_k^* \mb z}^2 \le 
\expect{\frac{1}{m} \sum_{k=1}^m \abs{\mb a_k^* \mb x}^2 \abs{\mb a_k^* \mb z}^2} + \frac{1}{1000}\norm{\mb x}{}^2 \norm{\mb z}{}^2. 
\end{align*}
for all $\mb z \in \Cp^n$. Combining the above estimates, we have that when $m \ge \max(C_1, C_2) n \log n$, w.h.p., 
\begin{align*}
\mb z^* \nabla_{\mb z} f(\mb z) \ge \expect{\mb z^* \nabla_{\mb z} f(\mb z)} - \frac{1}{100} \norm{\mb z}{}^4 - \frac{1}{1000}\norm{\mb x}{}^2 \norm{\mb z}{}^2 \ge \frac{1}{1000}\norm{\mb x}{}^2 \norm{\mb z}{}^2
\end{align*}
for all $\mb z \in \mc R_2^{\mb z}$, as desired. 
\end{proof}

\subsection{Proof of Proposition~\ref{prop:grad-region-zx}} \label{pf:prop:grad-region-zx}
\begin{proof}
We abbreviate $\phi(\mb z)$ as $\phi$ below. Note that 
\begin{align*}
(\mb z - \mb x \e^{\im \phi} )^* \nabla_{\mb z} f(\mb z)
& = \frac{1}{m}\sum_{k=1}^m \abs{\mb a_k^* \mb z}^2 (\mb z - \mb x \e^{\im \phi} )^* \mb a_k \mb a_k^* \mb z - \frac{1}{m}\sum_{k=1}^m \abs{\mb a_k^* \mb x}^2 (\mb z - \mb x \e^{\im \phi} )^* \mb a_k \mb a_k^* \mb z. 
\end{align*}
We first bound the second term. By Lemma~\ref{lem:op_norm_1}, when $m \ge C_1 n \log n$ for a sufficiently large constant $C_1$, w.h.p., for all $\mb z \in \Cp^n$, 
\begin{align*}
& \Re\paren{\frac{1}{m}\sum_{k=1}^m \abs{\mb a_k^* \mb x}^2 (\mb z - \mb x \e^{\im \phi})^* \mb a_k \mb a_k^* \mb z }  \\
= \; & \Re\paren{(\mb z - \mb x\e^{\im \phi} )^* \expect{\frac{1}{m}\sum_{k=1}^m \abs{\mb a_k^* \mb x}^2 \mb a_k \mb a_k^*}  \mb z} + \Re\paren{(\mb z - \mb x \e^{\im \phi} )^* \mb \Delta \mb z} \quad (\text{where $\norm{\mb \Delta}{} \le \norm{\mb x}{}^2/1000)$}\\
\le\; & \Re\paren{(\mb z - \mb x\e^{\im \phi} )^* \expect{\frac{1}{m}\sum_{k=1}^m \abs{\mb a_k^* \mb x}^2 \mb a_k \mb a_k^*}  \mb z} + \frac{1}{1000} \norm{\mb x}{}^2 \norm{\mb z - \mb x\e^{\im \phi} }{} \norm{\mb z}{}.
\end{align*}
To bound the first term, for a fixed $\tau$ to be determined later, define: 
\begin{align*}
S(\mb z) & \doteq \frac{1}{m}\sum_{k=1}^m \abs{\mb a_k^* \mb z}^2 \Re\paren{(\mb z - \mb x \e^{\im \phi})^* \mb a_k \mb a_k^* \mb z}, \\
S_1(\mb z) & \doteq \frac{1}{m}\sum_{k=1}^m \abs{\mb a_k^* \mb z}^2 \Re\paren{(\mb z - \mb x \e^{\im \phi})^* \mb a_k \mb a_k^* \mb z} \indicator{\abs{\mb a_k^* \mb x} \le \tau} \\
S_2(\mb z) & \doteq \frac{1}{m}\sum_{k=1}^m \abs{\mb a_k^* \mb z}^2 \Re\paren{(\mb z - \mb x \e^{\im \phi})^* \mb a_k \mb a_k^* \mb z} \indicator{\abs{\mb a_k^* \mb x} \le \tau} \indicator{\abs{\mb a_k^* \mb z} \le \tau}. 
\end{align*}
Obviously $S_1(\mb z) \ge S_2(\mb z)$ for all $\mb z$ as 
\begin{align*}
S_1(\mb z) - S_2(\mb z)
& = \frac{1}{m}\sum_{k=1}^m \abs{\mb a_k^* \mb z}^2 \Re\paren{(\mb z - \mb xe^{\im \theta})^* \mb a_k \mb a_k^* \mb z} \indicator{\abs{\mb a_k^* \mb x} \le \tau} \indicator{\abs{\mb a_k^* \mb z} > \tau} \\
& \ge \frac{1}{m}\sum_{k=1}^m \abs{\mb a_k^* \mb z}^2 \paren{\abs{\mb a_k^* \mb z}^2 - \abs{\mb a_k^* \mb x } \abs{\mb a_k^* \mb z}} \indicator{\abs{\mb a_k^* \mb x} \le \tau} \indicator{\abs{\mb a_k^* \mb z} > \tau} \ge 0. 
\end{align*}
Now for an $\eps \in (0, \norm{\mb x}{})$ to be fixed later, consider an $\eps$-net $N_\eps$ for the ball $\bb{CB}^n(\norm{\mb x}{})$, with $\abs{N_\eps} \le (3\norm{\mb x}{}/\eps)^{2n}$. On the complement of the event $\set{\max_{k \in [m]} \abs{\mb a_k^* \mb x} > \tau}$, we have for any $t > 0$ that 
\begin{align*}
\MoveEqLeft \prob{S(\mb z) - \expect{S(\mb z)} < -t,\; \forall\; \mb z \in N_\eps} \\
\le\; & \abs{N_\eps} \prob{S(\mb z) - \expect{S(\mb z)} < -t} \\
\le\; & \abs{N_\eps} \prob{\; S_1(\mb z) - \expect{S_1(\mb z)} < -t + \abs{\expect{S_1(\mb z)} - \expect{S(\mb z)}}\; }. 
\end{align*}
Because $S_1(\mb z) \ge S_2(\mb z)$ as shown above, 
\begin{align*}
\MoveEqLeft \prob{\; S_1(\mb z) - \expect{S_1(\mb z)} < -t + \abs{\expect{S_1(\mb z)} - \expect{S(\mb z)}}\; } \\
\le\; & \prob{\; S_2(\mb z) - \expect{S_2(\mb z)} < -t + \abs{\expect{S_1(\mb z)} - \expect{S(\mb z)}} + \abs{\expect{S_1(\mb z)} - \expect{S_2(\mb z)}}\; }. 
\end{align*}
Thus, the unconditional probability can be bounded as 
\begin{align*}
\MoveEqLeft \prob{S(\mb z) - \expect{S(\mb z)} < -t,\; \forall\; \mb z \in N_\eps} \\
\le\; & \abs{N_\eps} \prob{\; S_2(\mb z) - \expect{S_2(\mb z)} < -t + \abs{\expect{S_1(\mb z)} - \expect{S(\mb z)}} + \abs{\expect{S_1(\mb z)} - \expect{S_2(\mb z)}}\; } \\ 
& \quad + \prob{\max_{k \in [m]} \abs{\mb a_k^* \mb x} > \tau}. 
\end{align*}
Taking $\tau = \sqrt{10\log m} \norm{\mb x}{}$, we obtain 
\begin{align*}
\prob{\max_{k \in [m]} \abs{\mb a_k^* \mb x} > \tau} & \le m \exp\paren{-\frac{10\log m}{2}} \le m^{-4},  \\
\abs{\expect{S_1(\mb z)} - \expect{S(\mb z)}} 
& \le \sqrt{\expect{\abs{\mb a^* \mb z}^6 \abs{\mb a^*\paren{\mb z - \mb x \e^{\im \phi}}}^2}} \sqrt{\prob{\abs{\mb a^* \mb x} > \tau}} \\
& \js{\le \sqrt{\expect{\norm{\mb a}{}^8}}\sqrt{\bb P_{Z \sim \mc {CN}}\brac{\abs{Z} > \sqrt{10 \log m}}}  \norm{\mb z}{}^3 \norm{\mb z - \mb x \e^{\im \phi} }{}} \\ 
& \le 4\sqrt{3} m^{-5/2}\norm{\mb z}{}^3 \norm{\mb z - \mb x \e^{\im \phi} }{}, \\
\abs{\expect{S_1(\mb z)} - \expect{S_2(\mb z)}} 
& \le \sqrt{\expect{\abs{\mb a^* \mb z}^6 \abs{\mb a^*\paren{\mb z - \mb x \e^{\im \phi}}}^2 \indicator{\abs{\mb a^* \mb x} \le \tau}}} \sqrt{\prob{\abs{\mb a^* \mb z} > \tau}} \\
& \js{\le \sqrt{\expect{\norm{\mb a}{}^8}} \sqrt{\bb P_{Z \sim \mc {CN}}\brac{\abs{Z} > \sqrt{10 \log m}}}  \norm{\mb z}{}^3 \norm{\mb z - \mb x \e^{\im \phi} }{}} \\ 
& \le 4\sqrt{3} m^{-5/2}\norm{\mb z}{}^3 \norm{\mb z - \mb x \e^{\im \phi} }{},
\end{align*}
where we have used $\norm{\mb z}{} \le \norm{\mb x}{}$ to simplify the last inequality. Now we use the moment-control Bernstein's inequality (Lemma~\ref{lem:mc_bernstein_scalar}) to get a bound for probability on deviation of $S_2(\mb z)$. To this end, we have
\begin{align*}
\expect{\abs{\mb a^* \mb z}^{6} \abs{\mb a^* (\mb z - \mb x \e^{\im \phi})}^{2} \indicator{\abs{\mb a^* \mb x} \le \tau} \indicator{\abs{\mb a^* \mb z} \le \tau}} 
& \le \tau^{2} \expect{\abs{\mb a^* \mb z}^{4} \abs{\mb a^* (\mb z - \mb x \e^{\im \phi} )}^{2}} \\
& \le 240 \log m \norm{\mb x}{}^2 \norm{\mb z}{}^4 \norm{\mb z - \mb x \e^{\im \phi} }{}^2\\
\expect{\abs{\mb a^* \mb z}^{3p} \abs{\mb a^* (\mb z - \mb x \e^{\im \phi})}^{p} \indicator{\abs{\mb a^* \mb x} \le \tau} \indicator{\abs{\mb a^* \mb z} \le \tau}} 
& \le \tau^{2p} \expect{\abs{\mb a^* \mb z}^{p} \abs{\mb a^* (\mb z - \mb x \e^{\im \phi} )}^{p}} \\
& \le \paren{10\log m \norm{\mb x}{}^2}^p p! \norm{\mb z}{}^p \norm{\mb z - \mb x \e^{\im \phi} }{}^p, 
\end{align*}
where the second inequality holds for any integer $p \ge 3$. Hence one can take 
\begin{align*}
\sigma^2 & =  240 \log^2 m \norm{\mb x}{}^4 \norm{\mb z}{}^2 \norm{\mb z - \mb x \e^{\im \phi} }{}^2, \\
R & = 10 \log m \norm{\mb x}{}^2 \norm{\mb z}{} \norm{\mb z - \mb x \e^{\im \phi} }{}
\end{align*}
in Lemma~\ref{lem:mc_bernstein_scalar}, and 
\begin{align*}
t = \frac{1}{1000} \norm{\mb x}{}^2 \norm{\mb z}{} \norm{\mb z - \mb x \e^{\im \phi} }{} 
\end{align*}
in the deviation inequality of $S_2(\mb z)$ to obtain 
\begin{align*}
\prob{S_2(\mb z)- \expect{S_2(\mb z)} < -\frac{1}{200}\norm{\mb x}{}^2 \norm{\mb z}{} \norm{\mb z - \mb x \e^{\im \phi} }{} } \le \exp\paren{-\frac{c_2m}{\log^2 m}} , 
\end{align*}
where we have used the fact $\norm{\mb z}{} \le \norm{\mb x}{}$ and assumed $4\sqrt{3} m^{-5/2} \le 1/200$ to simplify the probability. Thus, with probability at least $1 - m^{-4} - \exp\paren{-c_2 m/\log^2 m + 2n \log(3\norm{\mb x}{}/\eps)}$, it holds that 
\begin{align}
S(\mb z) \ge \expect{S(\mb z)}-\frac{1}{1000} \norm{\mb x}{}^2 \norm{\mb z}{} \norm{\mb z - \mb x \e^{\im \phi}}{} \quad \forall \; \mb z \in N_\eps.
\end{align}
Moreover, for any $\mb z, \mb z' \in \mc R_2^{\mb h}$, we have 
\begin{align*}
\MoveEqLeft \abs{S(\mb z) - S(\mb z')} \\
\le\; & \frac{1}{m} \sum_{k=1}^m \abs{\abs{\mb a_k^* \mb z}^2 - \abs{\mb a_k^* \mb z'}^2} \abs{\mb h^*(\mb z) \mb a_k \mb a_k^* \mb z}  + \frac{1}{m} \sum_{k=1}^m \abs{\mb a_k^* \mb z'}^2 \abs{\mb h^*(\mb z) \mb a_k \mb a_k^* \mb z - \mb h^*(\mb z') \mb a_k \mb a_k^* \mb z'} \\
\le\; & 4\max_{k \in [m]} \norm{\mb a_k}{}^4 \norm{\mb x}{}^3 \norm{\mb z - \mb z'}{}  + 5\max_{k \in [m]} \norm{\mb a_k}{}^4 \norm{\mb x}{}^3 \norm{\mb z - \mb z'}{} \\
\le\; & 90 n^2 \log^2 m \norm{\mb x}{}^3 \norm{\mb z - \mb z'}{}, 
\end{align*}
as $\max_{k \in [m]} \norm{\mb a_k}{}^4 \le 10n^2 \log^2 m$ with probability at least $1 - c_3m^{-n}$, and $11\norm{\mb x}{}/20 \le \norm{\mb z}{} \le \norm{\mb x}{}$, and also $\|\mb x \e^{\im \phi(\mb z)} - \mb x \e^{\im \phi(\mb z')}\| \le 2 \norm{\mb z - \mb z'}{}$ for $\mb z, \mb z' \in \mc R_2^{\mb h}$. Every $\mb z \in \mc R_2^{\mb h}$ can be  written as $\mb z= \mb z' + \mb e$, with $\mb z' \in N_\eps$ and $\norm{\mb e}{} \leq \eps$. Thus, 
\begin{align*}
S(\mb z) 
& \ge S(\mb z') - 90 n^2 \log^2 m \norm{\mb x}{}^3 \eps \\
& \ge 2\norm{\mb z'}{}^4 - 2\norm{\mb z'}{}^2 \abs{\mb x^* \mb z'}-\frac{1}{1000} \norm{\mb x}{}^2 \norm{\mb z'}{} \norm{\mb z' - \mb x \e^{\im \phi} }{} - 90 n^2 \log^2 m \norm{\mb x}{}^3 \eps \\
& \ge \expect{S(\mb z)}-\frac{1}{1000} \norm{\mb x}{}^2 \norm{\mb z}{} \norm{\mb z- \mb x \e^{\im \phi} }{} - 11\eps \norm{\mb x}{}^3 -  90 n^2 \log^2 m \norm{\mb x}{}^3 \eps, 
\end{align*}
where the additional $11\eps \norm{\mb x}{}^3$ term in the third line is to account for the change from $\mb z'$ to $\mb z$, which has been simplified by assumptions that $11/20 \cdot \norm{\mb x}{} \le \norm{\mb z}{} \le \norm{\mb x}{}$ and that $\eps \le \norm{\mb x}{}$. Choosing $\eps = \norm{\mb x}{}/(c_5n^2\log^2 m)$ for a sufficiently large $c_5 > 0$ and additionally using $\mathrm{dist}(\mb z, \mc X) \ge \norm{\mb x}{}/3$, we obtain that 
\begin{align}
	S(\mb z) \geq \expect{S(\mb z)} - \frac{1}{500} \norm{\mb x}{}^2 \norm{\mb z}{} \norm{\mb z- \mb x \e^{\im \phi} }{}
\end{align}
for all $\mb z \in  \mc R_2^{\mb h}$, with probability at least $1 - c_6m^{-1} - c_7\exp\paren{-c_2m/\log^2 m + c_9n \log (C_8 n \log m )}$. 

Combining the above estimates, when $m \ge C_{10} n \log^3 n$ for sufficiently large constant $C_{10}$, w.h.p., 
\begin{align*}
	\Re \paren{(\mb z - \mb x \e^{\im \phi} )^* \nabla_{\mb z} f(\mb z) } \geq \frac{1}{1000} \norm{\mb x}{}^2 \norm{\mb z}{}\norm{\mb z- \mb x \e^{\im \phi} }{}
\end{align*}
for all $\mb z \in \mc R_2^{\mb h}$, as claimed. 
\end{proof}

\subsection{Proof of Proposition~\ref{prop:region-cover}} \label{pf:prop:region-cover}
\begin{proof}
\js{$\mc R_2^{\mb z}$ and $\mc R_2^{\mb h}$ are complicated regions defined by algebraic inequalities, making a direct proof daunting. We will present an indirect proof by showing that $\mc R_1 \cup \mc R_3 \cup \mc R_2^{\mb z} \cup \mc R_2^{\mb h} = \Cp^n$: since by definition, $\mc R_1 \cup \mc R_3$ and $\mc R_2$ are disjoint and $\paren{\mc R_1 \cup \mc R_3} \cup \mc R_2 = \Cp^n$, the stated set equality implies $\mc R_2 \subset \mc R_2^{\mb z} \cup \mc R_2^{\mb h}$.} 

For convenience, we will define a relaxed $\mc R_2^{\mb h}$ region 
\begin{align*}
\mc R_2^{\mb h'} \doteq \set{\mb z: \Re\paren{\innerprod{\mb h(\mb z) }{\nabla_{\mb z}\expect{f}}} \geq \frac{1}{250}\norm{\mb x}{}^2 \norm{\mb z}{} \norm{\mb h(\mb z)}{}, \norm{\mb z}{} \le \norm{\mb x}{}} \supset \mc R_2^{\mb h}
\end{align*}
and try to show that $\mc R_1 \cup \mc R_2^{\mb z} \cup \mc R_2^{\mb h'} \cup \mc R_3 = \Cp^n$. In the end, we will discuss how this implies the desired set equality. 

We will first divide $\Cp^n$ into several (overlapping) regions, and show that each such region is a subset of $\mc R_1 \cup \mc R_2^{\mb z} \cup \mc R_2^{\mb h'} \cup \mc R_3$. \js{The reason for introducing this parallel division is again that the original sets are very irregular: although each of them has nice geometric properties as we have established, they cannot be described by simple algebraic inequalities. This irregularity makes an analytic argument of the desired coverage hard. The new division scheme will partition $\Cp^n$ into only circular cones and cone segments. This makes keeping track of the covering process much easier. } 

\paragraph{Cover $\mc R_a \doteq \Brac{\mb z: \abs{\mb x^*\mb z}\leq  \frac{1}{2} \norm{\mb x}{} \norm{\mb z}{}} $: } In this case, when $\norm{\mb z}{}^2 \le \tfrac{398}{601}\norm{\mb x}{}^2$, 
\begin{align*}
	8 \abs{\mb x^*\mb z}^2 + \frac{401}{100} \norm{\mb x}{}^2 \norm{\mb z}{}^2 \leq \frac{601}{100} \norm{\mb x}{}^2 \norm{\mb z}{}^2 \le \frac{398}{100} \norm{\mb x}{}^4. 
\end{align*}
On the other hand, when $\norm{\mb z}{}^2 \ge \tfrac{626}{995} \norm{\mb x}{}^2$, 
\begin{align*}
\frac{501}{500}\norm{\mb x}{}^2 \norm{\mb z}{}^2 + \abs{\mb x^*\mb z}^2 \leq  \frac{313}{250} \norm{\mb x}{}^2 \norm{\mb z}{}^2 \le \frac{199}{100} \norm{\mb z}{}^4. 
\end{align*}
Since $\tfrac{398}{601} > \tfrac{626}{995}$, we conclude that $\mc R_a \subset \mc R_1\cup \mc R_2^{\mb z}$.

\paragraph{Cover $\mc R_b \doteq \Brac{\mb z:\; \abs{\mb x^*\mb z} \geq \frac{1}{2} \norm{\mb x}{} \norm{\mb z}{},\; \norm{\mb z}{} \leq \tfrac{57}{100}\norm{\mb x}{} } $: } In this case, 
\begin{align*}
	  8 \abs{\mb x^*\mb z}^2 + \frac{401}{100}\norm{\mb x}{}^2 \norm{\mb z}{}^2 \leq \frac{1201}{100} \norm{\mb x}{}^2 \norm{\mb z}{}^2 \leq \frac{398}{100} \norm{\mb x}{}^4. 
\end{align*}
So $\mc R_b$ is covered by $\mc R_1$. 

\paragraph{Cover $\mc R_c \doteq \Brac{\mb z: \frac{1}{2} \norm{\mb x}{} \norm{\mb z}{} \leq \abs{\mb x^*\mb z} \leq \frac{99}{100} \norm{\mb x}{} \norm{\mb z}{} ,\; \norm{\mb z}{} \geq \frac{11}{20} \norm{\mb x}{} } $: } We show this region is covered by $\mc R_2^{\mb z}$ and $\mc R_2^{\mb h'}$. First, for any $\mb z \in \mc R_c$, when $\norm{\mb z}{} \ge \sqrt{\tfrac{1983}{1990}}\norm{\mb x}{}$, 
\begin{align*}
\frac{501}{500} \norm{\mb x}{}^2 \norm{\mb z}{}^2 + \abs{\mb x^*\mb z}^2 \leq \frac{1983}{1000} \norm{\mb x}{}^2 \norm{\mb z}{}^2 \le \frac{199}{100} \norm{\mb z}{}^4, 
\end{align*}
implying that $\mc R_c \cap \set{\mb z: \norm{\mb z}{} \ge \sqrt{\tfrac{1983}{1990}}\norm{\mb x}{}} \subset \mc R_2^{\mb z}$. Next we suppose $\norm{\mb z}{} = \lambda \norm{\mb x}{}$ and $\abs{\mb x^* \mb z} = \eta \norm{\mb x}{} \norm{\mb z}{}$, where $\lambda \in [\tfrac{11}{20}, \sqrt{\tfrac{1984}{1990}}]$ and $\eta \in [\tfrac{1}{2}, \tfrac{99}{100}]$, and show the rest of $\mc R_c$ is covered by $\mc R_2^{\mb h'}$. To this end, it is enough to verify that 
\begin{align*}
	2\paren{\norm{\mb x}{}^2 - \norm{\mb z}{}^2 }\abs{\mb x^*\mb z} +2 \norm{\mb z}{}^4 - \norm{\mb x}{}^2 \norm{\mb z}{}^2 - \abs{\mb x^*\mb z}^2 - \frac{1}{250} \norm{\mb x}{}^2\norm{\mb z}{} \sqrt{\norm{\mb x}{}^2 + \norm{\mb z}{}^2 - 2\abs{\mb x^* \mb z}} \geq 0
\end{align*}
over this subregion. Writing the left as a function of $\lambda, \eta$ and eliminating $\norm{\mb x}{}$ and $\norm{\mb z}{}$, it is enough to show 
\begin{align*}
h(\lambda, \eta) \doteq 2(1-\lambda^2) \eta + 2 \lambda^3 - \lambda - \eta^2 \lambda - \frac{1}{250} \sqrt{1 + \lambda^2 - 2\eta \lambda} \ge 0,  
\end{align*}
which is implied by 
\begin{align*}
p(\lambda, \eta) \doteq 2(1-\lambda^2) \eta + 2 \lambda^3 - \lambda - \eta^2 \lambda \ge \frac{49}{10000}, 
\end{align*}
as $\frac{1}{250} \sqrt{1 + \lambda^2 - 2\eta \lambda} < 49/10000$. Let $\mb H_p$ be the Hessian matrix of this bivariate function, it is easy to verify that $\mathrm{det}(\mb H_p) = -4 (\eta + \lambda)^2 - 36 \lambda^2 < 0$ for all valid $(\lambda, \eta)$. Thus, the minimizer must occur on the boundary. For any fixed $\lambda$, $2(1-\lambda^2) \eta - \eta^2 \lambda$ is minimized at either $\eta = 99/100$ or $\eta = 1/2$. When $\eta = 99/100$, $p$ is minimized at $\lambda = (4 \cdot 0.99 + \sqrt{40 \cdot 0.99^2 + 24})/12 <\sqrt{1984/1990}$, giving $p \ge 0.019$; when $\eta = 1/2$, $p$ is minimized when $\lambda = (4\cdot 0.5 + \sqrt{40 \cdot 0.5^2 + 24}/12) = (2+\sqrt{34})/12 $, giving $p \ge 0.3$. Overall, $p \ge 0.019 > 49/10000$, as desired. 

\paragraph{Cover $\mc R_d \doteq \Brac{\mb z: \frac{99}{100} \norm{\mb x}{} \norm{\mb z}{} \leq \abs{\mb x^*\mb z } \leq \norm{\mb x}{} \norm{\mb z}{} ,\; \norm{\mb z}{} \geq \frac{11}{20} \norm{\mb x}{} }  $: } We show that this region is covered by $\mc R_2^{\mb z}$, $\mc R_3$, and $\mc R_2^{\mb h'}$ together. First, for any $\mb z \in \mc R_d$, when $\norm{\mb z}{} \ge \sqrt{\tfrac{1001}{995}} \norm{\mb x}{}$, 
\begin{align*}
	\frac{501}{500} \norm{\mb x}{}^2 \norm{\mb z}{}^2+ \abs{\mb x^*\mb z}^2 \leq \frac{1001}{500}\norm{\mb x}{}^2 \norm{\mb z}{}^2 \leq \frac{199}{100} \norm{\mb z}{}^4. 
\end{align*}
So $\mc R_d \cap \set{\mb z: \norm{\mb z}{} \ge \sqrt{\tfrac{1001}{995}} \norm{\mb x}{}} \subset \mc R_2^{\mb z}$. Next, we show that any $\mb z \in \mc R_d$ with $\norm{\mb z}{} \le 24/25 \cdot \norm{\mb x}{}$ is contained in $\mc R_2^{\mb h'}$. Similar to the above argument for $\mc R_c$, it is enough to show 
\begin{align*}
p(\lambda, \eta) \doteq 2(1-\lambda^2) \eta + 2\lambda^3 - \lambda - \eta^2 \lambda \ge 0.00185, 
\end{align*} 
as $\frac{1}{250} \sqrt{1 + \lambda^2 - 2\eta \lambda} < 0.00185$ in this case. Since the Hessian is again always indefinite, we check the optimal value for $\eta = 99/100$ and $\eta = 1$ and do the comparison. It can be verified $p \ge 0.00627 > 0.00185$ in this case. So $\mc R_d \cap \set{\mb z: \norm{\mb z}{} \le \frac{24}{25} \norm{\mb x}{}} \subset \mc R_2^{\mb h'}$. Finally, we consider the case $\tfrac{23}{25} \norm{\mb x}{} \le \norm{\mb z}{} \le \sqrt{\tfrac{1005}{995}} \norm{\mb x}{}$. A $\lambda, \eta$ argument as above leads to 
\begin{align*}
\norm{\mb h(\mb z)}{}^2 = \norm{\mb x}{}^2 + \norm{\mb z}{}^2 - 2 \abs{\mb x^* \mb z}  < \frac{1}{7} \norm{\mb x}{}^2, 
\end{align*} 
implying that $\mc R_d \cap \set{\mb z: \tfrac{23}{25} \norm{\mb x}{} \le \norm{\mb z}{} \le \sqrt{\tfrac{1005}{995}} \norm{\mb x}{}} \subset \mc R_3$. 

In summary, now we obtain that $\Cp^n = \mc R_a \cup \mc R_b \cup \mc R_c \cup \mc R_d \subset \mc R_1 \cup \mc R_2^{\mb z} \cup \mc R_2^{\mb h'} \cup \mc R_3$. Observe that $\mc R_{\mb h'}$ is only used to cover $\mc R_c \cup \mc R_d$, which is in turn a subset of $\set{\mb z: \norm{\mb z}{} \ge 11\norm{\mb x}{}/20}$. Thus, $\Cp^n = \mc R_1 \cup \mc R_2^{\mb z} \cup (\mc R_2^{\mb h'} \cap \set{\mb z: \norm{\mb z}{} \ge 11\norm{\mb x}{}/20}) \cup \mc R_3$. Moreover, by the definition of $\mc R_3$, 
\begin{align*}
& \mc R_1 \cup \mc R_2^{\mb z} \cup (\mc R_2^{\mb h'} \cap \set{\mb z: \norm{\mb z}{} \ge 11\norm{\mb x}{}/20}) \cup \mc R_3 \\
=\;  & \mc R_1 \cup \mc R_2^{\mb z} \cup (\mc R_2^{\mb h'} \cap \set{\mb z: \norm{\mb z}{} \ge 11\norm{\mb x}{}/20} \cap \mc R_3^c)  \cup \mc R_3 \\
\subset\; & \mc R_1 \cup \mc R_2^{\mb z} \cup \mc R_2^{\mb h} \cup \mc R_3 \subset \Cp^n, 
\end{align*}
implying the claimed coverage. 
\end{proof}

%% file: sec/proof_algorithm.tex

\section{Proofs of Technical Results for Trust-Region Algorithm}\label{app:algorithm}

\subsection{Auxiliary Lemmas}
\begin{lemma} \label{lem:gauss_concen_diff}
When $m \ge Cn$ for a sufficiently large $C$, it holds with probability at least $1 - c_a \exp(-c_b m)$ that 
\begin{align*}
\frac{1}{m}\sum_{k=1}^m \abs{\abs{\mb a_k^* \mb z}^2 - \abs{\mb a_k^* \mb w}^2} \le \frac{3}{2} \norm{\mb z - \mb w}{} (\norm{\mb z}{} + \norm{\mb w}{})
\end{align*}
for all $\mb z, \mb w \in \Cp^n$. Here $C$, $c_a$, $c_b$ are positive absolute constants. 
\end{lemma}
\begin{proof}
Lemma 3.1 in~\cite{candes2013phaselift} has shown that when $m \ge C_1 n$, it holds with probability at least $1 - c_2\exp(-c_3m)$ that 
\begin{align*}
\frac{1}{m}\sum_{k=1}^m \abs{\abs{\mb a_k^* \mb z}^2 - \abs{\mb a_k^* \mb w}^2} \le \frac{3}{2\sqrt{2}}\norm{\mb z \mb z^* - \mb w \mb w^*}{\ast}
\end{align*}
for all $\mb z$ and $\mb w$, where $\norm{\cdot}{\ast}$ is the nuclear norm that sums up singular values. The claims follows from 
\begin{align*}
\norm{\mb z \mb z^* - \mb w \mb w^*}{\ast} \le \sqrt{2} \norm{\mb z \mb z^* - \mb w \mb w^*}{} \le \sqrt{2}\norm{\mb z - \mb w}{} (\norm{\mb z}{} + \norm{\mb w}{}), 
\end{align*}
completing the proof. 
\end{proof}

\begin{lemma}\label{lem:hessian-conc-x}
When $m \geq C n \log n$, with probability at least $1 - c_a m^{-1} - c_b\exp\paren{-c_c m/\log m}$, 
	\begin{align*}
		\norm{ \nabla^2 f(\mb x \e^{\im \psi })- \bb E\brac{ \nabla^2 f(\mb x \e^{\im \psi })  }  }{} \leq \frac{1}{100} \norm{\mb x}{}^2
	\end{align*}
	for all $\psi \in [0, 2\pi)$. Here $C$, $c_a$ to $c_c$ are positive absolute constants. 
\end{lemma}
\begin{proof}
By Lemma~\ref{lem:op_norm_1}, we have that 
\begin{align*}
& \norm{\nabla^2 f(\mb x \e^{\im \psi }) - \bb E\brac{ \nabla^2 f(\mb x \e^{\im \psi })  }  }{} \\
\le\; & \norm{\frac{1}{m}\sum_{k=1}^m \abs{\mb a_k^* \mb x}^2 \mb a_k \mb a_k - \paren{\norm{\mb x}{}^2 \mb I + \mb x \mb x^*}}{} + \norm{\frac{1}{m}\sum_{k=1}^m  (\mb a_k^* \mb x)^2 \mb a_k \mb a_k^\top \e^{\im 2\psi} - 2\mb x \mb x^\top \e^{\im 2\psi}}{} \\
\le\; & \frac{1}{200} \norm{\mb x}{}^2  + \frac{1}{200}\norm{\mb x}{}^2 \le \frac{1}{100} \norm{\mb x}{}^2 
\end{align*}
holds w.h.p. when $m \ge C_1 n \log n$ for a sufficiently large $C_1$. 
\end{proof}

\subsection{Proof of Lemma~\ref{lem:lipschitz}} \label{pf:lem:lipschitz}
\begin{proof}
For any $\mb z, \mb z' \in \Gamma'$, we have
\begin{align*}
\abs{f(\mb z) - f(\mb z')} 
& = \frac{1}{2m}\abs{\sum_{k=1}^m \abs{\mb a_k^* \mb z}^4 - \abs{\mb a_k^* \mb z'}^4 - 2\sum_{k=1}^m \abs{\mb a_k^* \mb x}^2 (\abs{\mb a_k^* \mb z}^2  - \abs{\mb a_k^* \mb z'}^2) } \\
& \le \frac{1}{2m}\sum_{k=1}^m (\abs{\mb a_k^* \mb z}^2 + \abs{\mb a_k^* \mb z'}^2) \abs{\abs{\mb a_k^* \mb z}^2 - \abs{\mb a_k^* \mb z'}^2} + \frac{1}{m}\sum_{k=1}^m \abs{\mb a_k^* \mb x}^2 \abs{\abs{\mb a_k^* \mb z}^2 - \abs{\mb a_k^* \mb z'}^2} \\
& \le 4R_1^2 \norm{\mb A}{\ell^1 \to \ell^2}^2 \cdot \frac{3}{2} \cdot 4R_1 \norm{\mb z - \mb z'}{} +  2\norm{\mb A}{\ell^1 \to \ell^2}^2 \norm{\mb x}{}^2 \cdot \frac{3}{2} \cdot 4R_1 \norm{\mb z - \mb z'}{} \\
& \le (24 R_1^3 \norm{\mb A}{\ell^1 \to \ell^2}^2 + 12 \norm{\mb A}{\ell^1 \to \ell^2}^2 \norm{\mb x}{}^2 R_1) \norm{\mb z - \mb z'}{}, 
\end{align*}
where in the third line we invoked results of Lemma~\ref{lem:gauss_concen_diff}, and hence the derived inequality holds w.h.p. when $m \ge C_1 n$. Similarly, for the gradient, 
\begin{align*}
\MoveEqLeft \norm{\nabla f(\mb z) - \nabla f(\mb z')}{} \\
=\; &  \frac{\sqrt{2}}{m} \norm{\sum_{k=1}^m \paren{ \abs{\mb a_k^*\mb z}^2 - \abs{\mb a_k^* \mb x}^2 }\mb a_k \mb a_k^* \mb z - \sum_{k=1}^m \paren{ \abs{\mb a_k^*\mb z'}^2 - \abs{\mb a_k^* \mb x}^2 }\mb a_k \mb a_k^* \mb z' }{} \\
\le\; & \frac{\sqrt{2}}{m} \sum_{k=1}^m \norm{(\abs{\mb a_k^* \mb z}^2 - \abs{\mb a_k^* \mb z'}^2 ) \mb a_k \mb a_k^* \mb z}{} + \sqrt{2}  \norm{\frac{1}{m}\sum_{k=1}^m \mb a_k \mb a_k^* \abs{\mb a_k^* \mb z'}^2}{} \norm{\mb z - \mb z'}{}\\
& \quad + \sqrt{2}  \norm{\frac{1}{m}\sum_{k=1}^m \mb a_k \mb a_k^* \abs{\mb a_k^* \mb x}^2}{} \norm{\mb z - \mb z'}{}  \\
\le\; & \sqrt{2} \norm{\mb A}{\ell^1 \to \ell^2}^2 \cdot 2 R_1 \cdot \frac{3}{2} \cdot 4R_1\norm{\mb z - \mb z'}{} + (8\sqrt{2} \norm{\mb A}{\ell^1 \to \ell^2}^2 R_1^2 + 2\sqrt{2} \norm{\mb A}{\ell^1 \to \ell^2}^2 \norm{\mb x}{}^2) \norm{\mb z - \mb z'}{} \\
\le\; &   (20\sqrt{2} \norm{\mb A}{\ell^1 \to \ell^2}^2 R_1^2 + 2\sqrt{2} \norm{\mb A}{\ell^1 \to \ell^2}^2 \norm{\mb x}{}^2) \norm{\mb z - \mb z'}{}, 
\end{align*}
where from the second to the third inequality we used the fact $\norm{\frac{1}{m}\sum_{k=1}^m \mb a_k \mb a_k^*}{} \le 2$ with probability at least $1 - \exp(-c_2m)$. Similarly for the Hessian, 
\begin{align*}
\MoveEqLeft \norm{\nabla^2 f(\mb z) - \nabla^2 f(\mb z')}{} \\
=\; &  \sup_{\norm{\mb w}{} =1} 
\abs{\frac{1}{2}\begin{bmatrix}
\mb w \\
\ol{\mb w}
\end{bmatrix}^* 
\paren{\nabla^2 f(\mb z) - \nabla^2 f(\mb z')}
\begin{bmatrix}
\mb w \\
\ol{\mb w}
\end{bmatrix}} \\
\le \; &  \sup_{\norm{\mb w}{} = 1} 2\norm{\frac{1}{m}\sum_{k=1}^m (\abs{\mb a_k^* \mb z}^2 - \abs{\mb a_k^* \mb z'}^2) \abs{\mb a_k^* \mb w}^2}{} + \norm{\frac{1}{m}\sum_{k=1}^m \Re ((\mb a_k^* \mb z)^2 - (\mb a_k^* \mb z')^2) (\mb w^* \mb a_k)^2}{} \\
\le\; & 2 \norm{\mb A}{\ell^1 \to \ell^2}^2 \cdot \frac{3}{2} \cdot 4R_1 \norm{\mb z - \mb z'}{} +  \norm{\mb A}{\ell^1 \to \ell^2}^2 \cdot 4R_1 \cdot \norm{\mb z - \mb z'}{} \cdot 2 \\
\le\; & 16 \norm{\mb A}{\ell^1 \to \ell^2}^2 R_1 \norm{\mb z - \mb z'}{}, 
\end{align*}
where to obtain the third inequality we used that $\frac{1}{m} \norm{\mb A^*}{}^2 \le 2$ with probability at least $1 - \exp(-c_3m)$ when $m \ge C_4n$ for a sufficiently large constant $C_4$. 

Since $R_0 \le 10 \norm{\mb x}{}$ with probability at least $1 - \exp(-c_5 m)$ when $m \ge C_6n$, by definition of $R_1$, we have $R_1 \le 30 (n \log m)^{1/2}\norm{\mb x}{}$ w.h.p.. Substituting this estimate into the above bounds yields the claimed results. 
\end{proof}

\subsection{Proof of Lemma~\ref{lem:hessian-lower-upper}} \label{pf:lem:hessian-lower-upper}
\begin{proof}
For the upper bound, we have that for all $\mb z \in \mc R_3'$, 
\begin{align*}
	\norm{\mb H(\mb z) }{} \le \norm{\nabla^2 f(\mb z) }{} &\le \norm{\nabla^2 f(\mb x e^{\im \phi(\mb z) }) }{} + L_h \norm{\mb h(\mb z) }{} \\
		&\le  \norm{ \nabla^2 f(\mb x e^{\im \phi(\mb z) }) - \bb E \brac{ \nabla^2 f(\mb x e^{\im \phi(\mb z) })} }{} + \norm{ \bb E \brac{ \nabla^2 f(\mb x e^{\im \phi(\mb z) })} }{} + \frac{1}{10} \norm{\mb x}{}^2\\
		& \le \frac{1}{100} \norm{\mb x}{}^2 + 4\norm{\mb x}{}^2 + \frac{1}{10} \norm{\mb x}{}^2 \le \frac{9}{2} \norm{\mb x}{}^2, 
\end{align*}
where to obtain the third line we applied Lemma~\ref{lem:hessian-conc-x}. To show the lower bound for all $\mb z \in \mc R_3'$, it is equivalent to show that 
\begin{align*}
\frac{1}{2} 
\begin{bmatrix}
\mb w \\
\ol{\mb w}
\end{bmatrix}^* 
\nabla^2 f(\mb z)
\begin{bmatrix}
\mb w \\
\ol{\mb w}
\end{bmatrix} 
\ge m_H, \quad \forall\; \norm{\mb w}{} = 1 \; \text{with}\; \Im(\mb w^* \mb z) = 0, \; \text{and} \; \forall\; \mb z \in \mc R_3'. 
\end{align*}
 By Lemma \ref{lem:lipschitz} and Lemma \ref{lem:hessian-conc-x}, w.h.p., we have
\begin{align*}
    \frac{1}{2} \begin{bmatrix}\mb w \\ \ol{\mb w}\end{bmatrix}^* \nabla^2 f(\mb z) \begin{bmatrix}\mb w \\ \ol{\mb w}\end{bmatrix} 
    & \ge  \frac{1}{2} \begin{bmatrix}\mb w \\ \ol{\mb w}\end{bmatrix}^* \nabla^2 f(\mb x e^{\im \phi(\mb z)} ) \begin{bmatrix}\mb w \\ \ol{\mb w}\end{bmatrix} - L_h \norm{\mb h(\mb z) }{} \norm{\mb w}{}^2  \\
		& \ge \frac{1}{2} \begin{bmatrix}\mb w \\ \ol{\mb w}\end{bmatrix}^* \bb E\brac{ \nabla^2 f(\mb x e^{\im \phi(\mb z)} ) } \begin{bmatrix}\mb w \\ \ol{\mb w}\end{bmatrix} - \left(\frac{1}{10} + \frac{1}{100} \right)\norm{\mb x}{}^2 \\
		& = \left(1- \frac{1}{100} - \frac{1}{10}\right) \norm{\mb x}{}^2 + \abs{\mb w^*\mb x}^2 +2 \Re\paren{ (\mb w^*\mb x \e^{\im \phi(\mb z) })^2 } \\
		& \ge \frac{89}{100} \norm{\mb x}{}^2 + \Re\paren{ (\mb w^*\mb x \e^{\im \phi(\mb z) })^2 }. 
\end{align*}
Since $\Im\paren{\mb w^*\mb z }= 0$, we have $\Re\paren{(\mb w^*\mb z)^2 } = \abs{\mb w^*\mb z }^2$. Thus, 
\begin{align*}
		\Re\paren{ (\mb w^*\mb x \e^{\im \phi(\mb z) })^2  } 
		& =  \Re\paren{ (\mb w^*\mb z - \mb w^* \mb h(\mb z) )^2 } \\
		& =  \abs{\mb w^*\mb z }^2 + \Re\paren{( \mb w^*\mb h  )^2  } -2 \Re\paren{ (\mb w^*\mb h(\mb z) ) (\mb w^*\mb z)  }  \\
		& \ge \abs{\mb w^*\mb z }^2 - \norm{ \mb w }{}^2 \norm{ \mb h(\mb z) }{}^2 - 2 \norm{\mb w}{}^2 \norm{\mb h(\mb z)}{} \norm{\mb z}{} \\
	 & \ge - \frac{1}{100 L_h^2} \norm{\mb x}{}^4 - \frac{2}{10L_h} \norm{\mb x}{}^2 \left(\norm{\mb x}{} + \frac{1}{10L_h} \norm{\mb x}{}^2 \right) \\
	 & \ge -\frac{1}{100} \norm{\mb x}{}^2, 
\end{align*}
where we obtained the last inequality based on the fact that $L_h \doteq 480(n\log m)^{1/2} \norm{\mb A}{\ell^1 \to \ell^2}^2 \norm{\mb x}{} \ge 150 \norm{\mb x}{}$ whenever $\norm{\mb A}{\ell^1 \to \ell^2}^2 \ge 1$; this holds w.h.p. when $m \ge C_1n$ for large enough constant $C_1$. Together we obtain 
\begin{align*}
\frac{1}{2} \begin{bmatrix}\mb w \\ \ol{\mb w}\end{bmatrix}^* \nabla^2 f(\mb z) \begin{bmatrix}\mb w \\ \ol{\mb w}\end{bmatrix} 
\ge \frac{89}{100} \norm{\mb x}{}^2 -\frac{1}{100} \norm{\mb x}{}^2 \ge \frac{22}{25}\norm{\mb x}{}^2, 
\end{align*}
as desired. 
\end{proof}

\subsection{Proof of Lemma~\ref{lem:func-decent}} \label{pf:lem:func-decent}
\begin{proof}
In view of Lemma \ref{lem:Taylor-approx}, we have 
\begin{align*}
f(\mb z + \mb \delta_\star) &\le \widehat{f}(\mb \delta_\star; \mb z) + \tfrac{1}{3} L_h \Delta^3 \\
 &\le \widehat{f}(\mb \delta; \mb z ) + \tfrac{1}{3} L_h \Delta^3 \\
 &\le f(\mb z + \mb \delta) + \tfrac{2}{3} L_h \Delta^3 \\
 &\le f(\mb z) - d + \tfrac{2}{3} L_h \Delta^3,
 \end{align*}
	as desired.
\end{proof}

\subsection{Proof of Proposition~\ref{prop:negative-curvature}} \label{pf:prop:negative-curvature}
\begin{proof}
In view of Proposition~\ref{prop:nega-curv}, consider direction $\mb \delta \doteq \mb x \e^{\im \phi(\mb z)} /\norm{\mb x}{}$. Obviously, vectors of the form $t \sigma \mb \delta$ are feasible for~\eqref{eqn:trm-subproblem} for any $t \in [0, \Delta]$ and $\sigma \doteq -\sign([\mb \delta^*, \ol{\mb \delta}^*] \nabla f(\mb z^{(r)}))$. 
 By Lemma~\ref{lem:Taylor-integral-form}, we obtain 
\begin{align*}
	f(\mb z^{(r)}+ t\sigma \mb \delta) 
	\;&=\; f(\mb z^{(r)}) +t\sigma \begin{bmatrix}\mb \delta \\ \ol{\mb \delta}	\end{bmatrix}^*\nabla f(\mb z^{(r)}) +t^2 \int_0^1 (1-s) \begin{bmatrix}\mb \delta \\ \ol{\mb \delta}	\end{bmatrix}^* \nabla^2 f(\mb z^{(r)} + \sigma st \mb \delta )\begin{bmatrix}\mb \delta \\ \ol{\mb \delta}	\end{bmatrix}\; ds \nonumber \\
	& \leq \; f(\mb z^{(r)}) + \frac{t^2}{2} \begin{bmatrix}\mb \delta \\ \ol{\mb \delta}	\end{bmatrix}^* \nabla^2 f(\mb z^{(r)})\begin{bmatrix}\mb \delta \\ \ol{\mb \delta}	\end{bmatrix} \nonumber \\
	&+ t^2 \int_0^1 (1-s) \begin{bmatrix}\mb \delta \\ \ol{\mb \delta}	\end{bmatrix}^* \brac{\nabla^2 f(\mb z^{(r)}+\sigma st \mb \delta) -\nabla^2 f(\mb z^{(r)})  }\begin{bmatrix}\mb \delta \\ \ol{\mb \delta}	\end{bmatrix}\; ds \nonumber \\
	& \leq \; f(\mb z^{(r)}) + \frac{t^2}{2} \begin{bmatrix}\mb \delta \\ \ol{\mb \delta}	\end{bmatrix}^* \nabla^2 f(\mb z^{(r)})\begin{bmatrix}\mb \delta \\ \ol{\mb \delta}	\end{bmatrix} + \frac{L_h}{3} t^3.   
\end{align*}
Thus, we have
\begin{align*}
	f(\mb z^{(r)}+ t\sigma \mb \delta) -f(\mb z^{(r)}) \leq  - \frac{1}{200}t^2 \norm{\mb x}{}^2 + \frac{L_h}{3}t^3.
\end{align*}
Taking $t = \Delta$ and applying Lemma~\ref{lem:func-decent}, we have 
\begin{align*}
f(\mb z^{(r+1)}) - f(\mb z^{(r)}) \le  - \frac{1}{200}\Delta^2 \norm{\mb x}{}^2 + \frac{L_h}{3}\Delta^3 + \frac{2}{3} L_h \Delta^3 \le - \frac{1}{200}\Delta^2\norm{\mb x}{}^2 + L_h \Delta^3 \le -\frac{1}{400} \norm{\mb x}{}^2 \Delta^2, 
\end{align*}
where we obtain the very last inequality using the assumption that $\Delta \le \norm{\mb x}{}^2/(400L_h)$, completing the proof. 
\end{proof}

\subsection{Proof of Proposition~\ref{prop:large-gradient}} \label{pf:prop:large-gradient} 
\begin{proof}
We take 
\begin{align*}
\mb \delta = 
\begin{cases}
- \mb z^{(r)}/\norm{\mb z^{(r)}}{} & \; \mb z^{(r)} \in \mc R_2^{\mb z} \\
-\mb h(\mb z^{(r)})/\norm{\mb h(\mb z^{(r)})}{} & \; \mb z^{(r)} \in \mc R_2^{\mb h}
\end{cases}. 
\end{align*}
Obviously vectors of the form $t\mb \delta$ is feasible for~\eqref{eqn:trm-subproblem} for any $t \in [0, \Delta]$. By Lemma \ref{lem:Taylor-integral-form}, we have
	\begin{align*}
		f(\mb z^{(r)}+ t\mb \delta) \;&=\; f(\mb z^{(r)}) + t \int_0^1 \begin{bmatrix}\mb \delta \\ \ol{\mb \delta}\end{bmatrix}^* \nabla f(\mb z^{(r)}+ st \mb \delta) \; ds \nonumber \\
		&=\; f(\mb z^{(r)}) + t \begin{bmatrix}\mb \delta \\ \ol{\mb \delta}\end{bmatrix}^* \nabla f(\mb x^{(r)}) + t \int_0^1 \begin{bmatrix}\mb \delta \\ \ol{\mb \delta}\end{bmatrix}^* \brac{\nabla f(\mb z^{(r)}+ st \mb \delta) -\nabla f(\mb z^{(r)}) }  \; ds \nonumber \\
		&\leq\; f(\mb z^{(r)}) + t \begin{bmatrix}\mb \delta \\ \ol{\mb \delta}\end{bmatrix}^* \nabla f(\mb z^{(r)})+ t^2 L_g.
	\end{align*}
By Proposition \ref{prop:grad-region-z} and Proposition \ref{prop:grad-region-zx}, we have 
	\begin{align*}
		f(\mb z^{(r)}+ t\mb \delta) - f(\mb z^{(r)}) \leq  - \frac{1}{1000} t \norm{\mb x}{}^2 \|\mb z^{(r)}\| + t^2L_g.
	\end{align*} 
Since $\set{\mb z: \norm{\mb z}{} \le \norm{\mb x}{}/2} \subset \mc R_1$, $\mb z^{(r)}$ of interest here satisfies $\|\mb z^{(r)}\| \ge \norm{\mb x}{}/2$. Thus, 
\begin{align*}
	f(\mb z^{(r)}+ t\mb \delta) - f(\mb z^{(r)}) \leq - \frac{1}{2000} t \norm{\mb x}{}^3 + t^2L_g. 
\end{align*}
Combining the above with Lemma \ref{lem:func-decent}, we obtain 
\begin{align*}
f(\mb z^{(r+1)}) - f(\mb z^{(r)}) \le  - \frac{1}{2000} \Delta \norm{\mb x}{}^3 + \Delta^2L_g + \frac{2}{3}L_h \Delta^3 \le -\frac{1}{4000} \Delta \norm{\mb x}{}^3, 
\end{align*}
provided 
\begin{align*}
\Delta \le \min \set{\frac{\norm{\mb x}{}^3}{8000L_g}, \sqrt{\frac{3\norm{\mb x}{}^3}{16000L_h}}}, 
\end{align*}
as desired.  
\end{proof}

\subsection{Proof of Proposition~\ref{prop:func-value-decrease-R-3}} \label{pf:prop:func-value-decrease-R-3}
\begin{proof}
By Proposition~\ref{prop:str_cvx} and the integral form of Taylor's theorem in Lemma~\ref{lem:Taylor-integral-form}, we have that for any $\mb g$ satisfying $\Im(\mb g^* \mb x) = 0$ and $\norm{\mb g}{} = 1$ and any $t \in [0, \norm{\mb x}{}/\sqrt{7}]$, 
\begin{align*}
f(\mb x + t\mb g) 
& = f(\mb x) + t
\begin{bmatrix}
\mb g \\
\ol{\mb g}
\end{bmatrix}^* \nabla f(\mb x) 
+ t^2 \int_0^1 (1-s) 
\begin{bmatrix}
\mb g \\
\ol{\mb g}
\end{bmatrix}^*
\nabla^2 f(\mb x + st\mb g) 
\begin{bmatrix}
\mb g \\
\ol{\mb g}
\end{bmatrix}\; 
ds \\
& \ge f(\mb x) + t
\begin{bmatrix}
\mb g \\
\ol{\mb g}
\end{bmatrix}^* \nabla f(\mb x) 
+ \frac{1}{8} \norm{\mb x}{}^2 t^2. 
\end{align*}
Similarly, we have 
\begin{align*}
f(\mb x) \ge f(\mb x + t\mb g) -t \begin{bmatrix}
\mb g \\
\ol{\mb g}
\end{bmatrix}^* \nabla f(\mb x + t\mb g) 
+ \frac{1}{8} \norm{\mb x}{}^2 t^2. 
\end{align*}
Combining the above two inequalities, we obtain 
\begin{align*}
t \begin{bmatrix}
\mb g \\
\ol{\mb g}
\end{bmatrix}^*
(\nabla f(\mb x + t\mb g) - \nabla f(\mb x)) \ge \frac{1}{4} \norm{\mb x}{}^2 t^2 \Longrightarrow  \begin{bmatrix}
\mb g \\
\ol{\mb g}
\end{bmatrix}^* \nabla f(\mb x + t\mb g) \ge \frac{1}{4} \norm{\mb x}{}^2 t \ge \frac{1}{40L_h} \norm{\mb x}{}^4, 
\end{align*}
where to obtain the very last bound we have used the fact $\min_{\mb z \in \mc R_3\setminus \mc R_3'} \norm{\mb h(\mb z)}{} \ge \norm{\mb x}{}^2/(10L_h)$ due to~\eqref{eq:r3p_def}. This implies that for all $\mb z \in \mc R_3\setminus \mc R_3'$,  
\begin{align} \label{eq:grad_bound_r3d}
\begin{bmatrix}
\mb h(\mb z) \\
\ol{\mb h(\mb z)}
\end{bmatrix}^* \nabla f(\mb z) \ge \frac{1}{40L_h} \norm{\mb x}{}^4. 
\end{align}
The rest arguments are very similar to that of Proposition~\ref{prop:large-gradient}. Take $\mb \delta = -\mb h(\mb z^{(r)}) /\norm{\mb h(\mb z^{(r)})}{}$ and it can checked vectors of the form $t\mb \delta$ for $t \in [0, \Delta]$ are feasible for~\eqref{eqn:trm-subproblem}. By Lemma \ref{lem:Taylor-integral-form}, we have
	\begin{align*}
		f(\mb z^{(r)}+ t\mb \delta) \;&=\; f(\mb z^{(r)}) + t \int_0^1 \begin{bmatrix}\mb \delta \\ \ol{\mb \delta}\end{bmatrix}^* \nabla f(\mb z^{(r)}+ st \mb \delta) \; ds \nonumber \\
		&= f(\mb z^{(r)}) + t \begin{bmatrix}\mb \delta \\ \ol{\mb \delta}\end{bmatrix}^* \nabla f(\mb x^{(r)}) + t \int_0^1 \begin{bmatrix}\mb \delta \\ \ol{\mb \delta}\end{bmatrix}^* \brac{\nabla f(\mb z^{(r)}+ st \mb \delta) -\nabla f(\mb z^{(r)}) }  \; ds \nonumber \\
		&\leq f(\mb z^{(r)}) + t \begin{bmatrix}\mb \delta \\ \ol{\mb \delta}\end{bmatrix}^* \nabla f(\mb z^{(r)})+ t^2 L_g \\
		& \le f(\mb z^{(r)}) - \frac{1}{40L_h} t\norm{\mb x}{}^4 + t^2 L_g, 
	\end{align*}
where to obtain the last line we have used~\eqref{eq:grad_bound_r3d}. Combining the above with Lemma \ref{lem:func-decent}, we obtain 
\begin{align*}
f(\mb z^{(r+1)}) - f(\mb z^{(r)}) \le - \frac{1}{40L_h} \Delta \norm{\mb x}{}^4 + \Delta^2 L_g + \frac{2}{3}L_h\Delta^3 \le -\frac{1}{80L_h} \Delta \norm{\mb x}{}^4, 
\end{align*}
provided 
\begin{align*}
\Delta \le \min \set{\frac{\norm{\mb x}{}^4}{160L_h L_g}, \sqrt{\frac{3}{320}} \frac{\norm{\mb x}{}^2}{L_h}}, 
\end{align*}
as desired.  
\end{proof}

\subsection{Proof of Proposition~\ref{prop:fix-decrease-R-3'}} \label{pf:prop:fix-decrease-R-3'}
\begin{proof}
If we identify $\Cp^n$ with $\R^{2n}$, it can be easily verified that the orthoprojectors of a vector $\mb w$ onto $\mb z$ and its orthogonal complement are
\begin{align*}
\mc P_{\mb z}(\mb w) = \frac{\Re(\mb z^* \mb w) \mb z}{\norm{\mb z}{}^2}, \quad \text{and} \quad \mc P_{\mb z^\perp}(\mb w) = \mb w - \frac{\Re(\mb z^* \mb w) \mb z}{\norm{\mb z}{}^2}. 
\end{align*}
Now at any point $\mb z^{(r)} \in \mc R_3'$, consider a feasible direction of the form $\mb \delta \doteq -t\mc P_{(\im \mb z^{(r)})^\perp} \nabla_{\mb z^{(r)}} f(\mb z^{(r)})$ ($0 \le t \le \Delta/\|\mc P_{(\im \mb z^{(r)})^\perp} \nabla_{\mb z^{(r)}} f(\mb z^{(r)})\|$) to the trust-region subproblem~\eqref{eqn:trm-subproblem}. The local quadratic approximation obeys 
\begin{align*}
	\wh{f}(\mb \delta; \mb z^{(r)}) 
	& = f(\mb z^{(r)}) + 
	\begin{bmatrix}
	\mb \delta \\
	\ol{\mb \delta}
	\end{bmatrix}^* 
	\nabla f(\mb z^{(r)}) + \frac{1}{2} 	\begin{bmatrix}
	\mb \delta \\
	\ol{\mb \delta}
	\end{bmatrix}^* \nabla^2 f(\mb z^{(r)})
		\begin{bmatrix}
	\mb \delta \\
	\ol{\mb \delta}
	\end{bmatrix} \\
	& \le  f(\mb z^{(r)}) - 2t \norm{ \mc P_{(\im \mb z^{(r)})^\perp} \nabla_{\mb z^{(r)}} f(\mb z^{(r)}) }{}^2 + t^2 M_H \norm{ \mc P_{(\im \mb z^{(r)})^\perp} \nabla_{\mb z^{(r)}} f(\mb z^{(r)}) }{}^2 \\
	& = f(\mb z^{(r)}) - 2t \paren{1- \frac{M_H}{2}t } \norm{ \mc P_{(\im \mb z^{(r)})^\perp} \nabla_{\mb z^{(r)}} f(\mb z^{(r)}) }{}^2,
\end{align*}
where $M_H$ is as defined in Lemma~\ref{lem:hessian-lower-upper}. Taking $t =\min\{ M_H^{-1}, \Delta/\|\mc P_{(\im \mb z^{(r)})^\perp} \nabla_{\mb z^{(r)}} f(\mb z^{(r)})\|\}$, we have 
\begin{align*}
	\wh{f}(\mb \delta; \mb z^{(r)})  - f(\mb z^{(r)}) \le -\min\{ M_H^{-1}, \Delta/\|\mc P_{(\im \mb z^{(r)})^\perp} \nabla_{\mb z^{(r)}} f(\mb z^{(r)})\|\} \norm{ \mc P_{(\im \mb z^{(r)})^\perp} \nabla_{\mb z^{(r)}} f(\mb z^{(r)}) }{}^2. 
\end{align*}
Let $\mb U$ be an orthogonal (in geometric sense) basis for the space $\set{\mb w: \Im(\mb w^* \mb z^{(r)}) = 0}$. In view of the transformed gradient and Hessian in~\eqref{eqn:def-g-H}, it is easy to see 
\begin{align*}
\norm{ \mc P_{(\im \mb z^{(r)})^\perp} \nabla_{\mb z^{(r)}} f(\mb z^{(r)}) }{} = \frac{1}{\sqrt{2}} \norm{\mb g(\mb z^{(r)})}{}, 
\end{align*}
where $\mb g(\mb z^{(r)})$ is the transformed gradient. To lower bound $\norm{ \mc P_{(\im \mb z^{(r)})^\perp} \nabla_{\mb z^{(r)}} f(\mb z^{(r)}) }{}$, recall the step is constrained, we have
\begin{align*}
\Delta \le \norm{\mb H^{-1}(\mb z^{(r)}) \mb g(\mb z^{(r)})}{}  \le \norm{\mb H^{-1}(\mb z^{(r)})}{} \norm{\mb g(\mb z^{(r)})}{} \le \frac{1}{\lambda_{\min}(\mb H(\mb z^{(r)}))} \norm{\mb g(\mb z^{(r)})}{}. 
\end{align*}
By Lemma \ref{lem:hessian-lower-upper}, $\lambda_{\min}(\mb H(\mb z^{(r)})) \geq m_H$. Thus, 
\begin{align*}
\norm{\mb g(\mb z^{(r)})}{} \ge m_H\Delta.  
\end{align*}
Hence we have 
\begin{align*}
\wh{f}(\mb \delta; \mb z^{(r)})  - f(\mb z^{(r)}) \le -\min\set{\frac{m_H^2 \Delta^2}{2M_H}, \frac{\Delta^2 m_H}{\sqrt{2}}} \le -\frac{m_H^2 \Delta^2}{2M_H},  
\end{align*}
where the last simplification is due to that $M_H \ge m_H$. By Lemma \ref{lem:Taylor-approx}, we have
\begin{align*}
	f(\mb z^{(r)}+ \mb \delta) - f(\mb z^{(r)}) \leq -\frac{m_H^2 \Delta^2}{2M_H} + \frac{L_h}{3} \Delta^3.
\end{align*}
Therefore, for $\mb z^{(r+1)} = \mb z^{(r)} + \mb \delta_\star$, Lemma \ref{lem:func-decent} implies that
\begin{align*}
	f(\mb z^{(r+1)}) - f(\mb z^{(r)}) \leq -\frac{m_H^2 \Delta^2}{2M_H} + L_h \Delta^3.
\end{align*}
The claimed result follows provided $\Delta \le m_H^2/(4M_H L_h)$, completing the proof. 
\end{proof}

\subsection{Proof of Proposition~\ref{prop:grad-quad-conv}} \label{pf:prop:grad-quad-conv}

Before proceeding, we note one important fact that is useful below. For any $\mb z$, we have 
\begin{align*}
\mc P_{\im \mb z} \nabla_{\mb z} f(\mb z) = \frac{\Re((\im \mb z)^* \nabla_{\mb z} f(\mb z))}{\norm{\mb z}{}^2} \im \mb z = \mb 0. 
\end{align*}
Thus, if $\mb U(\mb z)$ is an (geometrically) orthonormal basis constructed for the space $\set{\mb w: \Im(\mb w^* \mb z) = 0}$ (as defined around~\eqref{eqn:def-g-H}), it is easy to verify that 
\begin{align}
\begin{bmatrix}
\mb U \\
\ol{\mb U}
\end{bmatrix}
\begin{bmatrix}
\mb U \\
\ol{\mb U}
\end{bmatrix}^*
\nabla f(\mb z) = 2\nabla f(\mb z). 
\end{align}

We next prove Proposition~\ref{prop:grad-quad-conv}.

\begin{proof}
Throughout the proof, we write $\mb g^{(r)}$, $\mb H^{(r)}$ and $\mb U^{(r)}$ short for $\mb g(\mb z^{(r)})$, $\mb H(\mb z^{(r)})$ and $\mb U(\mb z^{(r)})$, respectively. Given an orthonormal basis $\mb U^{(r)}$ for $\set{\mb w: \Im(\mb w^* \mb z^{(r)}) = 0}$, the unconstrained optimality condition of the trust region method implies that
\begin{align*}
\mb H^{(r)} \mb \xi_\star + \mb g^{(r)} = \mb 0 \Longleftrightarrow 
\begin{bmatrix}
\mb U^{(r)} \\
\ol{\mb U^{(r)}}
\end{bmatrix}^* 
\nabla^2 f(\mb z^{(r)})
\begin{bmatrix}
\mb U^{(r)} \\
\ol{\mb U^{(r)}}
\end{bmatrix}
\mb \xi_\star + \begin{bmatrix}
\mb U^{(r)} \\
\ol{\mb U^{(r)}}
\end{bmatrix}^* \nabla f(\mb z^{(r)}) = \mb  0. 
\end{align*}
Thus, we have 
{\small 
\begin{align*}
& \|\nabla f(\mb z^{(r+1)})\| \\
=\; & \frac{1}{2}\norm{\begin{bmatrix}
\mb U^{(r+1)} \\
\ol{\mb U^{(r+1)}}
\end{bmatrix}
\begin{bmatrix}
\mb U^{(r+1)} \\
\ol{\mb U^{(r+1)}}
\end{bmatrix}^* 
\nabla f(\mb z^{(r+1)})
}{} \\
=\; &  
\frac{1}{2}\norm{\begin{bmatrix}
\mb U^{(r+1)} \\
\ol{\mb U^{(r+1)}}
\end{bmatrix}
\begin{bmatrix}
\mb U^{(r+1)} \\
\ol{\mb U^{(r+1)}}
\end{bmatrix}^* 
\nabla f(\mb z^{(r+1)}) - 
\begin{bmatrix}
\mb U^{(r)} \\
\ol{\mb U^{(r)}}
\end{bmatrix}
\begin{bmatrix}
\mb U^{(r)} \\
\ol{\mb U^{(r)}}
\end{bmatrix}^* 
\paren{
\nabla^2 f(\mb z^{(r)})
\begin{bmatrix}
\mb U^{(r)} \\
\ol{\mb U^{(r)}}
\end{bmatrix}
\mb \xi_\star + \nabla f(\mb z^{(r)})}}{}  \\
\le\; &  
\frac{1}{2}\norm{
\begin{bmatrix}
\mb U^{(r+1)} \\
\ol{\mb U^{(r+1)}}
\end{bmatrix}
\begin{bmatrix}
\mb U^{(r+1)} \\
\ol{\mb U^{(r+1)}}
\end{bmatrix}^* \brac{\nabla f(\mb z^{(r+1)}) - \nabla f(\mb z^{(r)}) - \nabla^2 f(\mb z^{(r)})
\begin{bmatrix}
\mb U^{(r)} \\
\ol{\mb U^{(r)}}
\end{bmatrix}
\mb \xi_\star}}{}\\
& \quad + 
\frac{1}{2}
\norm{
\paren{
\begin{bmatrix}
\mb U^{(r+1)} \\
\ol{\mb U^{(r+1)}}
\end{bmatrix}
\begin{bmatrix}
\mb U^{(r+1)} \\
\ol{\mb U^{(r+1)}}
\end{bmatrix}^* - 
\begin{bmatrix}
\mb U^{(r)} \\
\ol{\mb U^{(r)}}
\end{bmatrix}
\begin{bmatrix}
\mb U^{(r)} \\
\ol{\mb U^{(r)}}
\end{bmatrix}^*} 
\paren{\nabla^2 f(\mb z^{(r)})
\begin{bmatrix}
\mb U^{(r)} \\
\ol{\mb U^{(r)}}
\end{bmatrix}
\mb \xi_\star + \nabla f(\mb z^{(r)})}}{} \\
\le\; & 
\norm{\nabla f(\mb z^{(r+1)}) - \nabla f(\mb z^{(r)}) - \nabla^2 f(\mb z^{(r)})
\begin{bmatrix}
\mb U^{(r)} \\
\ol{\mb U^{(r)}}
\end{bmatrix}
\mb \xi_\star}{} \\
& \quad + 
\frac{1}{2}\norm{\begin{bmatrix}
\mb U^{(r+1)} \\
\ol{\mb U^{(r+1)}}
\end{bmatrix}
\begin{bmatrix}
\mb U^{(r+1)} \\
\ol{\mb U^{(r+1)}}
\end{bmatrix}^* - 
\begin{bmatrix}
\mb U^{(r)} \\
\ol{\mb U^{(r)}}
\end{bmatrix}
\begin{bmatrix}
\mb U^{(r)} \\
\ol{\mb U^{(r)}}
\end{bmatrix}^*}{} 
\norm{\nabla^2 f(\mb z^{(r)})
\begin{bmatrix}
\mb U^{(r)} \\
\ol{\mb U^{(r)}}
\end{bmatrix}
\mb \xi_\star + \nabla f(\mb z^{(r)})}{}. 
\end{align*}
}
By Taylor's theorem and Lipschitz property in Lemma~\ref{lem:lipschitz}, we have 
\begin{align}
& \norm{\nabla f(\mb z^{(r+1)}) - \nabla f(\mb z^{(r)}) - \nabla^2 f(\mb z^{(r)})
\begin{bmatrix}
\mb U^{(r)} \\
\ol{\mb U^{(r)}}
\end{bmatrix}
\mb \xi_\star}{}  \nonumber \\
= & \norm{\int_0^1 \brac{\nabla^2 f(\mb z^{(r)} + t \begin{bmatrix}
\mb U^{(r)} \\
\ol{\mb U^{(r)}}
\end{bmatrix}
\mb \xi_\star) - \nabla^2 f(\mb z^{(r)})} \begin{bmatrix}
\mb U^{(r)} \\
\ol{\mb U^{(r)}}
\end{bmatrix} \mb \xi_\star\; dt}{} \nonumber\\
\le & \norm{\mb \xi_\star}{} \int_0^1 \norm{\nabla^2 f(\mb z^{(r)} + t \begin{bmatrix}
\mb U^{(r)} \\
\ol{\mb U^{(r)}}
\end{bmatrix}
\mb \xi_\star) - \nabla^2 f(\mb z^{(r)})}{} \; dt \nonumber \\
\le & \frac{1}{2}L_h \norm{\mb \xi_\star}{}^2. 
\end{align}
Moreover, 
\begin{align*}
\norm{\nabla f(\mb z^{(r)})}{} 
& = \frac{1}{\sqrt{2}} \norm{\begin{bmatrix}
\mb U^{(r)} \\
\ol{\mb U^{(r)}}
\end{bmatrix}^* \nabla f(\mb z^{(r)})}{}  \\
& = \frac{1}{\sqrt{2}} \norm{-\begin{bmatrix}
\mb U^{(r)} \\
\ol{\mb U^{(r)}}
\end{bmatrix}^* 
\nabla^2 f(\mb z^{(r)})
\begin{bmatrix}
\mb U^{(r)} \\
\ol{\mb U^{(r)}}
\end{bmatrix}
\mb \xi_\star}{} \le \sqrt{2} \norm{\nabla^2 f(\mb z^{(r)})}{} \norm{\mb \xi_\star}{}, 
\end{align*}
where to obtain the second equality we have used the optimality condition discussed at start of the proof. Thus, using the result above, we obtain
\begin{align}
\norm{\nabla^2 f(\mb z^{(r)})
\begin{bmatrix}
\mb U^{(r)} \\
\ol{\mb U^{(r)}}
\end{bmatrix}
\mb \xi_\star + \nabla f(\mb z^{(r)})}{} \le 2\sqrt{2} \norm{\nabla^2 f(\mb z^{(r)})}{} \norm{\mb \xi_\star}{}. 
\end{align}
On the other hand, 
\begin{align*}
& \norm{\begin{bmatrix}
\mb U^{(r+1)} \\
\ol{\mb U^{(r+1)}}
\end{bmatrix}
\begin{bmatrix}
\mb U^{(r+1)} \\
\ol{\mb U^{(r+1)}}
\end{bmatrix}^* - 
\begin{bmatrix}
\mb U^{(r)} \\
\ol{\mb U^{(r)}}
\end{bmatrix}
\begin{bmatrix}
\mb U^{(r)} \\
\ol{\mb U^{(r)}}
\end{bmatrix}^*}{} \\
\le\; & \norm{\mb U^{(r+1)} (\mb U^{(r+1)})^* - \mb U^{(r)} (\mb U^{(r)})^*}{} + \norm{\mb U^{(r+1)} (\mb U^{(r+1)})^\top - \mb U^{(r)} (\mb U^{(r)})^\top}{}. 
\end{align*}
Write $\mb U^{(r+1)} = \mb U^{(r+1)}_\Re + \im \mb U^{(r+1)}_\Im$, where $\mb U^{(r+1)}_\Re$ and $\mb U^{(r+1)}_\Im$ collect respectively entrywise real and imaginary parts of $\mb U^{(r+1)}$. It is not difficult to verify that $\mb V^{(r+1)} \doteq [\mb U^{(r+1)}_{\Re}; \mb U^{(r+1)}_{\Im}] \in \R^{2n \times (2n-1)}$ is an orthonormal matrix. We also define $\mb V^{(r)}$ accordingly. Thus, 
\begin{align*}
\norm{\mb U^{(r+1)} (\mb U^{(r+1)})^* - \mb U^{(r)} (\mb U^{(r)})^*}{} 
& = \norm{[\mb I, \im \mb I] \paren{\mb V^{(r+1)} (\mb V^{(r+1)})^\top - \mb V^{(r)} (\mb V^{(r)})^\top} [\mb I, -\im \mb I]^\top }{} \\
& \le 2 \norm{\mb V^{(r+1)} (\mb V^{(r+1)})^\top - \mb V^{(r)} (\mb V^{(r)})^\top}{}\\
& \le 2\sqrt{2} \norm{\mb V^{(r+1)} (\mb V^{(r+1)})^\top - \mb V^{(r)} (\mb V^{(r)})^\top}{R}, 
\end{align*}
where from the second to the third line we translate the complex operator norm to the real operator norm. Similarly, we also get 
\begin{align*}
\norm{\mb U^{(r+1)} (\mb U^{(r+1)})^\top - \mb U^{(r)} (\mb U^{(r)})^\top}{} \le 2\sqrt{2} \norm{\mb V^{(r+1)} (\mb V^{(r+1)})^\top - \mb V^{(r)} (\mb V^{(r)})^\top}{R}. 
\end{align*}
Since $\im \mb z^{(r)}$ is the normal vector of the space generated by $\mb U^{(r)}$, $[-\mb z^{(r)}_{\Im}; \mb z^{(r)}_{\Re}]$ is the corresponding normal vector of $\mb V^{(r)}$. By Lemma~\ref{lem:sp_angle_norm}, the largest principal angle $\theta_1$ between the subspaces designated by $\mb V^{(r+1)}$ and $\mb V^{(r)}$ are the angle between their normal vectors $\mb a \doteq [-\mb z^{(r)}_{\Im}; \mb z^{(r)}_{\Re}]$ and $\mb b \doteq [-\mb z^{(r+1)}_{\Im}; \mb z^{(r+1)}_{\Re}]$. Here we have decomposed $\mb z^{(r+1)}$ and $\mb z^{(r)}$ into real and imaginary parts. Similarly we define $\mb c \doteq [-(\mb \delta_\star)_{\Im}; (\mb \delta_\star)_{\Re}]$. By the law of cosines, 
\begin{align*}
\cos \theta_1 = \frac{\norm{\mb a}{}^2 + \norm{\mb b}{}^2 - \norm{\mb c}{}^2}{2\norm{\mb a}{} \norm{\mb b}{}} \ge 1 - \frac{\norm{\mb c}{}^2}{2\norm{\mb a}{} \norm{\mb b}{}} = 1 - \frac{\norm{\mb \xi_\star}{}^2}{2\norm{\mb z^{(r)}}{} \norm{\mb z^{(r+1)}}{}}. 
\end{align*}
Since $\norm{\mb z^{(r)}}{} \ge \min_{\mb z \in \mc R_3} \norm{\mb z}{} \ge (1-1/\sqrt{7})\norm{\mb x}{} \ge 3\norm{\mb x}{}/5$, and $\norm{\mb z^{(r+1)}}{} \ge \norm{\mb z^{(r)}}{} - \Delta \ge \norm{\mb x}{}/2$ provided
\begin{align*}
\Delta \le \norm{\mb x}{}/10, 
\end{align*} 
we obtain that 
\begin{align*}
\cos \theta_1 \ge 1 - \frac{5}{3 \norm{\mb x}{}^2}\norm{\mb \xi_\star}{}^2. 
\end{align*}
Thus, by Lemma~\ref{lem:sp_angle_norm} again, 
\begin{multline}
\norm{\mb V^{(r+1)} (\mb V^{(r+1)})^\top - \mb V^{(r)} (\mb V^{(r)})^\top}{R} = \sqrt{1-\cos^2 \theta_1} \\ \le \sqrt{\frac{10}{3\norm{\mb x}{}^2} \norm{\mb \delta_\star}{}^2 + \frac{25}{9\norm{\mb x}{}^4} \norm{\mb \xi_\star}{}^4} \le \frac{2}{\norm{\mb x}{}} \norm{\mb \xi_\star}{}, 
\end{multline}
where we used the assumption $\Delta \le \norm{\mb x}{}/10$ again to obtain the last inequality. 

Collecting the above results, we obtain 
\begin{align}
\norm{\nabla f(\mb z^{(r+1)})}{} \le \paren{\frac{1}{2}L_h + \frac{16}{\norm{\mb x}{}} M_H }\norm{\mb \xi_\star}{}^2. 
\end{align}
Invoking the optimality condition again, we obtain 
\begin{align}
 \norm{\mb \xi_\star}{}^2 = \norm{(\mb H^{(r)})^{-1} \mb g^{(r)}}{}^2 \le \frac{1}{m_H^2} \norm{\mb g^{(r)}}{}^2 = \frac{2}{m_H^2} \norm{\nabla f(\mb z^{(r)})}{}^2. 
\end{align}
Here $(\mb H^{(r)})^{-1}$ is well-defined because Lemma \ref{lem:hessian-lower-upper} shows that $\norm{\mb H^{(r)}}{}\geq m_H$ for all $\mb z^{(r)} \in \mc R_3'$. Combining the last two estimates, we complete the proof. 
\end{proof}

\subsection{Proof of Proposition~\ref{prop:unconstraint-steps}} \label{pf:prop:unconstraint-steps}
\begin{proof}
Throughout the proof, we write $\mb g^{(r)}$, $\mb H^{(r)}$ and $\mb U^{(r)}$ short for $\mb g(\mb z^{(r)})$, $\mb H(\mb z^{(r)})$ and $\mb U(\mb z^{(r)})$, respectively.

We first show $\mb z^{(r+1)}$ stays in $\mc R_3'$. From proof of Proposition~\ref{prop:func-value-decrease-R-3}, we know that for all $\mb z \in \mc R_3$, the following estimate holds: 
\begin{align*}
\norm{\nabla f(\mb z)}{} \ge \frac{1}{4\sqrt{2}} \norm{\mb x}{}^2 \norm{\mb h(\mb z)}{}. 
\end{align*}
From Proposition~\ref{prop:grad-quad-conv}, we know that 
\begin{align*}
\|\nabla f(\mb z^{(r+1)})\| \le \frac{1}{m_H^2} \left(L_h + \frac{32}{\norm{\mb x}{}} M_H \right) \| \nabla f(\mb z^{(r)})\|^2
\end{align*}
provided $\Delta \le \norm{\mb x}{}/10$. Moreover,
\begin{align*}
\| \nabla f(\mb z^{(r)})\|^2 = \frac{1}{2} \norm{\mb g^{(r)}}{}^2 \le M_H^2 \norm{(\mb H^{(r)})^{-1} \mb g^{(r)}}{}^2 \le M_H^2 \Delta^2,  
\end{align*}
where the last inequality followed because step $r$ is unconstrained. Combining the above estimates, we obtain that 
\begin{align*}
\|\nabla f(\mb z^{(r+1)})\| \le \frac{1}{m_H^2} \left(L_h + \frac{32}{\norm{\mb x}{}} M_H \right) M_H^2 \Delta^2. 
\end{align*}
Thus, 
\begin{align*}
\norm{\mb h(\mb z^{(r+1)})}{} \le \frac{4\sqrt{2}}{\norm{\mb x}{}^2} \|\nabla f(\mb z^{(r+1)})\| \le \frac{4\sqrt{2}}{m_H^2 \norm{\mb x}{}^2} \paren{L_h + \frac{32}{\norm{\mb x}{}} M_H} M_H^2 \Delta^2. 
\end{align*}
So, provided 
\begin{align*}
\frac{4\sqrt{2}}{m_H^2 \norm{\mb x}{}^2} \left(L_h + \frac{32}{\norm{\mb x}{}} M_H\right) M_H^2 \Delta^2 \le \frac{1}{10L_h} \norm{\mb x}{}^2, 
\end{align*}
$\mb z^{(r+1)}$ stays in $\mc R_3'$. 

Next we show the next step will also be an unconstrained step when $\Delta$ is sufficiently small. We have 
\begin{align*}
	& \| (\mb H^{(r+1)})^{-1} \mb g^{(r+1)}\| \\
	\le\; & \frac{1}{m_H} \|\mb g^{(r+1)}\| = \frac{\sqrt{2}}{m_H} \|\nabla f(\mb z^{(r+1)})\| \\
	\le\; & \frac{\sqrt{2}}{m_H^3} \paren{L_h + \frac{32}{\norm{\mb x}{}} M_H)}\|\nabla f(\mb z^{(r)})\|^2 = \frac{1}{\sqrt{2}m_H^3} \left(L_h + \frac{32}{\norm{\mb x}{}} M_H\right)\| \mb g^{(r)}\|^2 \\
	\le\; & \frac{M_H^2}{\sqrt{2}m_H^3} \left(L_h + \frac{32}{\norm{\mb x}{}} M_H \right) \|(\mb H^{(r)})^{-1} \mb g^{(r)}\|^2 \le \frac{M_H^2}{\sqrt{2}m_H^3} \left(L_h + \frac{32}{\norm{\mb x}{}} M_H \right) \Delta^2, 
\end{align*}
where we again applied results of Proposition~\ref{prop:grad-quad-conv} to obtain the third line, and applied the optimality condition to obtain the fourth line. Thus, whenever 
\begin{align*}
\frac{M_H^2}{\sqrt{2}m_H^3} \left(L_h + \frac{32}{\norm{\mb x}{}} M_H \right) \Delta < 1, 
\end{align*}
the transformed trust-region subproblem has its minimizer $\mb \xi^{(r+1)}$ with $\|\mb \xi^{(r+1)}\| < \Delta$. This implies the minimizer $\mb \delta^{(r+1)}$ to the original trust-region subproblem satisfies $\mb \delta^{(r+1)} < \Delta$, as $\|\mb \delta^{r+1}\| = \|\mb \xi^{(r+1)}\|$. Thus, under the above condition the $(r+1)$-th step is also unconstrained. 

Repeating the above arguments for all future steps implies that all future steps will be constrained within $\mc R_3'$. 

We next provide an explicit estimate for the rate of convergence in terms of distance of the iterate to the target set $\mc X$. Again by Proposition~\ref{prop:grad-quad-conv}, 
\begin{align*}
\|\nabla f(\mb z^{(r + r')})\| 
& \le m_H^2 \left(L_h + \frac{32}{\norm{\mb x}{}} M_H \right)^{-1} \paren{\frac{1}{m_H^2} \left(L_h + \frac{32}{\norm{\mb x}{}} M_H \right) \norm{\nabla f(\mb z^{(r)})}{}}^{2^{r'}} \\
& \le m_H^2 \left(L_h + \frac{32}{\norm{\mb x}{}} M_H\right)^{-1} \paren{\frac{1}{\sqrt{2}m_H^2} \left(L_h + \frac{32}{\norm{\mb x}{}} M_H\right) \norm{\mb g^{(r)}}{}}^{2^{r'}} \\
& \le m_H^2 \left(L_h + \frac{32}{\norm{\mb x}{}} M_H \right)^{-1} \paren{\frac{M_H}{\sqrt{2}m_H^2} \left(L_h + \frac{32}{\norm{\mb x}{}} M_H \right) \Delta}^{2^{r'}}. 
\end{align*}
Thus, provided 
\begin{align*}
\frac{M_H}{\sqrt{2}m_H^2} \left(L_h + \frac{32}{\norm{\mb x}{}} M_H\right) \Delta \le \frac{1}{2}, 
\end{align*}
we have 
\begin{align*}
\norm{\mb h(\mb z^{(r + r')})}{} \le \frac{4\sqrt{2}}{\norm{\mb x}{}^2} \|\nabla f(\mb z^{(r + r')})\| \le \frac{4\sqrt{2} m_H^2}{\norm{\mb x}{}^2} \left(L_h + \frac{32}{\norm{\mb x}{}} M_H\right)^{-1} 2^{-2^{r'}}, 
\end{align*}
as claimed. 
\end{proof}

%% file: sec/app_tools.tex
\section{Basic Tools and Results}

\begin{lemma}[Even Moments of Complex Gaussian]
For $a \sim \mc {CN}(1)$, it holds that 
\begin{align*}
\expect{\abs{a}^{2p}} = p! \quad \forall\; p \in \N. 
\end{align*}
\end{lemma}
\begin{proof}
Write $a = x + \im y$, then $x, y \sim_{i.i.d.} \mc N(0, 1/2)$. Thus, 
\begin{align*}
\expect{\abs{a}^{2p}} = \bb E_{x, y} \brac{\paren{x^2 + y^2}^p} = \frac{1}{2^p}\bb E_{z \sim \chi^2(2)} \brac{z^p} = \frac{1}{2^p} 2^p p! = p!,   
\end{align*}
as claimed. 
\end{proof}

\begin{lemma}[Integral Form of Taylor's Theorem]\label{lem:Taylor-integral-form}
Consider any continuous function $f(\mb z): \Cp^n \mapsto \R$ with continuous first- and second-order Wirtinger derivatives. For any $\mb \delta \in \bb C^n$ and scalar $t\in \R$, we have
	\begin{align*}
		f(\mb z+t\mb \delta ) &= f(\mb z ) + t \int_0^1 \begin{bmatrix}	\mb \delta \\ \ol{\mb \delta}\end{bmatrix}^* \nabla f(\mb z+ s t \mb \delta )  \; ds, \\
		f(\mb z+t\mb \delta ) &= f(\mb z ) + t \begin{bmatrix}\mb \delta \\ \ol{\mb \delta}\end{bmatrix}^* \nabla f(\mb z ) + t^2 \int_0^1 (1-s)\begin{bmatrix}\mb \delta \\ \ol{\mb \delta}\end{bmatrix}^* \nabla^2 f(\mb z+ st \mb \delta ) \begin{bmatrix}\mb \delta \\ \ol{\mb \delta}\end{bmatrix}\; ds.
	\end{align*}
\end{lemma}
\begin{proof}
Since $f$ is continuous differentiable, by the fundamental theorem of calculus, 
	\begin{align*}
		f(\mb z + t\mb \delta ) = f(\mb z ) +\int_0^t \begin{bmatrix}\mb \delta \\ \ol{\mb \delta}\end{bmatrix}^* \nabla f(\mb z + \tau \mb \delta ) \; d\tau. 
	\end{align*}
	Moreover, by integral by part, we obtain
	\begin{align*}
		f(\mb z+t \mb \delta ) &= f(\mb z) + \left.\brac{(\tau-t) \begin{bmatrix}\mb \delta \\ \ol{\mb \delta}\end{bmatrix} ^* \nabla f(\mb z + \tau \mb \delta )   }\right|_{0}^t - \int_{0}^t (\tau-t) \; d \brac{\begin{bmatrix}\mb \delta \\ \ol{\mb \delta}\end{bmatrix}^* \nabla f(\mb z + \tau \mb \delta )} \nonumber \\
		& = f(\mb x) + t  \begin{bmatrix}\mb \delta \\ \ol{\mb \delta}\end{bmatrix}^* \nabla f(\mb z) +\int_{0}^t (t-\tau ) \begin{bmatrix}\mb \delta \\ \ol{\mb \delta}\end{bmatrix}^* \nabla^2 f(\mb z + \tau \mb \delta ) \begin{bmatrix}\mb \delta \\ \ol{\mb \delta}\end{bmatrix}  d\tau. 
	\end{align*}
Change of variable $\tau = st (0 \le s \le 1)$ gives the claimed result. 
\end{proof}

\begin{lemma}[Error of Quadratic Approximation] \label{lem:Taylor-approx}
Consider any continuous function $f(\mb z): \Cp^n \mapsto \R$ with continuous first- and second-order Wirtinger derivatives. Suppose its Hessian $\nabla^2 f(\mb z)$ is $L_h$-Lipschitz. Then the second-order approximation 
	\begin{align*}
		\wh{f}(\mb \delta; \mb z) = f(\mb z) + \begin{bmatrix}\mb \delta \\ \ol{\mb \delta}	\end{bmatrix}^* \nabla f(\mb z) + \frac{1}{2}\begin{bmatrix}\mb \delta \\ \ol{\mb \delta}	\end{bmatrix}^* \nabla^2 f(\mb z) \begin{bmatrix}\mb \delta \\ \ol{\mb \delta}	\end{bmatrix}
	\end{align*}
	around each point $\mb z$ obeys 
	\begin{align*}
		\abs{f(\mb z+ \mb \delta ) - \wh{f}(\mb \delta; \mb z)} \leq \frac{1}{3} L_h \norm{\mb \delta }{}^3.
	\end{align*}	
\end{lemma}

\begin{proof}
By integral form of Taylor's theorem in Lemma~\ref{lem:Taylor-integral-form}, 
	\begin{align*}
		\abs{f(\mb z+ \mb \delta ) - \wh{f}(\mb \delta; \mb z )}
		&= \;\abs{ \int_0^1 (1-\tau)  \begin{bmatrix}\mb \delta \\ \ol{\mb \delta}	\end{bmatrix}^*  \brac{\nabla^2 f(\mb x+ \tau\mb \delta )- \nabla^2 f(\mb x) }\begin{bmatrix}\mb \delta \\ \ol{\mb \delta}	\end{bmatrix}  \; d\tau } \\
		&\leq\; 2\norm{\mb \delta}{}^2 \int_0^1 (1-\tau)\norm{\nabla^2 f(\mb x+ \tau\mb \delta )- \nabla^2 f(\mb x) }{} \;d\tau \\
		&\leq \; 2L_h\norm{\mb \delta}{}^3 \int_0^1 (1-\tau) \tau \;d\tau = \frac{L_h}{3}\norm{\mb \delta}{}^3, 
	\end{align*} 
	as desired.
\end{proof}


\begin{lemma}[Spectrum of Complex Gaussian Matrices] \label{lem:spec_cn}
Let $\mb X$ be an $n_1 \times n_2$ ($n_1 > n_2$) matrices with i.i.d. $\mc{CN}$ entries. Then, 
\begin{align*}
\sqrt{n_1} - \sqrt{n_2} \le  \expect{\sigma_{\min}(\mb X)} \le  \expect{\sigma_{\max}(\mb X)} \le  \sqrt{n_1} + \sqrt{n_2}. 
\end{align*}
Moreover, for each $t \ge 0$, it holds with probability at least $1 - 2\exp\paren{-t^2}$ that 
\begin{align*}
\sqrt{n_1} - \sqrt{n_2} - t \le \sigma_{\min}(\mb X) \le \sigma_{\max}(\mb X) \le \sqrt{n_1} + \sqrt{n_2} + t. 
\end{align*}
\end{lemma}

\begin{lemma}[Hoeffding-type Inequality, Proposition 5.10 of~\cite{vershynin2012introduction}]\label{lem:hoeffding_type}
	Let $X_1,\cdots, X_N$ be independent centered sub-Gaussian random variables, and let $K = \max_i \norm{X_i}{\psi_2} $, where the sub-Gaussian norm
	\begin{align}
		\norm{ X_i }{\psi_2 } \doteq \sup_{p\geq 1} p^{-1/2} \paren{ \bb E\brac{\abs{X}^p } }^{1/p}.
	\end{align}
	Then for every $\mb b = \brac{b_1; \cdots; b_N } \in \bb C^N$ and every $t\geq 0$, we have
	\begin{align}
		\bb P\paren{ \abs{ \sum_{k=1}^N b_k X_k } \geq t } \leq e \exp\paren{ - \frac{ct^2}{K^2 \norm{\mb b}{2}^2  } }. 
	\end{align}
Here $c$ is a universal constant. 
\end{lemma}

\begin{lemma}[Bernstein-type Inequality, Proposition 5.17 of~\cite{vershynin2012introduction}]\label{lem:bernstein_type}
	Let $X_1,\cdots, X_N$ be independent centered sub-exponential random variables, and let $K = \max_i \norm{X_i}{\psi_1} $, where the sub-exponential norm
	\begin{align}
		\norm{ X_i }{\psi_1} \doteq \sup_{p\geq 1} p^{-1} \paren{ \bb E\brac{\abs{X}^p } }^{1/p}.
	\end{align}
	Then for every $\mb b = \brac{b_1; \cdots; b_N } \in \bb C^N$ and every $t\geq 0$, we have
	\begin{align}
		\bb P\paren{ \abs{ \sum_{k=1}^N b_k X_k } \geq t } \leq 2 \exp\paren{ -c \min\paren{\frac{t^2}{K^2\norm{\mb b}{2}^2}, \frac{t}{K\norm{\mb b}{\infty}}} }. 
	\end{align}
Here $c$ is a universal constant. 
\end{lemma}

\begin{lemma}[Subgaussian Lower Tail for Nonnegative RV's, Problem 2.9 of~\cite{boucheron2013concentration}] \label{lem:subgauss_nonneg}
Let $X_1$, $\dots$, $X_N$ be i.i.d. copies of the nonnegative random variable $X$ with finite second moment. Then it holds that 
\begin{align*}
\prob{\frac{1}{N} \sum_{i=1}^N \paren{X_i - \expect{X_i}} < -t} \le \exp\paren{-\frac{Nt^2}{2\sigma^2}}
\end{align*}
for any $t > 0$, where $\sigma^2 = \expect{X^2}$. 
\end{lemma}
\begin{proof}
For any $\lambda > 0$, we have 
\begin{align*}
\log \expect{\e^{-\lambda(X - \expect{X})}} = \lambda \expect{X} +\log \expect{\e^{-\lambda X}} \le \lambda \expect{X} + \expect{\e^{-\lambda X}} - 1, 
\end{align*}
where the last inequality holds thanks to $\log u \le u -1$ for all $u > 0$. Moreover, using the fact $\e^u \le 1 + u + u^2/2$ for all $u \le 0$, we obtain 
\begin{align*}
\log \expect{\e^{-\lambda(X - \expect{X})}} \le \frac{1}{2} \lambda^2 \expect{X^2} \Longleftrightarrow \expect{\e^{-\lambda(X - \expect{X})}}  \le \exp\paren{\frac{1}{2} \lambda^2 \expect{X^2}}. 
\end{align*} 
Thus, by the usual exponential transform trick, we obtain that for any $t > 0$, 
\begin{align*}
\prob{\sum_{i=1}^N (X_i - \expect{X_i}) < -t} \le \exp\paren{-\lambda t + N\lambda^2 \expect{X^2}/2}. 
\end{align*}
Taking $\lambda = t/(N\sigma^2)$ and making change of variable for $t$ give the claimed result. 
\end{proof}

\begin{lemma}[Moment-Control Bernstein's Inequality for Random Variables] \label{lem:mc_bernstein_scalar}
Let $X_1, \dots, X_p$ be i.i.d. copies of a real-valued random variable $X$ Suppose that there exist some positive number $R$ and $\sigma_X^2$ such that
\begin{align*}
\expect{X^2} \le \sigma_X^2, \quad \text{and} \quad 
\expect{\abs{X}^m} \leq \frac{m!}{2} \sigma_X^2  R^{m-2}, \; \; \text{for all integers $m \ge 3$}.
\end{align*} 
Let $S \doteq \frac{1}{p}\sum_{k=1}^p X_k$, then for ... , it holds  that 
\begin{align*}
\prob{\abs{S - \expect{S}} \ge t} \leq 2\exp\left(-\frac{pt^2}{2\sigma_X^2 + 2Rt}\right).   
\end{align*}
\end{lemma}

\begin{lemma}[Angles Between Two Subspaces]  \label{lem:sp_angle_norm}
Consider two linear subspaces $\mc U$, $\mc V$ of dimension $k$ in $\R^n$ ($k \in [n]$) spanned by orthonormal bases $\mb U$ and $\mb V$, respectively. Suppose $\pi/2 \ge \theta_1 \ge \theta_2 \dots \ge \theta_k \ge 0$ are the principal angles between $\mc U$ and $\mc V$. Then it holds that \\
i) $\min_{\mb Q \in O_k} \norm{\mb U - \mb V \mb Q}{} \le \sqrt{2-2\cos \theta_1}$; \\
ii) $\sin \theta_1 = \norm{\mb U\mb U^* - \mb V\mb V^*}{}$;\\
iii) Let $\mc U^\perp$ and $\mc V^\perp$ be the orthogonal complement of $\mc U$ and $\mc V$, respectively. Then $\theta_1(\mc U, \mc V) = \theta_1(\mc U^\perp, \mc V^\perp)$. 
\end{lemma}
\begin{proof}
Proof to i) is similar to that of II. Theorem 4.11 in~\cite{stewart1990matrix}. For $2k \le n$, w.l.o.g., we can assume $\mb U$ and $\mb V$ are the canonical bases for $\mc U$ and $\mc V$, respectively. Then 
\begin{align*}
\min_{\mb Q \in O_k} \norm{
\begin{bmatrix}
\mb I - \mb \Gamma \mb Q \\
- \mb \Sigma \mb Q \\
\mb 0
\end{bmatrix}
}{} \le 
\norm{
\begin{bmatrix}
\mb I - \mb \Gamma  \\
- \mb \Sigma \\
\mb 0
\end{bmatrix}
}{} \le \norm{
\begin{bmatrix}
\mb I - \mb \Gamma  \\
- \mb \Sigma
\end{bmatrix}
}{}. 
\end{align*}
Now by definition 
\begin{align*}
\norm{
\begin{bmatrix}
\mb I - \mb \Gamma  \\
- \mb \Sigma
\end{bmatrix}
}{}^2 
& = \max_{\norm{\mb x}{} = 1} \norm{\begin{bmatrix}
\mb I - \mb \Gamma  \\
- \mb \Sigma
\end{bmatrix} \mb x}{}^2 = \max_{\norm{\mb x}{} = 1} \sum_{i=1}^k (1 - \cos  \theta_i)^2 x_i^2 + \sin^2\theta_i x_i^2 \\
& = \max_{\norm{\mb x}{} = 1} \sum_{i=1}^k (2-2\cos \theta_i) x_i^2 \le 2- 2\cos \theta_1. 
\end{align*}
Note that the upper bound is achieved by taking $\mb x = \mb e_1$. When $2k > n$, by the results from CS decomposition (see, e.g., I Theorem 5.2 of~\cite{stewart1990matrix}). 
\begin{align*}
\min_{\mb Q \in O_k} \norm{
\begin{bmatrix}
\mb I & \mb 0 \\
\mb 0 & \mb I \\
\mb 0 & \mb 0
\end{bmatrix}
 - 
\begin{bmatrix}
\mb \Gamma & \mb 0 \\
\mb 0 & \mb I \\
\mb \Sigma & \mb 0
\end{bmatrix} 
}{} \le \norm{
\begin{bmatrix}
\mb I - \mb \Gamma  \\
- \mb \Sigma
\end{bmatrix}
}{}, 
\end{align*}
and the same argument then carries through. To prove ii), note the fact that $\sin \theta_1 = \norm{\mb U \mb U^* - \mb V \mb V^*}{}$ (see, e.g., Theorem 4.5 and Corollary 4.6 of~\cite{stewart1990matrix}). Obviously one also has 
\begin{align*}
\sin \theta_1 = \norm{\mb U \mb U^* - \mb V \mb V^*}{} = \norm{(\mb I - \mb U \mb U^*) - (\mb I - \mb V \mb V^*)}{}, 
\end{align*}
while $\mb I - \mb U \mb U^*$ and $\mb I - \mb V \mb V^*$ are projectors onto $\mc U^\perp$ and $\mc V^\perp$, respectively. This completes the proof. 
\end{proof}